\documentclass[11pt,a4paper]{article}

\RequirePackage[OT1]{fontenc}
\RequirePackage{amsthm,amsmath,amsfonts,amssymb}
\usepackage{graphicx}
\usepackage[textwidth=8em,textsize=small]{todonotes}
\usepackage{natbib}
\usepackage{algorithm}
\usepackage{hyperref}
\usepackage{array,multirow}

\newtheorem{proposition}{Proposition}[section]
\newtheorem{theorem}{Theorem}[section]
\newtheorem{lemma}{Lemma}[section]
\newtheorem{condition}{Condition}[section]
\newtheorem{remark}{Remark}[section]
\newcommand{\I}{\mathbb{I}}
\newcommand{\R}{\mathbb{R}}

\newcommand{\comment}[1]{}
\newcommand{\mom}[1]{\kappa_{#1}}
\newcommand{\at}[1]{\Big|_{#1}}

\usepackage[margin = 1in]{geometry}
\usepackage{authblk}

\begin{document}
\title{The Barker proposal: combining robustness and efficiency in gradient-based MCMC}
\author[1]{Samuel Livingstone\footnote{samuel.livingstone@ucl.ac.uk}}
\author[2]{Giacomo Zanella\footnote{giacomo.zanella@unibocconi.it}}

\affil[1]{Department of Statistical Science, University College London.}
\affil[2]{Department of Decision Sciences, BIDSA and IGIER, Bocconi University.}
\maketitle

\begin{abstract}
There is a tension between robustness and efficiency when designing Markov chain Monte Carlo (MCMC) sampling algorithms. Here we focus on robustness with respect to tuning parameters, showing that more sophisticated algorithms tend to be more sensitive to the choice of step-size parameter and less robust to  heterogeneity of the distribution of interest. We characterise this phenomenon by studying the behaviour of spectral gaps as an increasingly poor step-size is chosen for the algorithm. Motivated by these considerations, we propose a novel and simple gradient-based MCMC algorithm, inspired by the classical Barker accept-reject rule, with improved robustness properties. Extensive theoretical results, dealing with robustness to tuning, geometric ergodicity and scaling with dimension, suggest that the novel scheme combines the robustness of simple schemes with the efficiency of gradient-based ones. We show numerically that this type of robustness is particularly beneficial in the context of adaptive MCMC, giving examples where our proposed scheme significantly outperforms state-of-the-art alternatives.
\end{abstract}

\section{Introduction}

The need to compute high-dimensional integrals is ubiquitous in modern statistical inference and beyond (e.g. \cite{brooks2011handbook, krauth2006statistical, stuart2010inverse}). Markov chain Monte Carlo (MCMC) is a popular solution, in which the central idea is to construct a Markov chain with a certain limiting distribution and use ergodic averages to approximate expectations of interest.
In the celebrated Metropolis--Hastings algorithm, the Markov chain transition is constructed using a combination of a `candidate' kernel, to suggest a possible move at each iteration, together with an accept-reject mechanism \citep{metropolis1953equation,hastings1970monte}. 
Many different flavours of Metropolis--Hastings exist, with the most common difference being in the construction of the candidate kernel. In the Random walk Metropolis, proposed moves are generated using a symmetric distribution centred at the current point.  Two more sophisticated methods are the Metropolis-adjusted Langevin algorithm \citep{roberts1996exponential} and Hamiltonian/hybrid Monte Carlo \citep{duane1987hybrid,neal2011mcmc}.  Both use gradient information about the distribution of interest (the \textit{target}) to inform proposals.  Gradient-based methods are widely considered to be state-of-the-art in MCMC, and much current work has been dedicated to their study and implementation (e.g. \cite{beskos2013optimal,durmus2017nonasymptotic,dalalyan2017theoretical}).

Several measures of performance have been developed to help choose a suitable candidate kernel for a given task.  
One of these is high-dimensional scaling arguments, which compare how the efficiency of the method decays with $d$, the dimension of the state space.  For the random walk algorithm this decay is of the order $d^{-1}$ \citep{roberts1997weak}, while for the Langevin algorithm the same figure is $d^{-1/3}$ \citep{roberts1998optimal} and for Hamiltonian Monte Carlo it is $d^{-1/4}$ \citep{beskos2013optimal}.  Another measure is to find general conditions under which a kernel will produce a geometrically ergodic Markov chain.  For the random walk algorithm this essentially occurs when the tails of the posterior decay at a faster than exponential rate and are suitably regular (more precise conditions are given in \citep{jarner2000geometric}).  The same is broadly true of the Langevin and Hamiltonian schemes \citep{roberts1996exponential,livingstone2016geometric, durmus2017convergence}, but here there is an additional restriction that the tails should not decay too quickly.  This limitation is caused by the way in which gradients are used to construct the candidate kernel, which can result in the algorithm generating unreasonable proposals that are nearly always rejected in certain regions \citep{roberts1996exponential, livingstone2016geometric}.

There is clearly some tension between the different results presented above.  According to the scaling arguments gradient information is preferable. The ergodicity results, however, imply that gradient-based schemes are typically less \emph{robust} than others, in the sense that there is a smaller class of limiting distributions for which the output will be a geometrically ergodic Markov chain.  It is natural to wonder whether it is possible to incorporate gradient information in such a way that this measure of robustness is not compromised. Simple approaches to modifying the Langevin algorithm for this purpose have been suggested (based on the idea of truncating gradients, e.g.  \cite{roberts1996exponential,atchade2006adaptive}), but these typically compromise the favourable scaling of the original method. 
In addition to this, it is often remarked that gradient-based methods can be difficult to tune.
Algorithm performance is often highly sensitive to the choice of scale within the proposal \citep[Fig.15]{neal2003slice}, and if this is chosen to be too large in certain directions then performance can degrade rapidly.  
Because of this, practitioners must spend a long time adjusting the tuning parameters to ensure that the algorithm is running well, or develop sophisticated adaptation schemes for this purpose (e.g. \cite{hoffman2014no}), which can nonetheless still
require a large number of iterations to find good tuning parameters
(see Sections \ref{sec:simulations} and \ref{sec:sim_adaptive}). 
We will refer to this issue as \emph{robustness to tuning}.

In this article we present a new gradient-based MCMC scheme, \emph{the Barker proposal}, which combines favourable high-dimensional scaling properties with favourable ergodicity and robustness to tuning properties.  
To motivate the new scheme, in Section \ref{sec:hetero} we present a direct argument showing how the spectral gaps for the random walk, Langevin and Hamiltonian algorithms behave as the tuning parameters are chosen to be increasingly unsuitable for the problem at hand.  In particular, we show that the spectral gaps for commonly used gradient-based algorithms decay to zero exponentially fast in the degree of mismatch between the scales of the proposal and target distributions, while for the random walk Metropolis (RWM) the decay is polynomial.  
In Section \ref{sec:barker} we derive the Barker proposal scheme beginning from a family of $\pi$-invariant continuous-time jump processes, and discuss its connections to the concept of `locally-balanced' proposals, introduced in \citep{zanella2019informed} for discrete state spaces. The name \emph{Barker} comes from the particular choice of `balancing function' used to uncover the scheme, which is inspired by the classical Barker accept-reject rule \citep{barker1965monte}.
In Section \ref{sec:barker_proofs} we conduct a detailed analysis of the ergodicity, scaling and robustness properties of this new method, establishing that it shares the favourable robustness to tuning of the random walk algorithm, can be geometrically ergodic in the presence of very light tails, and enjoys the $d^{-1/3}$ scaling with dimension of the Langevin scheme. The theory is then supported by an extensive simulation study in Sections \ref{sec:simulations} and \ref{sec:sim_adaptive}, including comparisons with state-of-the-art alternative sampling schemes, which highlights that this kind of robustness is particularly advantageous in the context of adaptive MCMC. 
The code to reproduce the experiments is available from the online repository at the link \url{https://github.com/gzanella/barker}.
Proofs and further numerical simulations are provided in the supplement.

\subsection{Basic setup and notation}

Throughout we work on the Borel space $(\R^d, \mathcal{B})$, with $d\geq 1$ indicating the dimension.
For $\lambda \in \R$, we write $\lambda \uparrow \infty$ and $\lambda \downarrow 0$ to emphasize the direction of convergence when this is important.
For two functions $f,g:\R \to \R$, we use the Bachmann--Landau notation $f(t) = \Theta\left(g(t)\right)$ if 
$
\liminf_{t\to\infty} f(t)/g(t) >0$ 
and $\limsup_{t \to\infty} f(t)/g(t) <\infty$.

The Markov chains we consider will be of the Metropolis--Hastings type, meaning that 
the $\pi$-invariant kernel $P$ is constructed as $P(x, dy) := \alpha(x,y)Q(x,dy) + r(x)\delta_x(dy)$, where $Q:\R^d \times \mathcal{B} \to [0,1]$ is a candidate kernel,
\begin{equation} \label{eq:arate_hastings}
\alpha(x,y) := \min \left( 1, \frac{\pi(dy)Q(y,dx)}{\pi(dx)Q(x,dy)} \right)
\end{equation}
is the \textit{acceptance rate} for a proposal $y$ given the current point $x$ (provided that the expression is well-defined, see \cite{tierney1998note} for details here), and $r(x) := 1- \int \alpha(x,y)Q(x,dy)$ is the average probability of rejection given that the current point is $x$.

\section{Robustness to tuning} \label{sec:hetero}

In this section, we seek to quantify the robustness of the random walk, Langevin and Hamiltonian schemes with respect to the mismatch between the scales of $\pi(\cdot)$ and $Q$ in a given direction.
Unlike other analyses in the MCMC literature (e.g. \cite{roberts2001optimal,beskos2018asymptotic}), we are interested in studying how MCMC algorithms perform when they are \emph{not} optimally tuned, in order to understand how crucially performance depends on such design choices (e.g.\ the choice of proposal step-size or pre-conditioning matrix).
The rationale for performing such an analysis is that achieving optimal or even close to optimal tuning can be extremely challenging in practice, especially when $\pi(\cdot)$ exhibits substantial heterogeneity.
This is typically done using past samples in the chain to compute online estimates of the average acceptance rate and the covariance of $\pi$ (or simply its diagonal terms for computational convenience), and then using those estimates to tune the proposal step-sizes in different directions \citep{andrieu2008tutorial}.
If the degree of heterogeneity is large, it can take a long time for certain directions to be well-explored, and hence for the estimated covariance to be representative and the tuning parameters to converge.

In such settings, algorithms that are more robust to tuning are not only easier to use when such tuning is done manually by the user, but can also greatly facilitate the process of learning the tuning parameters \textit{adaptively} within the algorithm.
We show in Sections \ref{sec:simulations} and \ref{sec:sim_adaptive} that if an algorithm is robust to tuning then this adaptation process can be orders of magnitude faster than in the alternative case, drastically reducing the overall computational cost for challenging targets. The intuition for this is that more robust algorithms will start performing well (i.e.\ sampling efficiently) earlier in the adaptation process (when tuning parameters are not yet optimally tuned), which in turn will speed up the exploration of the target and the learning of the tuning parameters.

\subsection{Analytical framework}\label{sec:framework}
The most general scenario we consider is a family of target densities $\pi^{(\lambda,k)}$
indexed by $\lambda >0$ and $k \in \{1,...,d\}$ defined as
\begin{equation} \label{eq:hetero_target}
\pi^{(\lambda,k)}(x):=\lambda^{-k}\pi(x_1/\lambda,\dots,x_k/\lambda,x_{k+1},\dots,x_d)\,,
\qquad
x=(x_1,\dots,x_d)\in\R^d\,,
\end{equation}
where $\pi$ is a density defined on $\R^d$ for which $\pi(x)>0$ for all $x \in \R^d$ and $\log\pi \in C^1(\R^d)$.  The set up allows modification of the scale of the first $k$ components of $\pi^{(\lambda,k)}$ through the parameter $\lambda$.  Our main results are presented for the case $k=1$, and we write $\pi^{(\lambda)} := \pi^{(\lambda,1)}$ for simplicity, before discussing extensions to the $k>1$ setting in Section \ref{sec:extensions}.
We consider targeting $\pi^{(\lambda)}$ using a  Metropolis--Hastings algorithm with fixed tuning parameters, and study performance as $\lambda$ varies.
Intuitively, we can think of $\lambda$ as a parameter quantifying the level of heterogeneity in the problem.
As a concrete example, consider a random walk Metropolis algorithm in which given the current state $x^{(t)}$ the candidate move is $y = x^{(t)} + \sigma\xi$, with $\sigma>0$ a fixed tuning parameter and $\xi \sim N(0, \I_d)$, where $\I_d$ is the $d\times d$ identity matrix. It is instructive to take $\sigma$ as the optimal choice of global scale for $\pi$, meaning when $\lambda$ is far from one then $\sigma$ is no longer a suitable choice for the first coordinate of $\pi^{(\lambda)}$.

In the context of the above, the $\lambda\downarrow 0$ regime is representative of distributions in which one component (in this case the first) has a very small scale compared to all others.  Conversely the $\lambda \uparrow \infty$ regime reflects the case in which one component has a much larger scale than its counterparts.
Studying robustness to tuning in the context of heterogeneity is particularly relevant, as highlighted above, as this is exactly the context in which tuning is more challenging.
The $\lambda\downarrow 0$ regime is particularly interesting and has been recently considered in \citet{beskos2018asymptotic}, where the authors study the behaviour of the random walk Metropolis for `ridged' densities for different values of $k$ using a diffusion limit approach. The focus in that work, however, was on the finding optimal tuning parameters for the algorithm as a function of $\lambda$, whereas the present paper is concerned with the regime in which the tuning parameters are fixed (as discussed above).

The above framework could be equivalently formulated by keeping the target distribution $\pi$ fixed and instead rescaling the first component of the candidate kernel by a factor $1/\lambda$.
This is indeed the formulation we mostly use in the proofs of our theoretical results. 
A proof of the mathematical equivalence between the two formulations can be found in the supplement.

\subsection{Measure of performance}
Our measure of performance for the various algorithms will be the spectral gap of the resulting Markov chains. Consider the space of functions
\[
L_{0,1}^2(\pi) = \{ f:\R^d \to\R ~|~ \mathbb{E}_\pi[f] = 0, \text{Var}_\pi[f] = 1  \}.
\]
Note that any function $g$ with $\mathbb{E}_\pi g^2<\infty$ can be associated with an $f \in L_{0,1}^2(\pi)$ through the map $f = (g - \mathbb{E}_\pi g)/\surd\text{Var}_\pi g$, and that if $X^{(t)} \sim \pi(\cdot)$ and $X^{(t+1)} |X^{(t)} \sim P(X^{(t)}, \cdot)$ then $\text{Corr}\{g(X^{(t)}),g(X^{(t+1)})\} = \text{Corr}\{f(X^{(t)}),f(X^{(t+1)})\}$.
The (right) spectral gap of a $\pi$-reversible Markov chain with transition kernel $P$ is 
\begin{equation} \label{eq:spectral_gap}
\text{Gap}(P) = \inf_{f \in L_{0,1}^2(\pi)} \frac{1}{2}\int \left( f(y) - f(x) \right)^2 \pi(dx)P(x,dy).
\end{equation}
The expression inside the infimum is called a \textit{Dirichlet form}, and can be thought of as the `expected squared jump distance' for the function $f$ provided the chain is stationary. This can in turn be re-written as $1- \text{Corr}\{f(X^{(t)}),f(X^{(t+1)})\}$.
Maximising the spectral gap of a reversible Markov chain can therefore be understood as minimising the \textit{worst-case} first-order auto-correlation among all possible square-integrable test functions.  

The spectral gap allows to bound the variances of ergodic averages (see Proposition 1 of \citealp{rosenthal2003}). 
Also, a direct connection between the spectral gap and mixing properties of the chain can be made if the operator $Pf(x) := \int f(y)P(x,dy)$ is positive on $L^2(\pi)$. This will always be the case if the chain is made lazy, which is the approach taken in \cite{woodard2009sufficient}, and the same adjustment can be made here if desired.

\subsection{The small $\lambda$ regime} \label{sub:h_to_inf}

In this section we assess the robustness to tuning of the random walk, Langevin and Hamiltonian schemes as $\lambda\downarrow 0$.
This corresponds to the case in which the proposal scale is chosen to be too large in the first component of $\pi^{(\lambda)}$.
The results in this section will support the idea that classical gradient-based schemes pay a very high price for any direction in which this tuning parameter is chosen to be too large, as already noted in the literature (e.g.\ \citealp[page 738]{neal2003slice}), while the random walk Metropolis is less severely affected by such issues. 

\subsubsection{Random walk Metropolis} \label{sub:RWM}

In the random walk Metropolis (RWM), given a current point $x \in \R^d$, a proposal $y$ is calculated using the equation
\begin{equation} \label{eq:rwm}
y = x + \sigma \xi,
\end{equation}
with $\sigma>0$ and $\xi \sim \mu(\cdot)$ for some centred symmetric distribution $\mu$. 
The resulting candidate kernel $Q^R$ is given by $Q^R(x,dy) = q^R(x,y)dy$ with $q^R(x,y) = \sigma^{-d}\mu((y-x)/\sigma)$, where $\mu(\xi)$ for $\xi\in\R^d$ denotes the density of $\mu$. 
Following Section \ref{sec:framework}, we consider Metropolis--Hastings algorithms with proposal $Q^R$ and target distribution $\pi^{(\lambda)}$ defined in \eqref{eq:hetero_target}, and denote the resulting transition kernels as $P_\lambda^R$.

We 
impose the following mild regularity conditions on the density $\mu(\xi)$. These are satisfied for most popular choices of $\mu$, as shown in the subsequent proposition.

\begin{condition} \label{cond:rwm}
There exists $\lambda_0>0$ such that for any $x,y \in \R^d$ and $\lambda<\lambda_0$we have $\mu\left( \delta_\lambda \right) \geq \mu( \delta )$, where $\delta = y - x $ and  
\begin{equation} \label{eqn:delta_h}
\delta_\lambda := \left( \lambda(y_1 - x_1), y_2 - x_2,...,y_d - x_d \right).
\end{equation}
In addition, $\sup_{\xi_1\in \R}\mu_1(\xi_1)<\infty$, where $\mu_1(\xi_1)=\int_{\R^{d-1}}\mu(\xi_1,\xi_2,\dots,\xi_d)d\xi_2\dots d\xi_d$ is the marginal distribution of $\xi_1$ under $\xi\sim\mu$.
\end{condition}

\begin{proposition} \label{prop:rwm_condition}
Denoting the usual $p$-norm as 
$\|x\|_p=(\sum_{i=1}^dx_i^p)^{1/p}$, Condition \ref{cond:rwm} holds in each of the below cases: \\
(i) $q^R(x,y) = (2\pi\sigma^2)^{-d/2}\exp\left( -\|x - y\|_2^2/(2\sigma^2) \right)$ (Gaussian)\\
(ii) $q^R(x,y) = 2^{-d}\exp \left( - \|x - y\|_1 \right)$ (Laplace)\\
(iii) $q^R(x,y) \propto (1 + \|y-x\|_2^2/\nu)^{-(\nu+d)/2} $ for $\nu \in\{1,2,...\}$ (Student's t)
\end{proposition}

We conclude the section with a characterization of the rate of convergence to zero of the spectral gap for the Random Walk Metropolis as $\lambda\downarrow 0$.

\begin{theorem}\label{thm:RW_lambda_to_0} 
Assume Condition \ref{cond:rwm} and $\textup{Gap}(P^R_1)>0$.
Then it holds that
$$
\textup{Gap}(P^R_\lambda) = \Theta\left(\lambda\right)\,,
\quad \hbox{as }\lambda\downarrow 0\,.
$$
\end{theorem}
Note that Theorem \ref{thm:RW_lambda_to_0} requires very few assumptions on the target $\pi$ other than $\textup{Gap}(P^R_1)>0$.
Note also that the lower bound is of the form $\text{Gap}(P_\lambda^R) \geq \lambda \text{Gap}(P_1^R)$, see proof of Theorem \ref{thm:RW_lambda_to_0} for details. 
No dependence on the dimension of the problem other than that intrinsic to $\text{Gap}(P_1^R)$ is therefore introduced.

\subsubsection{The Langevin algorithm} \label{sub:MALA}

In the Langevin algorithm (or more specifically the Metropolis-adjusted Langevin algorithm, MALA), given the current point $x \in \R^d$, a proposal $y$ is generated by setting
\begin{equation}\label{eq:MALA_proposal}
y  = x + \frac{\sigma^2}{2}\nabla \log\pi^{(\lambda)}(x) + \sigma \xi,
\end{equation}
for some $\sigma>0$ and  $\xi \sim N(0,\I_d)$.  In this case the proposal is no longer symmetric and so the full Hastings ratio \eqref{eq:arate_hastings} must be used. 
The proposal mechanism is based on the overdamped Langevin stochastic differential equation $dX_t = \nabla\log\pi^{(\lambda)}(X_t)dt + \sqrt{2}dB_t$.  
We write $Q^M_\lambda$ for the corresponding candidate distribution and $P^M_\lambda$ for the Metropolis--Hastings kernel with proposal $Q^M_\lambda$ and target $\pi^{(\lambda)}$.

We present results for the Langevin algorithm in two settings.
Initially we consider more restrictive conditions under which our upper bound on the spectral gap depends on the tail behaviour of $\pi$ in a particularly explicit manner, and then give a broader result.

\begin{condition}
\label{cond:regularity_conditions}Assume the following:

(i) $\pi$ has a density of the form $\pi(x)=\pi_{1}(x_{1})\pi_{2:n}(x_{2},...,x_{d})$, for some densities $\pi_{1}$ and $\pi_{2:n}$ on $\R$ and $\R^{d-1}$, respectively.

(ii) For some $q \in [0,1)$, it
holds that

\begin{equation}
\left|\frac{d}{d x_{1}}\log\pi_1(x_1)\right|=\Theta\left(|x_{1}|^{q}\right)\qquad\text{as }|x_{1}|\uparrow\infty.\label{eq:condition_1}
\end{equation}
\end{condition}
\begin{theorem}\label{thm:MALA_product}
If Condition \ref{cond:regularity_conditions} holds, then there is a $\gamma >0$ such that
\[
\textup{Gap}(P_\lambda^M) \leq \Theta \left( e^{-\gamma \lambda^{-(1+q)} + q\log(\lambda)} \right)
\qquad \hbox{ as }\lambda\downarrow0\,.
\]
\end{theorem}

When compared with the random walk algorithm, Theorem \ref{thm:MALA_product} shows that the Langevin scheme is much less robust to heterogeneity.
Indeed, the spectral gap decays \textit{exponentially fast} in $\lambda^{-(1+q)}$,
meaning that even small errors in the choice of step-size can have a large impact on algorithm efficiency, and so
practitioners must invest considerable effort tuning the algorithm for good performance, as shown through simulations in Sections \ref{sec:simulations} and \ref{sec:sim_adaptive}.
Theorem \ref{thm:MALA_product} also illustrates that the Langevin algorithm is more sensitive to $\lambda$ when the tails of $\pi(\cdot)$ are lighter.  This is intuitive, as in this setting gradient terms can become very large in certain regions of the state space.  

\begin{remark}
Theorem \ref{thm:MALA_product} (and Theorem \ref{thm:HMC} below) could be extended to the case $q\geq 1$ in \eqref{eq:condition_1}, however in these cases  samplers typically fail to be geometrically ergodic when $\lambda$ is small \citep{roberts1996exponential,livingstone2016geometric} meaning the spectral gap is typically 0 and the theorem becomes trivial.
\end{remark}
\begin{remark}
Condition \ref{cond:regularity_conditions} (ii) could be replaced with the simpler requirement that $|\nabla \log\pi_1(x_1)|\uparrow \infty$, with the corresponding bound $\textup{Gap}(P_\lambda^M) \leq \Theta(e^{-1/\lambda})$. 
\end{remark}

A different set of conditions, which hold much more generally, and corresponding upper bound are presented below.

\begin{condition}\label{cond:gap_MALA}
Assume the following:

(i) There is a $\gamma > 0$ such that
\begin{equation}\label{eq:gap_MALA_gradient}
\liminf_{|x_1|\to\infty} \left(\inf_{(x_2,\dots,x_d)\in\R^{d-1}}\left|\frac{\partial \log \pi(x)}{\partial x_1}\right|\|x\|_2^\gamma \right)>0\,,
\end{equation}

(ii) Given $X\sim\pi$ there is a $\beta > 0$ such that
\begin{equation}\label{eq:gap_MALA_tail}
\mathbb{P}\left( \|X\|_2>t\right)\leq \Theta\left(e^{-t^\beta}\right)
\qquad
\hbox{as }t\to\infty\,.
\end{equation}
\end{condition}

\begin{theorem}\label{thm:MALA_general}
If Condition \ref{cond:gap_MALA} holds, then 
\[
\textup{Gap}(P_\lambda^M) \leq \Theta(e^{- \lambda^{-\alpha}})
\qquad \hbox{ as }\lambda\downarrow 0\,.
\]
for some $\alpha>0$, which can be taken as $\alpha=\min\{\beta/2,\beta/\gamma,2/3\}$.
\end{theorem}

We expect Condition \ref{cond:gap_MALA} to be satisfied in many commonly encountered scenarios, with the exception of particularly heavy-tailed models.  In the exponential family class $\pi(x) \propto \exp \{ -\alpha \|x\|_2^\beta\}$, for example, Condition \ref{cond:gap_MALA} holds for any $\alpha$ and $\beta >0$ (see proof in the supplement).

\subsubsection{Hamiltonian Monte Carlo} \label{sub:HMC}

In Hamiltonian Monte Carlo (HMC) we write the current point $x \in \R^d$ as $x(0)$, and construct the proposal $y := x(L)$ for some prescribed integer $L$ using the update 
\begin{equation}
x(L) = x(0)+\sigma^{2}\left( \frac{L}{2}\nabla \log\pi^{(\lambda)}(x(0))+\sum_{j=1}^{L-1} (L-j)\nabla\log\pi^{(\lambda)}\left(x(j)\right)\right) +L\sigma \xi(0),\label{eq:transition}
\end{equation}
where each $x(j)$ is defined recursively in the same manner, and $\xi(0)\sim N(0,\I_d)$.  The transition is based on numerically solving Hamilton's equations for the Hamiltonian system $H(x,\xi) = -\log\pi^{(\lambda)}(x) + \xi^T\xi/2$ for $L\sigma$ units of time. 
The decision of whether or not the proposal is accepted is taken using the acceptance probability $\min(1, \pi^{(\lambda)}(y)/\pi^{(\lambda)}(x) \times \ e^{-\xi(L)^T\xi(L)/2+\xi(0)^T\xi(0)/2})$, where 
\[
\xi(L) = \xi(0) + \frac{\sigma}{2}\left( \nabla \log\pi^{(\lambda)}(x(0)) + \nabla\log\pi^{(\lambda)}(x(L)) \right) + \sigma \sum_{j=1}^{L-1} \nabla \log\pi^{(\lambda)}(x(j)).
\]
A more detailed description is given in \cite{neal2011mcmc}.
We write $P^H_\lambda$ for the corresponding 
Metropolis--Hastings kernel with proposal mechanism as above and target $\pi^{(\lambda)}$.
Here we present a heterogeneity result under Condition \ref{cond:regularity_conditions} of the previous subsection.

\begin{theorem}\label{thm:HMC}
If Condition \ref{cond:regularity_conditions} holds, then there is a $\gamma >0$ such that
\[
\textup{Gap}(P_\lambda^H) \leq \Theta\left(e^{-\gamma \lambda^{-(1+q)} + q\log(\lambda)}\right)
\qquad\hbox{ as }\lambda\downarrow 0\,.
\]
\end{theorem}

It is no surprise that Theorem \ref{thm:HMC} is comparable to Theorem \ref{thm:MALA_product}, since setting $L=1$ equates the Langevin and Hamiltonian methods.  

\subsection{The large $\lambda$ regime} \label{sec:h_to_zero}
In this section we briefly discuss the $\lambda\uparrow \infty$ regime, where $\sigma$ is chosen to be too small for the first component of $\pi^{(\lambda)}$, arguing that all samplers under consideration behave similarly in this regime and pay a similar price for too small tuning parameters in a given direction. 
The intuition for this is that as $\lambda\uparrow \infty$ the gradient-based proposal mechanisms discussed here all tend towards that of the random walk sampler in the first coordinate. 
For example, if we consider one-dimensional models, for any $x \in \R$ we can write $\nabla\log\pi^{(\lambda)}(x) = \lambda^{-1}\nabla\log\pi(x/\lambda)$, meaning as $\lambda \uparrow \infty$ the amount of gradient information in the proposal is reduced provided $\pi$ is suitably regular. The following result makes this intuition precise. To avoid repetitions, we state here the result for both the Langevin and the Barker proposal that we will introduce in the next section. 
\begin{proposition} \label{prop:h_to_zero}
Fix $x \in \R$ and let the density $\pi$ be such that $\nabla\log\pi$ is bounded in some neighbourhood of zero.  Then the Langevin and Barker candidate kernels $Q^M_{\lambda}$ and $Q^B_{\lambda}$, defined in \eqref{eq:MALA_proposal} and \eqref{eq:barker_prop} respectively, both satisfy
\[
\|Q^{M/B}_\lambda(x,\cdot)-Q^R(x,\cdot)\|_{TV} \leq \Theta \left( 1/\lambda \right),
\]
where $Q^R$ is the (Gaussian) random walk candidate kernel.
\end{proposition}
The same intuition applies to the Hamiltonian case provided $L$ is fixed, since each gradient term in the proposal is also $\Theta(1/\lambda)$.
While there are many well-known measures of distance between two distributions, we argue that total variation is an appropriate choice here, since it has an explicit focus on how much the two kernels overlap and is invariant under bijective transformations of the state space (including re-scaling coordinates).  

While the 
above statements provide
useful heuristic arguments, in order to obtain more rigorous results one should prove that the spectral gaps decay to 0 at the same rate as $\lambda\uparrow \infty$, which we leave to future work.
We note, however, that the conjecture that the algorithms behave similarly for large values of $\lambda$ is supported by the simulations of Section \ref{sec:sim_tuning}.

\subsection{Extensions}\label{sec:extensions} 

The lower bound of Theorem \ref{thm:RW_lambda_to_0} extends naturally to the $k>1$ setting, becoming instead $\geq \Theta(\lambda^k)$, and so the rate of decay remains polynomial in $\lambda$ for any $k$.  Analogously, we expect the corresponding upper bound for gradient-based schemes to remain exponential and become $\leq\Theta\left(e^{-k(\gamma \lambda^{-(1+q)} + q\log(\lambda))}\right)$, although the details of this are left for future work.  We explore examples of this nature through simulations in Section \ref{sec:simulations} and find empirically that the single component case is informative also of more general cases.  Further extensions in which a different $\lambda_i$ is chosen in each of the $k$ directions can also be considered, with each $\lambda_i \downarrow 0$ at a different rate. We conjecture that in this setting the $\lambda_i$ that decays most rapidly will dictate the behaviour of spectral gaps, though such an analysis is beyond the scope of the present work.

\section{Combining robustness and efficiency} \label{sec:barker}
The results of Section \ref{sec:hetero} show that the two gradient-based samplers considered there are much less robust to heterogeneity than the random walk algorithm.  In this section, we introduce a novel and simple to implement gradient-based scheme that shares the superior scaling properties of the Langevin and Hamiltonian schemes, but also retains the robustness of the random walk sampler, both in terms of geometric ergodicity and robustness to tuning.

\subsection{Locally-balanced Metropolis--Hastings}
Consider a continuous-time Markov jump process on $\R^d$ with associated generator
\begin{equation} \label{eqn:zprocess}
    \mathcal{L} f(x) = \int [f(y)-f(x)]g\left(\frac{\pi(y)q(y,x)}{\pi(x)q(x,y)}\right) Q(x,dy),
\end{equation}
for some suitable function $f:\R^d \to \R$, where $\pi(x)$ is a probability density, $Q(x,dy):=q(x,y)dy$ is a transition kernel and the \emph{balancing} function $g:(0,\infty) \to (0,\infty)$ satisfies
\begin{equation} \label{eqn:balance}
g(t) = t g\left(1/t\right).
\end{equation}
A discrete state-space version of this process with symmetric $Q$ was introduced in \cite{power2019accelerated}. The dynamics of the process are such that if the current state $X_t = x$, the next jump will be determined by a Poisson process with intensity
\begin{equation} \label{eqn:z_intensity}
    Z(x) := \int g\left(\frac{\pi(y)q(y,x)}{\pi(x)q(x,y)}\right) Q(x,dy),
\end{equation}
and the next state is drawn from the kernel
\begin{equation*}
Q^{(g)}(x,dy) := Z(x)^{-1} g\left(\frac{\pi(y)q(y,x)}{\pi(x)q(x,y)}\right) Q(x,dy).
\end{equation*}
The $L^2(\R^d)$ adjoint, or \emph{forward operator} $\mathcal{A}$ (e.g. \cite{fearnhead2018piecewise}) is given by
\begin{equation*}
\mathcal{A} h(x) = \int h(y) g\left(\frac{\pi(x)q(x,y)}{\pi(y)q(y,x)}\right) q(y,x)dy - h(x)Z(x).
\end{equation*}
Note that in the case $h(x) = \pi(x)$ using \eqref{eqn:balance} the first expression on the right-hand side can be written
\begin{equation*}
\int  g\left(\frac{\pi(y)q(y,x)}{\pi(x)q(x,y)}\right) \pi(x)q(x,y)dy = \pi(x)Z(x),
\end{equation*}
meaning $\mathcal{A} \pi = 0$, suggesting $\pi$ is invariant. It can therefore serve as a starting point for designing Markov chain Monte Carlo algorithms.  

In the `locally-balanced' framework for discrete state-space Metropolis--Hastings introduced in \cite{zanella2019informed}, candidate kernels are of the form
\begin{equation}\label{eq:balanced_proposal}
\tilde{Q}(x,dy) = \tilde{Z}(x)^{-1}g\left(\frac{\pi(y)}{\pi(x)}\right)\mu_{\sigma}(y-x)dy,
\end{equation}
meaning the \textit{embedded Markov chain} of \eqref{eqn:zprocess} with the choice $Q(x,dy):=\mu_\sigma(y-x)dy$, where $\mu_\sigma(y-x) := \sigma^{-d}\mu((y-x)/\sigma)$ for some symmetric density $\mu$. It is well-known that the invariant density of the embedded chain does not coincide with that of the process when jumps are not of constant intensity, in this case becoming proportional to $Z(x)\pi(x)$, as shown in \cite{zanella2019informed}. As a result a Metropolis--Hastings step is employed to correct for the discrepancy.  In \cite{power2019accelerated} it is suggested that as an alternative the jump process can be simulated exactly.

The challenge with employing either of these strategies on a continuous state space is that the integral \eqref{eqn:z_intensity} will typically be intractable. To overcome this issue we take two steps, and for simplicity we first describe these on $\R$ (there are two options on $\R^d$ for $d > 1$, which are discussed in Section \ref{sec:Barker_multi_d}).  The first step is to consider a first-order Taylor series expansion of $\log\pi$ within $g$ (again with a symmetric choice of $Q$), leading to the family of processes with generator
\begin{equation*}
L f(x) = \int [f(y)-f(x)]g\left(e^{\nabla\log\pi(x)(y-x)}\right) \mu_\sigma(y-x)dy.
\end{equation*}
We refer to candidate kernels in Metropolis--Hastings algorithms that are constructed using the embedded Markov chain of this new process as \textit{first-order} locally-balanced proposals, taking the form
\begin{equation}\label{eq:gradient_based}
Q^{(g)}(x,dy) = Z(x)^{-1} g\left( e^{\nabla\log\pi(x)(y-x)} \right) \mu_\sigma(y-x)dy,
\end{equation}
where $Z(x) :=  \int g(e^{\nabla\log\pi(x)(y-x)})\mu_\sigma(y-x)dy$. This second step is to note that, 
if particular choices of $g$ are made, then $Z(x)$ becomes tractable. In fact, if the balancing function $g(t) = \sqrt{t}$ and a Gaussian kernel $\mu_\sigma$ are chosen, then the result is the Langevin proposal
\[
Q^M(x,dy) \propto e^{\nabla \log\pi(x)(y-x)/2}\mu_\sigma(y-x)dy.
\]
Thus, MALA can be viewed as a particular instance of this class.  Other choices of $g$ are, however, also possible, and give rise to different gradient-based algorithms.  In the next section we explore what a sensible choice of $g$ might look like.

\begin{remark}
One can also think at \eqref{eqn:balance} as a requirement to ensure that the proposals in \eqref{eq:gradient_based} are exact (i.e.\ $\pi$-reversible) at the first order.
In particular, in the supplement it is shown that a proposal $Q^{(g)}$ defined as in \eqref{eq:gradient_based} is $\pi$-reversible with respect to log-linear density functions 
if and only if \eqref{eqn:balance} holds.
\end{remark}

\subsection{The Barker proposal on $\R$}
The requirement for the balancing function $g$ to satisfy $g(t) = t g(1/t)$ is in fact also imposed on the acceptance rate of a Metropolis--Hastings algorithm to produce a $\pi$-reversible Markov chain.  Indeed, setting $t := \pi(y)q(y,x)/(\pi(x)q(x,y))$ and assuming $\alpha(x,y) := \alpha(t)$, then the detailed balance equations can be written $\alpha(t) = t\alpha(1/t)$.  Possible choices of $g$ can therefore be found by considering suggestions for $\alpha$ in the literature.  One choice proposed in \cite{barker1965monte} is
\[
g(t) = \frac{t}{1+t}.
\]
The work of \cite{peskun1973optimum} and \cite{tierney1998note} showed that this choice of $\alpha$ is inferior to the more familiar Metropolis--Hasting rule $\alpha(t) = \min(1, t)$ in terms of asymptotic variance.  The same conclusion cannot, however, be drawn when considering the choice of balancing function $g$.

In fact, the choice $g(t) = t/(1+t)$ was shown to minimize asymptotic variances of Markov chain estimators in some discrete settings in \cite{zanella2019informed}.  In addition, as shown below, this particular choice of $g$ leads to a fully tractable candidate kernel that can be easily sampled from.

\begin{proposition} \label{prop:barker_normalising}
If $g(t) = t/(1+t)$, then the normalising constant $Z(x)$ in \eqref{eq:gradient_based} is $1/2$.
\end{proposition}

The  resulting proposal distribution is
\begin{equation}\label{eq:barker_prop}
Q^B(x,dy) = 2 
\frac{\mu_\sigma(y-x)}{1+e^{-\nabla\log\pi(x)(y-x)}}
dy.
\end{equation}
We refer to $Q^B$ as the \emph{the Barker proposal}.
A simple sampling strategy to generate $y\sim Q^B(x,\cdot)$
is given in Algorithm \ref{alg:barker_1d}.

\begin{algorithm}
\caption{Generating a Barker proposal on $\R$}
\label{alg:barker_1d}
\raggedright \textbf{Require:} the current point $x \in \R$.
\begin{enumerate}
\item Draw $z \sim \mu_\sigma(\cdot)$
\item Calculate $p(x,z) = 1/(1+e^{-z\nabla\log\pi(x)})$
\item Set $b(x,z) = 1$ with probability $p(x,z)$, and $b(x,z) = -1$ otherwise
\item Set $y = x + b(x,z) \times z$
\end{enumerate}
\textbf{Output:} the resulting proposal $y$.
\end{algorithm}

\begin{proposition} \label{prop:barker_sample}
Algorithm \ref{alg:barker_1d} produces a sample from $Q^B$ on $\R$.
\end{proposition}

Algorithm \ref{alg:barker_1d} shows that the magnitude $|y - x| = |z|$ of the proposed move does not depend on the gradient $\nabla\log\pi(x)$ here, it is instead dictated only by the choice of symmetric kernel $\mu_\sigma$.  The \emph{direction} of the proposed move is, however, informed by both the magnitude and direction of the gradient.  Examining the form of $p(x,z)$, it becomes clear that if the signs of $z$ and $\nabla\log\pi(x)$ are in agreement, then $p(x,z) >1/2$, and indeed as $z\nabla\log\pi(x) \uparrow \infty$ then $e^{-z\nabla\log\pi(x)} \downarrow 0$ and so $p(x,z) \uparrow 1$. Hence, if the indications from $\nabla \log\pi(x)$ are that $\pi(x+z) \gg \pi(x)$, then it is highly likely that $b(x,z)$ will be set to $1$ and $y = x+z$ will be the proposed move. Conversely, if $z\nabla\log\pi(x)<0$, then there is a larger than 50\% chance that the proposal will instead be $y = x - z$. As $\nabla\log\pi(x) \uparrow \infty$ the Barker proposal converges to $\mu_\sigma$ truncated on the right, and similarly to $\mu_\sigma$ truncated on the left as $\nabla\log\pi(x) \downarrow -\infty$.  See Figure \ref{fig:Barker_illustration} for an illustration.

The multiplicative term $1/(1+e^{-\nabla\log\pi(x)(y-x)})$ in \eqref{eq:barker_prop}, which incorporates the gradient information, injects skewness into the base kernel $\mu_\sigma$ (as can be clearly seen in the left-hand plot of Figure \ref{fig:Barker_illustration}).
Indeed, the resulting distribution $Q^B$ is an example of a \emph{skew-symmetric} distribution \citep[eq.(1.3)]{azzalini2013skew}.
Skew-symmetric distributions are a tractable family of (skewed) probability density functions that are obtained by multiplying a symmetric base density function with the cumulative distribution function (cdf) of a symmetric random variable.
We refer to \citet[Ch.1]{azzalini2013skew} for more details, including a more general version of Propositions \ref{prop:barker_normalising} and \ref{prop:barker_sample}.
In the context of skewed distributions the Gaussian cdf is often used, leading to the skew-normal distribution introduced in \cite{azzalini1985class}.
In our context, however, the Barker proposal (which leads to the logistic cdf $p(x,z)$ in Algorithm \ref{alg:barker_1d}) is the only skew-symmetric distribution that can be obtained from \eqref{eq:gradient_based} using a balancing function $g$ satisfying $g(t)=tg(1/t)$. See the supplement for more detail.

\begin{figure}[t]
\centering
\includegraphics[width=\linewidth]{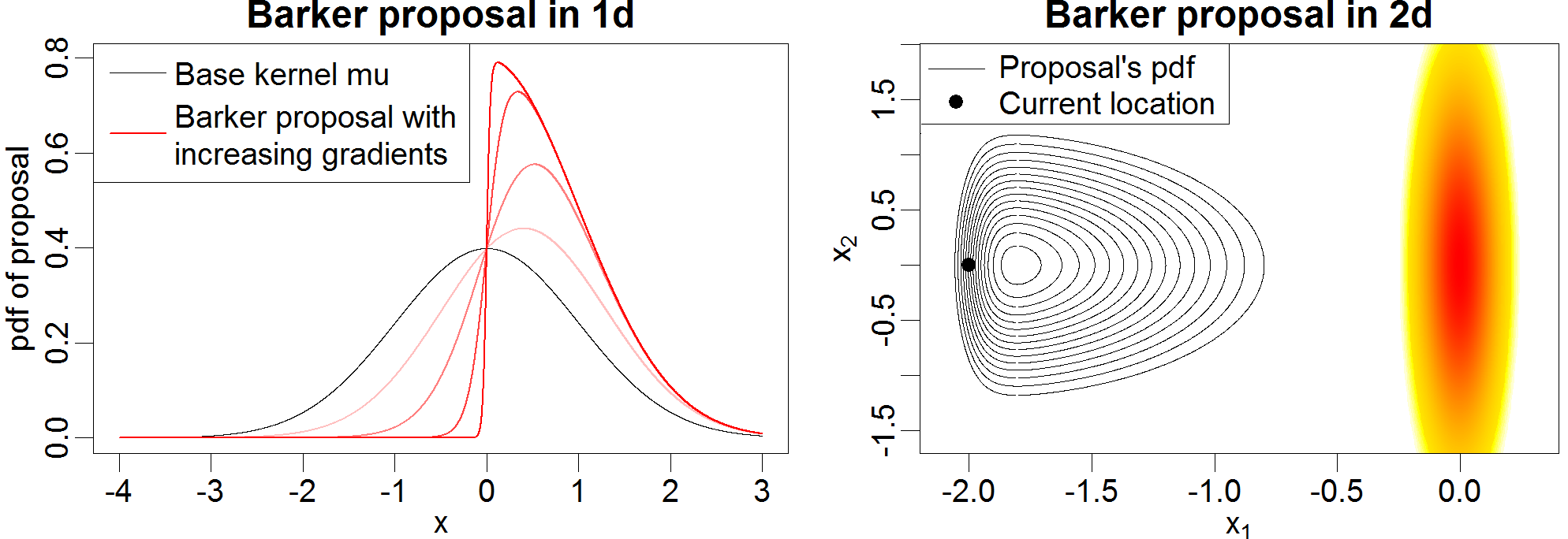}
\caption{Left: density of the Barker proposal in one dimension. Current location is $x=0$ and the four lines with increasing red intensity correspond to $\nabla\log\pi(x)$ equal to $1$, $3$, $10$ and $50$.
Right: density of the Barker proposal in two dimensions. 
Solid lines display the proposal density contours, heat colours refer to the target density, and the current location is $x=(-2,0)$.}
\label{fig:Barker_illustration}
\end{figure}

\subsection{The Barker proposal on $\R^d$}\label{sec:Barker_multi_d}
There are two natural ways to extend the Barker proposal to $\R^d$, for $d > 1$.
The first is to treat each coordinate separately, and generate the proposal $y = (y_1,...,y_d)$ by applying Algorithm \ref{alg:barker_1d} independently to each coordinate.
This corresponds to generating a $z_i$ and $b_i(x,z_i)$ for each $i \in \{1,...,d\}$, and choosing the sign of each $b_i$ using
\[
p_i(x,z_i) = \frac{1}{1+e^{-z_i\partial_i \log\pi(x)}},
\]
where $\partial_i \log\pi(x)$ denotes the partial derivative of $\log\pi(x)$ with respect to $x_i$.  Writing $Q^B_i(x,dy_i)$ to denote the resulting Barker proposal candidate kernel for the $i$th coordinate, the full candidate kernel $Q^B$ can then be written
\begin{equation} \label{eq:barker_multi}
Q^B(x,dy) = \prod_{i=1}^d Q^B_i(x,dy_i).
\end{equation}
The full Metropolis--Hastings scheme using the Barker proposal mechanism for a target distribution is given in Algorithm \ref{alg:barkerMH} (see the supplement for more details and variations of the algorithm, such as a pre-conditioned version).
Note that the computational cost of each iteration of the algorithm is essentially equivalent to that of MALA and will be typically dominated by the cost of computing the gradient and density of the target.

\begin{algorithm}
\caption{Metropolis--Hastings with the Barker proposal on $\R^d$}
\label{alg:barkerMH}
\raggedright \textbf{Require:} starting point for the chain $x^{(0)} \in \R^d$, and scale $\sigma>0$.  

Set $t=0$ and do the following:
\begin{enumerate}
\item Given $x^{(t)}=x$, draw $y_i$ using Algorithm \ref{alg:barker_1d} independently for $i \in \{1,...,d\}$
\item Set $x^{(t+1)} = y$ with probability $\alpha^B(x,y)$, where
\begin{equation}\label{eq:Bark_accept_diag}
\alpha^B(x,y) = \min\left( 1, \frac{\pi(y)}{\pi(x)}\times 
\prod_i \frac{1+e^{(x_i-y_i)\partial_i\log\pi(x)}}
{1+e^{(y_i-x_i)\partial_i\log\pi(y)}} \right).    
\end{equation}
Otherwise set $x^{(t+1)} = x$
\item If $t+1<N$, set $t \leftarrow t+1$ and return to step 1, otherwise stop.
\end{enumerate}
\textbf{Output:} the Markov chain $\{x^{(0)},\dots,x^{(N)}\}$.
\end{algorithm}

The second approach to deriving a multivariate Barker proposal consists of sampling $z \in \R^d$ from a $d$-dimensional symmetric distribution, and then choosing whether or not to flip the sign of \emph{every} coordinate at the same time, using a single global $\check{b}(x,z) \in \{-1,1\}$, to produce the global proposal $y = x + \check{b}(x,z) \times z$.  In this case the probability that $\check{b}(x,z) = 1$ will be
\begin{equation} \label{eqn:barkermod}
\check{p}(x,z) = \frac{1}{1+e^{-z^T \nabla\log\pi(x)}}.
\end{equation}
This second approach doesn't allow gradient information to feed into the proposal as effectively as in the first case.
Specifically, only the global inner product $z^T \nabla \log\pi(x)$ is considered, and the decision to alter the sign of every component of $z$ is taken based solely on this value.
In other words, once $z \sim \mu_\sigma$ has been sampled, gradient information is only used to make a single binary decision of choosing between the two possible proposals $x+z$ and $x-z$, while in the first strategy gradient information is used to choose between $2^d$ possible proposals $\{x+b \cdot z: b \in \{-1,1\}^d \}$ (where $b \cdot z := (b_1z_1,...,b_dz_d)$).
Indeed, the following proposition shows that the second  strategy cannot improve over the random walk Metropolis by more than a factor of two.

\begin{proposition} \label{prop:barkermod}
Let $\check{P}^B$ denote the modified Barker proposal on $\mathbb{R}^d$ using \eqref{eqn:barkermod}. Then
$\text{Gap}(P^R) \geq \text{Gap}(\check{P}^B)/2$.
\end{proposition}
One can also make a stronger statement than the above proposition, namely that if this strategy is employed, only a constant factor improvement over the Random Walk Metropolis can be achieved in terms of asymptotic variance, for any $L^2(\pi)$ function of interest.  Given Proposition \ref{prop:barkermod} we choose to use the first strategy described to produce Barker proposals on $\R^d$, and the multi-dimensional candidate kernel given in \eqref{eq:barker_multi}.
In the following sections we will show both theoretically and empirically that this choice does indeed have favourable robustness and efficiency properties.

\section{Robustness, scaling and ergodicity results for the Barker proposal}\label{sec:barker_proofs}
In this section we establish results concerning robustness to tuning, scaling with dimension and geometric ergodicity for the Barker proposal scheme.
As we will see, the method enjoys the superior efficiency of gradient-based algorithms in terms of scaling with dimension, but also shares the favourable robustness properties of the random walk Metropolis when considering both robustness to tuning and geometric ergodicity.

\subsection{Robustness to tuning}\label{sec:barker_heterog}
We now examine the robustness to tuning of the Barker proposal using the framework introduced in Section \ref{sec:hetero}.  We write $Q_\lambda^B$ and $P_\lambda^B$ to denote the candidate and Metropolis--Hastings kernels for the Barker proposal targeting the distribution $\pi^{(\lambda)}$ defined therein, and $P^B$ for the case $\lambda = 1$.  The following result characterizes the behaviour of the spectral gap of $P_\lambda^B$ as $\lambda \downarrow 0$. 
\begin{theorem}\label{thm:Barker_lambda_to_0} Assume Condition \ref{cond:rwm} and $\textup{Gap}(P^B)>0$.
Then it holds that
$$
\textup{Gap}(P^B_\lambda) = \Theta\left(\lambda\right)\,,
\quad \hbox{as }\lambda\downarrow 0\,.
$$
\end{theorem}
Comparing Theorem \ref{thm:Barker_lambda_to_0} with Theorems \ref{thm:RW_lambda_to_0}-\ref{thm:HMC} from Section \ref{sub:h_to_inf} we see that the Barker proposal inherits the robustness to tuning of random walk schemes and is significantly more robust than the Langevin and Hamiltonian algorithms. In the next section we establish general conditions under which $\text{Gap}(P^B)>0$.

\subsection{Geometric ergodicity}\label{sec:geom_erg_barker}
In this section we study the class of target distributions for which the Barker proposal produces a geometrically ergodic Markov chain.  We show that geometric ergodicity can be obtained even when the gradient term in the proposal grows faster than linearly, which is typically not the case for MALA and HMC.

Recall that a Markov chain is called \textit{geometrically ergodic} if 
\begin{equation}\label{eq:geom_erg_defi}
\|P^t(x,\cdot) - \pi(\cdot)\|_{TV} \leq CV(x)\rho^t,
\qquad t\geq 1\,,
\end{equation}
for some $C<\infty$, Lyapunov function $V:\R^d \to [1,\infty)$, and $\rho<1$, where $\|\mu(\cdot) - \nu(\cdot)\|_{TV}:=\sup_{A\in\mathcal{B}}|\mu(A) - \nu(A)|$ for probability measures $\mu$ and $\nu$. When such a condition can be established for a reversible Markov chain, then a Central Limit Theorem exists for any square-integrable function \citep{roberts2004general}. 

We prove geometric ergodicity results for generic proposals 
as in \eqref{eq:gradient_based},
assuming $g$ to be bounded and monotone, and $\mu_\sigma$ to have lighter than exponential tails.
Following the discussion in Section \ref{sec:Barker_multi_d} we consider proposals that are independent across components, leading to
\begin{equation} \label{eq:balanced_multi}
Q^{(g)}(x,dy) = 
\prod_{i=1}^d Q^{(g)}_i(x,dy_i)
=
\prod_{i=1}^d
\frac{g\left( e^{\partial_i\log\pi(x)(y_i-x_i)} \right) \mu_\sigma(y_i-x_i)dy_i}{Z_i(x)}\,,
\end{equation}
where $Z_i(x) :=  \int_{\mathbb{R}} g(e^{\partial_i\log\pi(x)(y_i-x_i)})\mu_\sigma(y_i-x_i)dy_i$.
With a slight abuse of notation, we use $\mu_\sigma$ to represent one and $d$-dimensional densities.
The Barker proposal in \eqref{eq:barker_multi} is the special case obtained by taking $g(t)=t/(1+t)$.


For the results of this section, we make the simplifying assumption that $\pi$ is spherically symmetric outside a ball of radius $R<\infty$.

\begin{condition}\label{assumption:spherical}
There exists $R<\infty$ and a differentiable function $f:(0,\infty)\to (0,\infty)$ with $\lim_{r\to\infty}f'(r)= -\infty$ and $f'(r)$ non-increasing for $r>R$ such that $\log\pi(x)=f(\|x\|)$ for $r>R$.
\end{condition}

\begin{theorem}\label{thm:barker_ergodicity_multi_d}
Let $g:(0,\infty)\to(0,\infty)$ be a bounded and non-decreasing function, $\int_{\R} \exp(sw)\mu_{\sigma}(w)dw<\infty$ for every $s>0$, and $\inf_{w\in(-\delta,\delta)}\mu_\sigma(w)>0$ for some $\delta>0$.
If the target density $\pi$ satisfies Condition \ref{assumption:spherical}, then the Metropolis--Hastings chain with proposal $Q^{(g)}$ is $\pi$-a.e.\ geometrically ergodic.
\end{theorem}

We note that tail regularity assumptions such as Condition \ref{assumption:spherical} are common in this type of analysis (e.g. \cite{jarner2000geometric,durmus2017convergence}).  As an intuitive example, the condition is satisfied in the exponential family $\pi(x)\propto \exp(-\alpha\|x\|^\beta)$ for all $\beta> 1$.  As a contrast, for MALA and HMC it is known that for $\beta > 2$ the sampler fails to be geometrically ergodic \citep{roberts1996exponential,livingstone2016geometric}.
We expect the Barker proposal to be geometrically ergodic also for the case $\beta=1$, although we do not prove it in this work.

\subsection{Scaling with dimensionality}\label{sec:scaling_barker}
In this section we provide preliminary results suggesting that the Barker proposal enjoys scaling behaviour analogous to that of MALA in high-dimensional setings, meaning that under appropriate assumptions it requires the number of iterations per effective sample to grow as $\Theta(d^{1/3})$ with the number of dimensions $d$ as $d\to\infty$.
Similarly to Section \ref{sec:geom_erg_barker}, we prove results for general proposals $Q^{(g)}$ as in \eqref{eq:balanced_multi} with balancing functions $g$ satisfying $g(t)=t\,g(1/t)$. The Barker proposal is a special case of the latter family.

We perform an asymptotic analysis for $d\to\infty$ using the framework introduced in \cite{roberts1997weak}.
The main idea is to study the rate at which the proposal step size $\sigma$ needs to decrease as $d\to\infty$ to obtain well-behaved limiting behaviour for the MCMC algorithm under consideration (such as a $\Theta(1)$ acceptance rate and convergence 
to a non-trivial diffusion process after appropriate time re-scaling).
Based on the rate of decrease of $\sigma$ one can infer how the number of MCMC iterations required for each effective sample increases as $d\to\infty$.
For example, in the case of the random walk Metropolis $\sigma^2$ must be scaled as $\Theta(d^{-1})$ as $d\to\infty$ to have a well-behaved limit \citep{roberts1997weak}, which leads to RWM requiring $\Theta(d)$ iterations for each effective sample.
By contrast, for MALA it is sufficient to take $\sigma^2=\Theta(d^{-1/3})$ as $d\to\infty$, which leads to only $\Theta(d^{1/3})$ iterations for each effective sample \citep{roberts1998optimal}.
While these analyses are typically performed under simplifying assumptions, such as having a target distribution with i.i.d.\ components, the results have been extended in many ways (e.g.\ removing the product-form assumption, see \citet{mattingly2012diffusion}) obtaining analogous conclusions. 
See also \citet{beskos2013optimal} for optimal scaling analysis of HMC and 
\citet{roberts2016complexity} for rigorous connections between optimal scaling results and computational complexity statements.

In this section we focus on the scaling behaviour of Metropolis--Hastings algorithms with proposal $Q^{(g)}$ as in \eqref{eq:balanced_multi}, when targeting distributions of the form $\pi(x)=\prod_{i=1}^d f(x_i)$, where $f$ is a one-dimensional smooth density function.
Given the structure of $Q^{(g)}$ and $\pi(\cdot)$, the acceptance rate takes the form $\alpha(x,y)=\min\left\{1,\prod_{i=1}^d\alpha_i(x_i,y_i)\right\}$, where
\begin{align}\label{eq:alpha_i}
\alpha_i(x_i,y_i)=&
\frac{f(y_i)}{f(x_i)}
\frac{
g\left( e^{\phi'(y_i)(x_i-y_i)} \right)
}{
g\left( e^{\phi'(x_i)(y_i-x_i)} \right)
}
\frac{Z_i(x_i)}{Z_i(y_i)}\,,
\end{align}
and $\phi=\log f$.
In such a context, the scaling properties of the MCMC algorithms under consideration are typically governed by the behaviour of $\log(\alpha_i(x_i,y_i))$ as $y_i$ gets close to $x_i$, or more precisely by degree of the leading term in the Taylor series expansion of $\log(\alpha_i(x_i,x_i+\sigma u_i))$ in powers of $\sigma$ as $\sigma\to 0$ for fixed $x_i$ and $u_i$.
For example, in the case of the random walk Metropolis one has $\log(\alpha_i(x_i,x_i+\sigma u_i))=\Theta(\sigma)$ as $\sigma\to 0$, which in fact implies the proposal variance $\sigma^2$ must decrease at a rate $\Theta(d^{-1})$ to obtain a non-trivial limit.
By contrast, when the MALA proposal is used, one has $\log(\alpha_i(x_i,x_i+\sigma u_i))=\Theta(\sigma^3)$ as $\sigma\to 0$, which in turn leads to $\sigma^2=\Theta(d^{-1/3})$. 
See Sections 2.1-2.2 of \citet{durmus2017fast} for a more detailed and rigorous discussion on the connection between the Taylor series expansion of $\log(\alpha_i(x_i,y_i))$ and MCMC scaling results.
The following proposition shows that the condition $g(t)=t\,g(1/t)$, when combined with some smoothness assumptions, is sufficient to ensure that the proposals $Q^{(g)}$ lead to $\log(\alpha_i(x_i,x_i+\sigma u_i)) \leq \Theta(\sigma^3)$ as $\sigma\to 0$.

\begin{proposition}\label{prop:scaling_heuristic}
Let $g:(0,\infty)\to(0,\infty)$ and $g(t)=t\,g(1/t)$ for all $t$.
If $g$ is three times continuously differentiable and $\int_{\R}g^{(j)}(e^{sw})\mu(w)dw<\infty$ for all $s>0$ and  $j\in\{0,1,2,3\}$, where $g^{(j)}:(0,\infty)\to(0,\infty)$ is the $j$-th derivative of $g$, then
\begin{align}
\log(\alpha_i(x_i,x_i+\sigma u_i)) &\leq \Theta(\sigma^3) &\hbox{as }\sigma\to 0\,,
\end{align}
 for any $x_i$ and $u_i$ in $\R$.
\end{proposition}
Proposition \ref{prop:scaling_heuristic} suggests that Metropolis--Hastings algorithms with proposals $Q^{(g)}$ such that $g(t)=t\,g(1/t)$ have scaling behaviour analogous to MALA, meaning that $\sigma^2=\Theta(d^{-1/3})$ is sufficient to ensure a non-trivial limit and thus $\Theta(d^{1/3})$ iterations are required for each effective sample.  To make these arguments rigorous one should prove weak convergence results for $d\to\infty$, as in \cite{roberts1998optimal}.
Proving such a result for a general $g$ would require a significant amount of technical work, thus going beyond the scope of this section.
In this paper we rather support the conjecture of $\Theta(d^{1/3})$ scaling for $Q^{(g)}$ by means of simulations (see Section \ref{sec:sim_scaling}).
While Proposition \ref{prop:scaling_heuristic} only shows $\log(\alpha_i(x_i,x_i+\sigma u_i)) \leq \Theta(\sigma^3)$, it is possible to show that $\log(\alpha_i(x_i,x_i+\sigma u_i)) = \Theta(\sigma^3)$ with some extra assumptions on $\phi$ to exclude exceptional cases (see the supplement for more detail).

\section{Simulations with fixed tuning parameters}\label{sec:simulations}
Throughout Sections \ref{sec:simulations} and \ref{sec:sim_adaptive}, we choose the symmetric density  $\mu_\sigma$ within the random walk  and Barker proposals to be $N(0,\sigma^2\I_d)$ for simplicity. 
Note, however, that any symmetric density $\mu_\sigma$ could in principle be used.
It would be interesting to explore the impact of different choices of $\mu_\sigma$ to the performances of the Barker algorithm, and we leave such a comparison to future work.

\subsection{Illustrations of robustness to tuning}\label{sec:sim_tuning}
We first provide an illustration of the robustness to tuning of the random walk, Langevin and Barker algorithms in three simple one-dimensional settings. 
In each case we approximate the expected squared jump distance (ESJD) using $10^4$ Monte Carlo samples and standard Rao--Blackwellisation techniques, across of range of different proposal step-sizes between 0.01 and 100.  As is clearly shown in Figure \ref{fig:1d}, all algorithms perform similarly when the step-size is smaller than optimal, as suggested in Section \ref{sec:h_to_zero}. As the step-size increases beyond this optimum, however, behaviours begin to differ.  In particular the ESJD for MALA rapidly decays to zero, whereas in the random walk and Barker cases the reduction is much less pronounced.  In fact, the rate of decay is similar for the two schemes, which is to be expected following the results of Sections \ref{sec:barker_heterog} and \ref{sub:h_to_inf}.
See the supplement for a similar illustration on a 20-dimensional example.
\begin{figure}[h!]
\centering
\includegraphics[width=\linewidth]{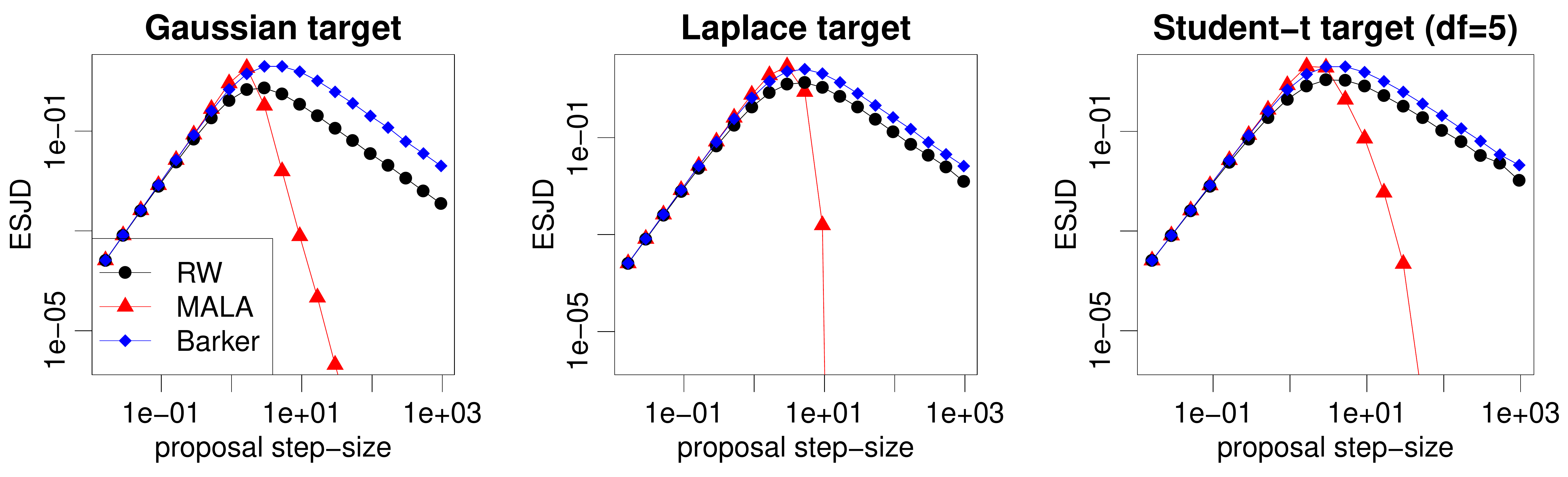}
\caption{Expected squared jump distance (ESJD) against proposal step-size for RWM, MALA and Barker on different 1-dimensional targets.}
\label{fig:1d}
\end{figure}

\subsection{Comparison of efficiency on isotropic targets}\label{sec:sim_scaling}
Next we compare the expected squared jump distance of the random walk, Langevin and Barker schemes when sampling from isotropic distributions of increasing dimension, with optimised proposal scale (chosen to maximise expected squared jumping distance).
This setup is favourable to MALA, which is the least robust scheme among the three, as the target distribution is homogeneous and the proposal step-size optimally-chosen.
We consider target distributions with independent and identically distributed (i.i.d.) components, corresponding to the scenario studied theoretically in Section \ref{sec:scaling_barker}.
We set the distribution of each coordinate to be either a standard normal distribution or a hyperbolic distribution, corresponding to $\log\pi(x)=-\sum_{i=1}^d x_i^2/2+const$ and $\log\pi(x)=-\sum_{i=1}^d (0.1+x_i^2)^{1/2}+const$, respectively.
\begin{figure}[h!]
\centering
\includegraphics[width=\linewidth]{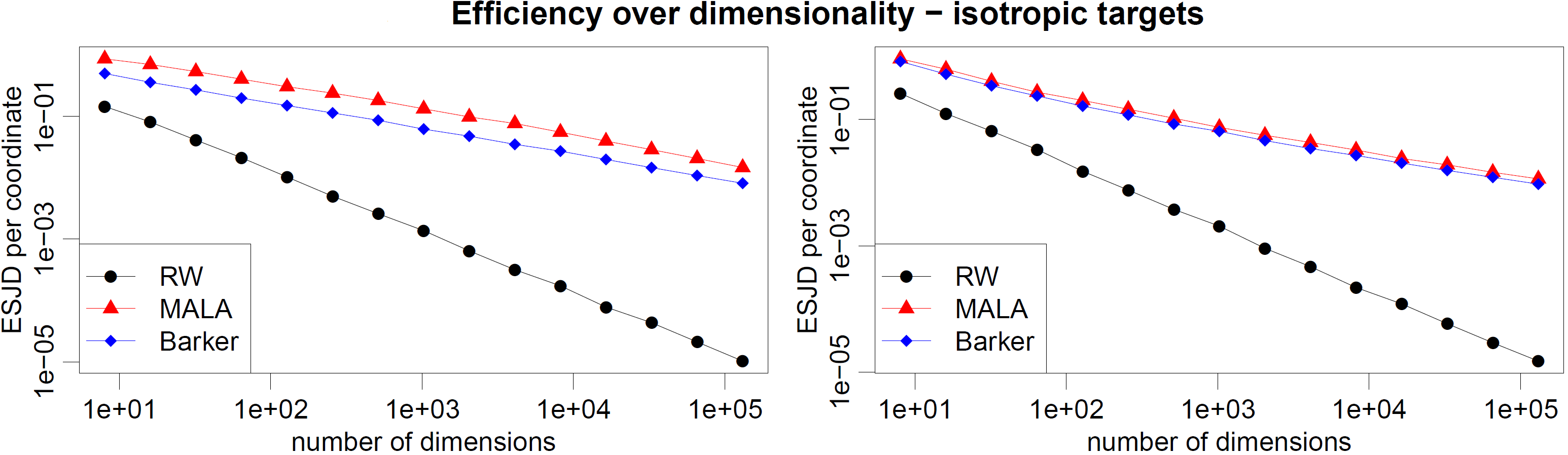}
\caption{ESJD against dimensionality for RWM, MALA and Barker schemes with optimally-tuned step size. 
The target distribution has i.i.d.\ coordinates following either a Gaussian distribution (left plot) or a hyperbolic one (right plot).
}
\label{fig:efficiency_isotropic_Gaussian}
\end{figure}
Figure \ref{fig:efficiency_isotropic_Gaussian} shows how the ESJD per coordinate decays as dimension increases for the three algorithms. 
For MALA and Barker the ESJD appears to decrease at the same rate as $d$ increases, which is in accordance with the preliminary results in Section \ref{sec:scaling_barker}.
In the Gaussian case, MALA outperforms Barker roughly by a factor of 2 regardless of dimension (more precisely, the ESJD ratio lies between $1.7$ and $2.5$ for all values of $d$ in Figure \ref{fig:efficiency_isotropic_Gaussian}), while in the hyperbolic case the same factor is around 1.2, again independently of dimension (ESJD ratio between $1.1$ and $1.25$ for all values of $d$ in Figure \ref{fig:efficiency_isotropic_Gaussian}).
The rate of decay for the random walk Metropolis is faster, as predicted by the theory.

\section{Simulations with Adaptive Markov chain Monte Carlo}\label{sec:sim_adaptive}
In this section we illustrate how robustness to tuning affects the performance of adaptive MCMC methods. 

\subsection{Adaptation strategy and algorithmic set-up}
We use Algorithm 4 in Section 5 of \cite{andrieu2008tutorial} to adapt the tuning parameters within each scheme.  Specifically, in each case a Markov chain is initialised using a chosen global proposal scale $\sigma_0$ and an identity pre-conditioning matrix $\Sigma_0 = \I_d$, 
and at each iteration the global scale and pre-conditioning matrix are updated using the equations
\begin{align}
\log(\sigma_t) 
&= 
\log(\sigma_{t-1}) + \gamma_t \times (\alpha(X^{(t)},Y^{(t)}) - \bar{\alpha}_*)
\label{eq:adapt_sigma_t}
\\
\mu_t
&=
\mu_{t-1} + \gamma_t \times (X^{(t)} - \mu_{t-1})
\\
\Sigma_t
&=
\Sigma_{t-1} + \gamma_t \times ((X^{(t)} - \mu_t)(X^{(t)}  -\mu_t)^T - \Sigma_{t-1}).\label{eq:adapt_Sigma}
\end{align}
Here $X^{(t)}$ denotes the current point in the Markov chain, $Y^{(t)}$ is the proposed move, $\mu_0 = 0$, $\bar{\alpha}_*$ denotes some ideal acceptance rate for the algorithm and the parameter $\gamma_t$ is known as the learning rate. 
We set $\bar{\alpha}_*$ to be $0.23$ for RWM, $0.57$ for MALA and $0.40$ for Barker.
We tried changing the value of $\bar{\alpha}_*$ for Barker in the range $[0.2,0.6]$ without observing major differences.
In our simulations we constrain $\Sigma_t$ to be diagonal (i.e.\ all off-diagonal terms in \eqref{eq:adapt_Sigma} are set to 0).
This is often done in practice to avoid having to learn a dense pre-conditioning matrix, which has both a high computational cost and would require a large number of MCMC samples.
See the supplement for full details on the pre-conditioned Barker schemes obtained with both diagonal and dense matrix $\Sigma_t$, including pseudo-code of the resulting algorithms. 

We set the learning rate to $\gamma_t := t^{-\kappa}$ with $\kappa\in(0.5,1)$, as for example suggested in \citep{shaby2010exploring}. 
Small values of $\kappa$ correspond to more aggressive adaptation, and for example  \cite{shaby2010exploring} suggest using $\kappa=0.8$.
In the simulations of Section \ref{sec:diag_adapt} we use $\kappa = 0.6$ as this turned out to be a good balance between fast adaptation and stability for MALA ($\kappa = 0.8$ resulted in too slow adaptation, while values of $\kappa$ lower than $0.6$ led to instability).
The adaptation of RWM and Barker was not very sensitive to the value of $\kappa$. 
Unless specified otherwise, all algorithms are randomly initialized with each coordinate sampled independently from a normal distribution with standard deviation $10$. 
Following the results from the optimal scaling theory \citep{roberts2001optimal}, we set the starting value for the global scale as $\sigma_0^2=2.4^2/d$ for RWM and $\sigma_0^2=2.4^2/d^{1/3}$ for MALA. For Barker we initialize $\sigma_0$ to the same values as MALA.

\subsection{Performance on target distributions with heterogeneous scales}\label{sec:diag_adapt}
In this section we compare the adaptive algorithms described above when sampling from target distributions with significant heterogeneity of scales across their components.
We consider 100-dimensional target distributions with different types of heterogeneity, tail behaviour and degree of skewness according to the 
following four scenarios:
\begin{enumerate}
    \item[(1)] \emph{(One coordinate with small scale; Gaussian target)} In the first scenario, we consider a Gaussian target with zero mean and diagonal covariance matrix. 
    We set the standard deviation of the first coordinate to $0.01$ and that of the other coordinates to 1. 
    This scenario mirrors the theoretical framework of Sections \ref{sec:hetero} and \ref{sec:barker_heterog} in which a single coordinate is the source of heterogeneity.
    \item[(2)] \emph{(Coordinates with random scales; Gaussian target)}
    Here we modify scenario 1 by generating the standard deviations of each coordinate randomly, sampling them independently from a log-normal distribution.
    More precisely, we sample $\log(\eta_i)\sim N(0,1)$ independently for $i=1,\dots,100$, where $\eta_i$ is the standard deviation of the $i$-th component.
    \item[(3)] \emph{(Coordinates with random scales; Hyperbolic target)} In the third scenario we change the tail behaviour of the target distribution, replacing the Gaussian with a hyperbolic distribution (a smoothed version of the Laplace distribution to ensure $\log\pi\in C^1(\R^d)$). 
    In particular, we set $\log\pi(x)=-\sum_{i=1}^d (\epsilon+(x_i/\eta_i)^2)^{1/2}+c$, with $\epsilon=0.1$ and $c$ being a normalizing constant. 
    The scale parameters $(\eta_i)_{i}$ are generated randomly as in scenario 2.
    \item[(4)] \emph{(Coordinates with random scales; Skew-normal target)} Finally, we consider a non-symmetric target distribution, which represents a more challenging and realistic situation.
    We assume that the $i$-th coordinate follows a skew-normal distribution with scale $\eta_i$ and skewness parameter $\alpha$, meaning that $\log\pi(x)=-\frac{1}{2}\sum_{i=1}^d (x_i/\eta_i)^2+\sum_{i=1}^d\log\Phi(\alpha x_i/\eta_i)+c$, with $c$ being a normalizing constant. 
    We set $\alpha=4$ and generate the $\eta_i$'s randomly as in scenario 2.
\end{enumerate}

First we provide an illustration of the behaviour of the three algorithms by plotting the trace plots of tuning parameters and MCMC trajectories - see Figure \ref{fig:k=1_learn06} for the results in scenario 1. 
\begin{figure}[t!]
\centering
\includegraphics[width=\linewidth]{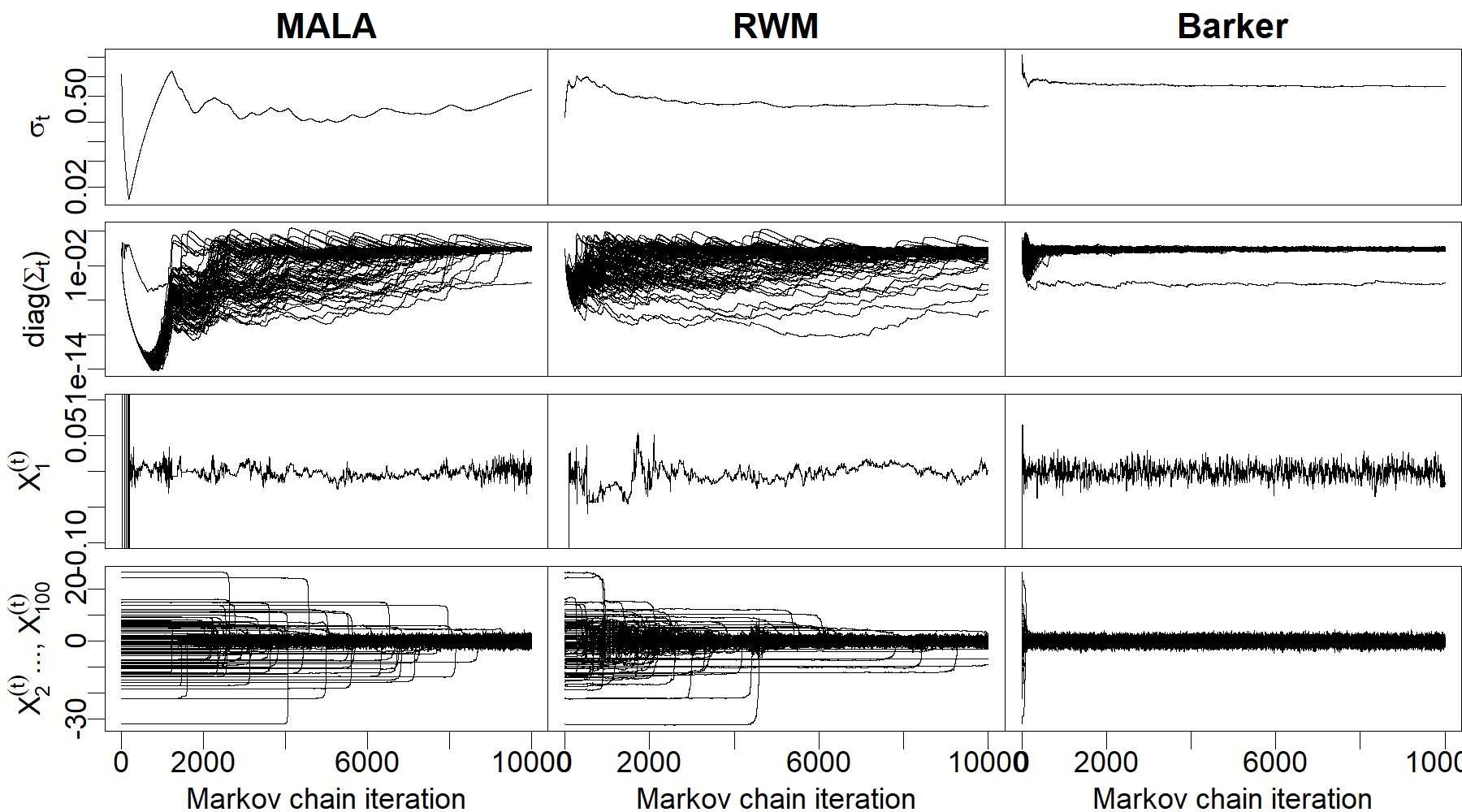}
\caption{
RWM, MALA and Barker schemes with adaptive tuning as in \eqref{eq:adapt_sigma_t}-\eqref{eq:adapt_Sigma} and learning rate set to $\gamma_t=t^{-\kappa}$ with $\kappa=0.6$.
The target distribution is a 100-dimensional Gaussian in which the first component has standard deviation $0.01$ and all others have unit scale. 
First row: adaptation of the global scale $\sigma_t$; 
second row: adaptation of the local scales $diag(\Sigma_t)=(\Sigma_{t,ii})_{i=1}^{100}$;
third row: trace plot of first coordinate;
fourth row: trace plots of coordinates from 2 to 100 (superposed).
}
\label{fig:k=1_learn06}
\end{figure}
The adaptation of tuning parameters for the Barker scheme stabilises within a few hundred iterations, after which the algorithm performance appears to be stable and efficient.
On the contrary both RWM and MALA struggle to learn the heterogeneous scales and the adaptation process has either just stabilized or not yet stabilized after $10^4$ iterations.
Looking at the behaviour of MALA in Figure \ref{fig:k=1_learn06} we see that, in order for the algorithm to achieve a non-zero acceptance rate, the global scale parameter $\sigma_t$ must first be reduced considerably to accommodate the smallest scale of $\pi(\cdot)$. At this point the algorithm can slowly begin to learn the components of the pre-conditioning matrix $\Sigma_t$, but this learning occurs very slowly because the comparatively small value for $\sigma_t$ results in poor mixing across all other dimensions than the first. 
Analogous plots for Scenarios 2, 3 and 4 are given in the supplement and display comparable behaviour.

We then compare algorithms in a more quantitative way, by looking at the average mean squared error (MSE) of MCMC estimators of the first moment of each coordinate, which is a standard metric in MCMC.
For any $h:\mathbb{R}^d\to\mathbb{R}$, define the corresponding MSE as $\mathbb{E}\big[(\hat{h}^{(t)}-\mathbb{E}_\pi[h])^2\big]$ where $\hat{h}^{(t)}=(t-t_{burn})^{-1}\sum_{i=t_{burn}+1}^t h(X^{(i)})$ is the MCMC estimator of $\mathbb{E}_\pi[h]$ after $t$ iterations of the algorithm. Here $t_{burn}$ is a burn-in period, which we set to $t_{burn}=\lfloor t/2 \rfloor$, where $\lfloor \cdot \rfloor$ denotes the floor function.
Below, we report the average MSE for the collection of test functions given by $h(x)=x_i/\eta_i$ for $i=1,\dots,d$ after $t$ MCMC iterations (rescaling by $\eta_i$ is done to give equal importance to each coordinate). 

In addition, we also monitor the rate at which the pre-conditioning matrix $\Sigma_t$ converges to the covariance of $\pi$, denoted as $\Sigma$, in order to measure how quickly the adaptation mechanism learns suitable local tuning parameters. 
We consider the $l^2$-distance between the diagonal elements of $\Sigma_t$ and $\Sigma$ on the log scale.
This leads to the following measure of convergence of the tuning parameters 
after $t$ MCMC iterations:
\begin{equation}\label{eq:d_t}
d_t
=\mathbb{E}\left[\frac{1}{\sqrt{d}}\left(\sum_{i=1}^d(\log(\Sigma_{t,ii})-\log(\Sigma_{ii}))^2\right)^{1/2}\right]\,,
\end{equation}
where the expectation is with respect the Markov chain $(X^{(t)})_{t\geq 1}$.
We use the log scale as it is arguably more appropriate than the natural one when comparing step-size parameters, and we focus on diagonal terms as both $\Sigma_t$ and $\Sigma$ are diagonal here.
Monitoring the convergence of $d_t$ to 0 we can compare the speed at which good tuning parameters are found during the adaptation process for different schemes. 

\begin{figure}[t!]
\centering
\includegraphics[width=\linewidth]{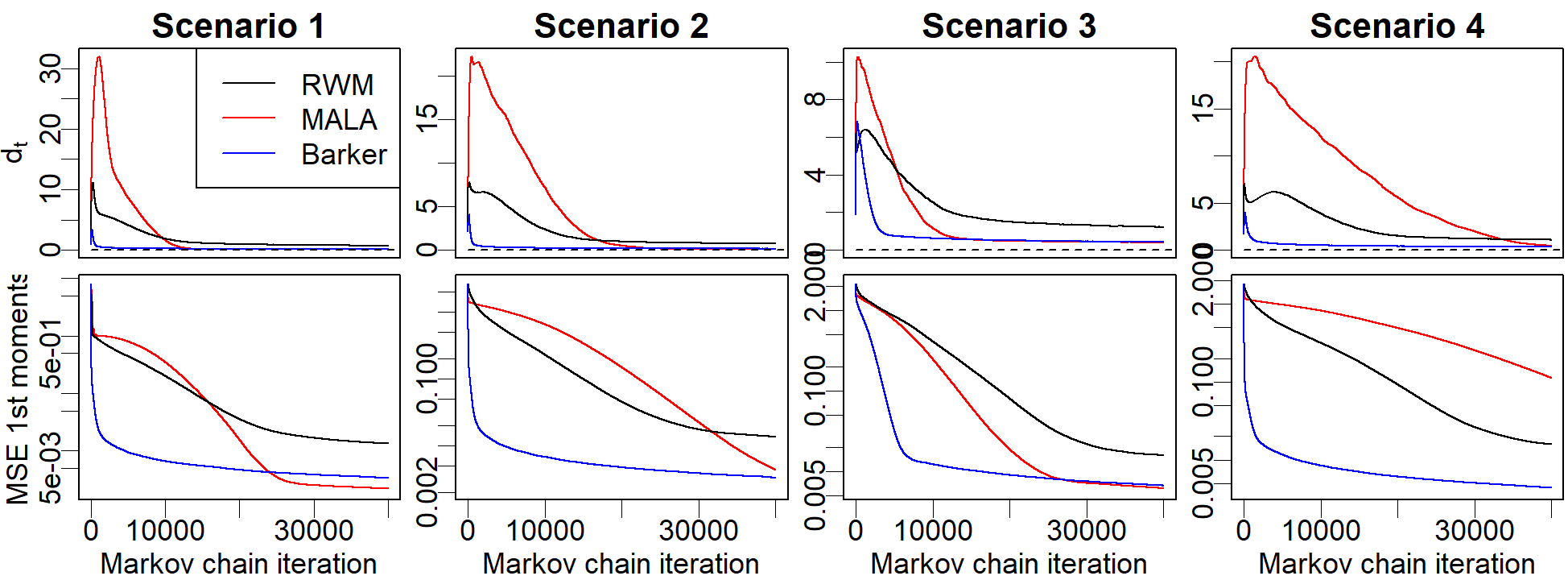}
\caption{Comparison of RWM, MALA and Barker on the four target distributions (Scenarios 1 to 4) described in Section \ref{sec:diag_adapt}, averaging over ten repetitions of each algorithm.
First row: convergence of tuning parameters, measured by $d_t$ defined in \eqref{eq:d_t}.
Second row: Mean Square Error (MSE) of MCMC estimators of first moments averaged over all coordinates.}
\label{fig:tuning_MSE_06}
\end{figure}

Figure \ref{fig:tuning_MSE_06} displays the evolution of $d_t$ and the MSE defined above over $4\times 10^4$ iterations of each algorithms, where $d_t$ and the MSE are estimated by averaging over 100 independent runs of each algorithm.
The results are in accordance with the illustration in Figure \ref{fig:k=1_learn06}, and suggest that the Barker scheme is robust to different types of targets and heterogeneity and results in very fast adaptation, while both MALA and RWM require significantly more iterations to find good tuning parameters.
The tuning parameters of MALA appear to exhibit more unstable behaviour than RWM in the first few thousands iterations (larger $d_t$), while after that they converge more quickly, which again is in accordance with the behaviour observed in Figure \ref{fig:tuning_MSE_06} and with the theoretical considerations of Sections \ref{sec:hetero} and \ref{sec:barker_heterog}.
To further quantify the tuning period, we define the time to reach a stable level of tuning as
$\tau_{adapt}(\epsilon)=\inf\{t\geq 1\,:\,d_t\leq \epsilon\}$ for some $\epsilon>0$. 
We take $\epsilon=1$ and report the resulting values in Table \ref{table:scenarios_table}, denoting $\tau_{adapt}(1)$ simply as $\tau_{adapt}$.
The results show that in these examples Barker always has the smallest adaptation time, with a speed-up compared to RWM of at least 34x in all four scenarios, and a speed-up compared to MALA ranging between 3x (scenario 3) and 30x (scenario 2).  The adaptation times $\tau_{adapt}$ tend to increase from scenario 1 to scenario 4, suggesting that the target distribution becomes more challenging as we move from scenario 1 to 4.  The hardest case for Barker seems to be the hyperbolic target, although even there the tuning stabilized in roughly 3,000 iterations, while the hardest case for MALA is the skew-normal, in which tuning stabilized in roughly 30,000 iterations.

\begin{table}\label{table:scenarios_table}
\caption{Adaptation times ($\tau_{adapt}$) and mean squared errors (\emph{MSE}) from 10k, 20k and 40k iterations of the RWM, MALA and Barker algorithms under each of the four heterogeneous scenarios described in Section \ref{sec:diag_adapt}.}
\centering
\fbox{%
\begin{tabular}{c*{5}{p{1.9cm}}}
 & \em Method & $\tau_{adapt}$ & \em MSE$_{10k}$ & \em MSE$_{20k}$ & \em MSE$_{40k}$   
\\ \hline
\parbox[t]{5mm}{\multirow{3}{*}{\em 1}}  
 & RWM    & 18,757 & 0.200 & 0.036 & 0.013 \\
 & MALA   & 10,785 & 0.348 & 0.016 & 0.002 \\
 & Barker & 524    & 0.007 & 0.005 & 0.003 \\
\hline
\parbox[t]{5mm}{\multirow{3}{*}{\em 2}}  
 & RWM    & 19,163 & 0.228 & 0.045 & 0.013 \\
 & MALA   & 17,298 & 0.644 & 0.147 & 0.004 \\
 & Barker & 542    & 0.007 & 0.005 & 0.003 \\
 \hline
\parbox[t]{5mm}{\multirow{3}{*}{\em 3}}   
 & RWM    & $>$40k & 0.409 & 0.080 & 0.016 \\
 & MALA   & 10,630 & 0.248 & 0.019 & 0.006 \\
 & Barker & 3,294  & 0.012 & 0.009 & 0.007 \\
 \hline
\parbox[t]{5mm}{\multirow{3}{*}{\em 4}}   
 & RWM    & $>$40k & 0.315 & 0.092 & 0.016 \\
 & MALA   & 34,340 & 0.813 & 0.488 & 0.112 \\
 & Barker & 1,427  & 0.008 & 0.006 & 0.004
\end{tabular}}
\end{table}

The differences in the adaptation times have a direct implication on the resulting MSE of MCMC estimators, which is intuitive because the Markov chain will typically start sampling efficiently from $\pi$ only once good tuning parameters are found.
As we see from the second row of Figure \ref{fig:tuning_MSE_06} and the second part of Table \ref{table:scenarios_table}, the MSE of Barker is already quite low (between $0.007$ and $0.012$) after $10^4$ iterations in all scenarios, while RWM and MALA need significantly more iterations to achieve the same MSE. 
After finding good tuning parameters and having sampled enough, MALA is slightly more efficient than Barker for the Gaussian target in Scenario 1 and equally efficient in the hyperbolic target of Scenario 3, which is consistent with the simulations of Section \ref{sec:sim_scaling} under optimal tuning.

\subsection{Comparison on a Poisson random effects model}\label{sec:hier_poisson}
In this section we consider 
a Poisson hierarchical model of the form
\begin{align}\nonumber
y_{ij}|\eta_i&\stackrel{ind}\sim \hbox{Poisson}(\exp(\eta_i))
&
j=1,\dots,n_i\,,\\\label{eq:Poisson_model}
\eta_i|\mu&\stackrel{ind}\sim \hbox{N}(\mu,\sigma_{\eta}^2)
&i=1,\dots,I\,,\\\nonumber
\mu&\sim \hbox{N}(0,10^2)\,,&
\end{align}
and test the algorithms on the task of sampling from the resulting posterior distribution $p(\mu,\eta_1,\dots,\eta_I|\textbf{y})$, where $\textbf{y}=(y_{ij})_{ij}$ denotes the observed data.
In our simulations we set $I=50$ and $n_i=5$ for all $i$, leading to $51$ unkown parameters and $250$ observations.

The model in \eqref{eq:Poisson_model} is an example of a generalized linear model that induces a posterior distribution with light tails and potentially large gradients of $\log\pi$, which creates a challenge for gradient-based algorithms.
In particular, the task of sampling from the posterior becomes harder when either the observations $(y_{ij})_{ij}$ contain large values or they are heterogeneous across values of $i\in\{1,\dots,I\}$.
The former case results in a more peaked posterior distribution with larger gradients, while the latter induces heterogeneity across the posterior distributions of the parameters $\eta_i$.

In our simulations we consider three scenarios, corresponding to increasingly challenging target distributions:
\begin{enumerate}
    \item[(1)] In the first scenario we take $\sigma_\eta=1$ and generate the data $\textbf{y}$ from the model in \eqref{eq:Poisson_model} assuming the data-generating value of $\mu$ to be $\mu^*=5$ and sampling the data-generating values of $\eta_1,\dots,\eta_I$ from their prior distribution.
    \item[(2)] In the second scenario we increase the value of $\sigma_\eta$ to $3$, which induces more heterogeneity across the parameters $\eta_1,\dots,\eta_I$.
    \item[(3)] In the third scenario we keep $\sigma_\eta=3$ and increase the values of $\mu^*$ to $10$, thus inducing larger gradients.
\end{enumerate}

\begin{figure}[h!]
\centering
\includegraphics[width=\linewidth]{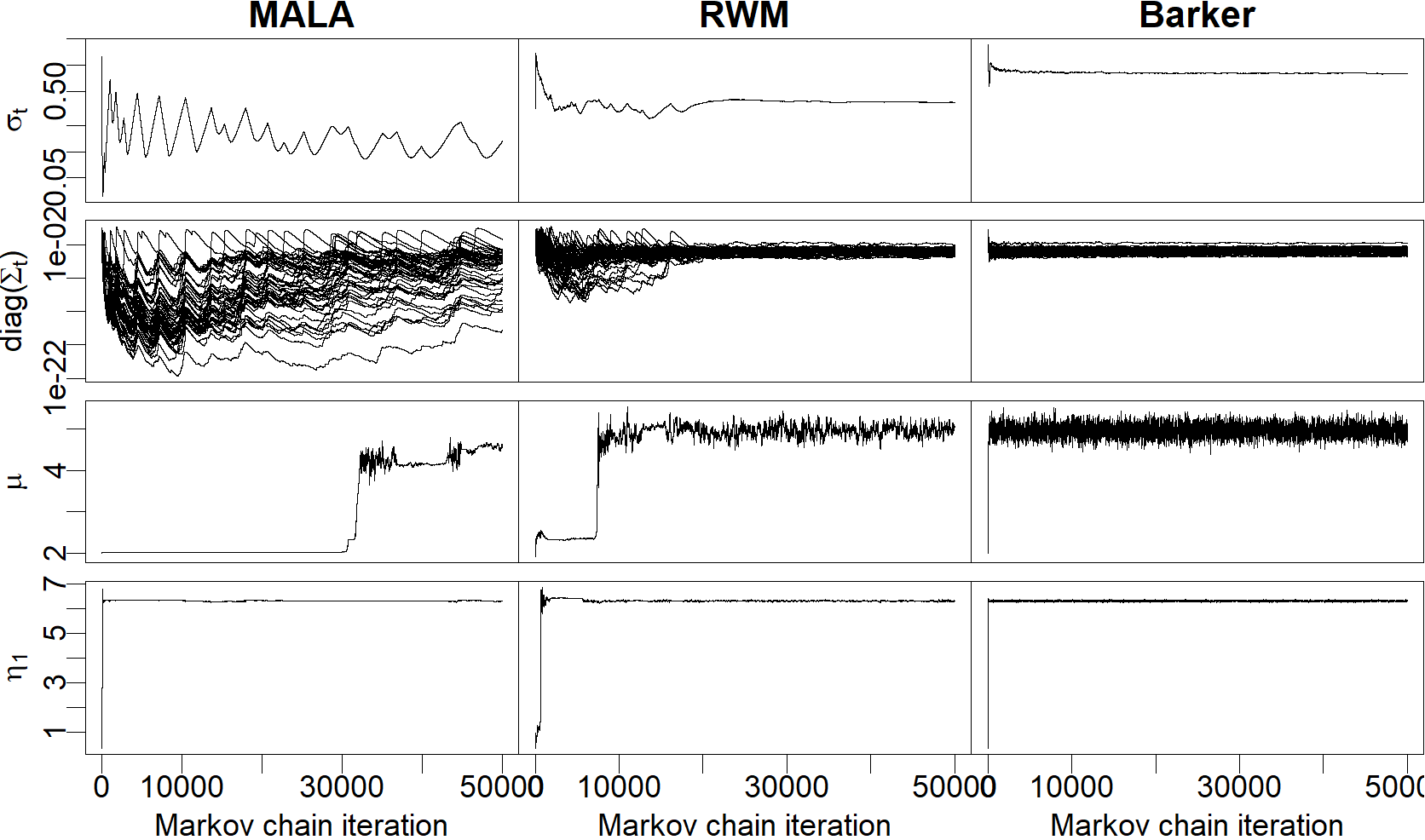}
\caption{Behavior of RWM, MALA and Barker on the posterior distribution from the Poisson hierarchical model in \eqref{eq:Poisson_model}. Data are generated as in the first scenario of Section \ref{sec:hier_poisson}.
First row: adaptation of the global scale $\sigma_t$; 
second row: adaptation of the local scales $diag(\Sigma_t)=(\Sigma_{t,ii})_{i=1}^{100}$;
third row: trace plot of the parameter $\mu$;
fourth row: trace plots of the parameter $\eta_1$.}
\label{fig:HierPoisson_mu5_sd1}
\end{figure}
Similarly to Section \ref{sec:diag_adapt}, we first provide an illustration of the behaviour of the tuning parameters and MCMC trace plots for RWM, MALA and Barker in Figure \ref{fig:HierPoisson_mu5_sd1}.
Here all algorithms are run for $5\times 10^4$ iterations, with the target defined in the first scenario.
We use the adaptation strategy of Section \ref{sec:diag_adapt} for tuning, following \eqref{eq:adapt_sigma_t}-\eqref{eq:adapt_Sigma} with $\kappa=0.6$ and $\Sigma_t$ constrained to be diagonal, and initialize the chains from a random configuration sampled from the prior distribution of the model.
In this example, the random walk converges to stationarity in roughly 10,000 iterations while the Barker scheme takes a few hundreds. 
By contrast MALA struggles to converge and exhibits unstable behaviour even after $5\times 10^4$ iterations.
Note that the first $3\times 10^4$ iterations of MALA, in which the parameter $\mu$ appears to be constant, do not correspond to rejections but rather to moves with very small increments in the $\mu$ component. 

We then provide a more systematic comparison between the algorithms under consideration in Table \ref{table:hier_poiss_table}.
In addition to RWM, MALA and Barker, we also consider a state-of-the-art implementation of adaptive Hamiltonian Monte Carlo (HMC), namely the Stan \citep{stan} implementation of the No-U-Turn Sampler (NUTS) \citep{hoffman2014no} as well as of standard HMC (referred to as ``static HMC" in the Stan package). 
The NUTS algorithm is a variant of standard HMC in which the number of leapfrog iterations, i.e.\ the parameter $L$ in \eqref{eq:transition}, is allowed to depend on the current state (using a ``No-U-Turn" criterion).
The resulting number of leapfrog steps (and thus log-posterior gradient evaluations) per iteration is not fixed in advance but rather tuned adaptively depending on the hardness of the problem.
This is also the case for the static HMC algorithm implementation in Stan, as in that case the total integration time in  \eqref{eq:transition} is fixed and the step-size and mass matrix are adapted. 
For both algorithms we use the default Stan version that learns a diagonal covariance/mass matrix during the adaptation process.
This is analogous to constraining the preconditioning matrix $\Sigma_t$ for RWM, MALA and Barker to be diagonal, as we are doing here. 

\begin{table}\label{table:hier_poiss_table}
\caption{Comparison of sampling schemes on the posterior distribution arising from the Poisson hierarchical model in \eqref{eq:Poisson_model}. 
Blocks of rows from 1 to 3 refer to the three data-generating scenarios described in Section \ref{sec:hier_poisson}.
All numbers are averaged across ten repetitions of each algorithm.
For each algorithm we report: number of iterations; number of leapfrog steps per iteration and total number of gradient evaluations (when applicable); estimated ESS (minimum and median across parameters); minimum ESS per hundred gradient evaluations (with standard deviation across the ten repetitions).
}
\centering
\fbox{%
\begin{tabular}{c*{6}{p{1.9cm}}}
 & \em Method & \em Iterations $(n)$ & \em Leapfrog steps $/n$ & \em Gradient calls $(g)$ & \em ESS & \em ESS$/g \times 100$ \\ 
\hline
\parbox[t]{5mm}{\multirow{5}{*}{\em 1}}  
& RWM & 
$5\times 10^4$ & -  & - &  (49,66) & -
\\
& MALA & 
$5\times 10^4$& -  & $5\times 10^4$ & (648,727) & 1.30 $\pm$ 2.73
\\
& Barker & 
$5\times 10^4$& - & $5\times 10^4$ & (1445,1587) & 2.89 $\pm$ 0.07
\\ 
& HMC &
 $2 \times 10^3$ & 89.5 & $1.8\times 10^5$ & (285,1954) & 0.25 $\pm$ 0.78
\\ 
& NUTS  & 
$2 \times 10^3 $ & 8.5  & $1.7\times 10^4$ & (1175,1822) & 6.95 $\pm$ 1.68
\\ 
\hline
\parbox[t]{5mm}{\multirow{5}{*}{\em 2}}  
& RWM & 
$5\times 10^4$ & -  & - &  (0.4,10.6) & -
\\
& MALA & 
$5\times 10^4$& -  & $5\times 10^4$ & (0.0,8.0) & $<$0.01
\\
& Barker & 
$5\times 10^4$& - & $5\times 10^4$ & (1365,1563) & 2.73 $\pm$ 0.13
\\ 
& HMC &
 $2 \times 10^3$ & 797 & $1.6\times 10^6$ & (25,1949) & $<$0.01
\\ 
& NUTS  & 
$2\times 10^3$ & 57.7  & $1.2\times 10^5$ & (942,1826) & 1.19 $\pm$ 1.14
\\ 
\hline
\parbox[t]{5mm}{\multirow{5}{*}{\em 3}}  
& RWM & 
$5\times 10^4$ & -  & - &  (0.0,5.3) & -
\\
& MALA & 
$5\times 10^4$& -  & $5\times 10^4$ & (0.0,0.2) & $<$0.01
\\
& Barker & 
$5\times 10^4$& - & $5\times 10^4$ & (1301,1594) & 2.60 $\pm$ 0.92
\\ 
& HMC &
 $2\times 10^3$ & 8103 & $1.6\times 10^7$ & (3.3,899) & $<$0.01
\\ 
& NUTS  & 
$2 \times 10^3$ & 179  & $3.5\times 10^5$ & (137,348) & 0.012$\pm$0.14
\end{tabular}}
\end{table}

Table \ref{table:hier_poiss_table} reports the results of the simulations for the five algorithms in each of the three scenarios.
For each algorithm, we report the number of log-posterior gradient evaluations and the minimum and median effective sample size (ESS) across the $51$ unknown parameters.
The ESS values are computed with the \texttt{effectiveSize} function from the \texttt{coda} R package \citep{coda}, discarding the first half of the samples as burn-in.
The RWM, MALA and Barker schemes are run for $5\times 10^4$ iterations, and the HMC and NUTS schemes for $2\times 10^3$ iterations.
The latter is the default value in the Stan package and in this example corresponds to a number of gradient evaluations between $1.7\times 10^4$ and $1.6\times 10^7$.
All numbers in Table \ref{table:hier_poiss_table} are averaged over ten independent replications of each algorithm.
We use the minimum ESS per gradient evaluation as an efficiency metric, of which we report the mean and standard deviation across the ten replicates (multiplied by 100 to facilitate readability).

According to Table \ref{table:hier_poiss_table}, NUTS is the most efficient scheme in scenario 1, while Barker is the most efficient one in scenarios 2 and 3.
This is in accordance with the intuition of Barker being a more robust scheme, as the target distribution becomes more challenging as we move from scenario 1 to 3.
MALA struggles to converge to stationarity in scenarios 2 and 3 (with an estimated ESS around zero), while it performs better in scenario 1, although with a high variability across different runs (shown by the large standard deviation in the last column).
The RWM displays low ESS values for all three scenarios, although with a less dramatic deterioration going from scenario 1 to 3.
Interestingly, the performances of Barker 
are remarkably stable across scenarios (with an ESS of around 1400), as well as across parameters for which ESS is computed (in all cases the minimum and median ESS are close to each other) and across repetitions (shown by the relatively small standard deviation in the last column).
We note that NUTS is also remarkably effective taking into consideration that it is not an algorithm designed with a major emphasis on robustness, but that performance does degrade when moving from scenario 1 to scenario 3.
As in the MALA case, static HMC struggles to converge in scenarios 2 and 3 and is not very efficient in scenario 1.
Note that NUTS, and in particular HMC, compensate for the increasing difficulty of the target by increasing the number of leapfrog steps per iteration. For example, the drop in efficiency of NUTS between scenarios 1 and 2 is mostly due to the increase in average number of leapfrog iterations from $8.5$ to $57.7$ rather than in a decrease in ESS.
Somewhat surprisingly, in static HMC the number of leapfrog steps per iteration is increased significantly more than NUTS, which 
could either be due to genuine algorithmic differences or to variations in the details of the adaptation strategy implemented in Stan. 
 Overall, Barker and NUTS are the two most efficient algorithms in these simulation, with a relative efficiency that depends on the scenario under consideration: NUTS being roughly 2.4 times more efficient in scenario 1, Barker 2.3 times more efficient in scenario 2 and Barker 40 times more efficient in scenario 3. 

\subsection{Additional simulations reported in the supplement}
In the supplement we report additional simulations for 
some of the above experiments.  As a sensitivity check, we also performed simulations using the tamed Metropolis-adjusted Langevin algorithm \citep{brosse2018tamed} and the truncated Metropolis-adjusted Langevin algorithm \citep{roberts1996exponential,atchade2006adaptive}, two more robust modifications to MALA in which large gradients are controlled by monitoring the size of $\|\nabla\log\pi(x)\|$.  The schemes do offer some added stability compared to MALA in terms of controlling large gradients, but ultimately are still very sensitive to heterogeneity of the target distribution and to the choice of the truncation level, and do not exhibit the same robustness observed in the case of the Barker scheme.  See the supplement for implementation details, results and further discussion.

\section{Discussion} 

We have introduced a new gradient-based MCMC method, \textit{the Barker proposal}, and have demonstrated both analytically and numerically that it shares the favourable scaling properties of other gradient-based approaches, along with an increased level of robustness, both in terms of geometric ergodicity and robustness to tuning (as defined in the present paper).  
The most striking benefit of the method appears to be in the context of adaptive Markov chain Monte Carlo.
Evidence suggests that combining the efficiency of a gradient-based proposal mechanism with a method that exhibits robustness to tuning gives a combination of stability and speed that is very desirable in this setting, and can lead to efficient sampling that requires minimal practitioner input.

The theoretical results in this paper could be extended by studying in greater depth the large $\lambda$ regime (Section \ref{sec:h_to_zero}) and the high-dimensional scaling of the Barker proposal (Section \ref{sec:scaling_barker}).
Of course, there are many other algorithms that could be considered under the robustness to tuning framework, and it is worthwhile future work to explore which features of a scheme result in either robustness to tuning or a lack of it.
Extensions to the Barker proposal that incorporate momentum and exhibit the $d^{-1/4}$ decay in efficiency with dimension enjoyed by Hamiltonian Monte Carlo may be possible, as well as the development of other methods within the first-order locally-balanced proposal framework introduced in Section \ref{sec:barker}, or indeed schemes that are exact at higher orders.

\section*{Acknowledgements}
The authors thank Marco Frangi for preliminary work in his Master's thesis related to the ergodicity arguments in the paper.
SL acknowledges support from the engineering and physical sciences research council through grant number EP/K014463/1 (i-like). GZ acknowledges support from the European research council starting grant 306406 (N-BNP), and by the Italian ministry of education, universities and research PRIN Project 2015SNS29B.

\bibliography{barker_ref}

\newpage

\appendix

\vspace{2cm}

\begin{center}
    {\LARGE Supplement to `On the robustness of gradient-based MCMC algorithms.'}
\end{center}

\vspace{2cm}

The ary material contains proofs of the theoretical results 
 and additional figures related to the simulations.
It also includes some background on the key techniques used for the proofs of Section \ref{sec:hetero}, a proof of Condition \ref{cond:gap_MALA}  for the exponential family class and details related to skew-symmetry and pre-conditioning of the Barker proposal.
In this supplement, we number equations, figures and lemmas differently from the main paper, e.g.\ (1) rather than (1.1), to avoid confusion between the two documents.

\section{Tools to bound spectral gaps}

To establish lower bounds on spectral gaps we use the following Lemma.

\begin{lemma} \label{lemma:lowerbound} Consider two Metropolis--Hastings kernels $P_1$ and $P_2$ with associated candidate kernels $Q_1(x,dy) = q_1(x,y)dy$ and $Q_2(x,dy) = q_2(x,y)dy$ and common target distribution $\pi$. 
If there is a $\gamma >0$ such that $q_1(x,y) \geq \gamma q_2(x,y)$ for all fixed $x,y$ with $x \neq y$, then
\begin{equation} \label{eq:lowerbound}
\textup{Gap}(P_1)\geq \gamma\textup{Gap}(P_2) .
\end{equation}
\end{lemma}

\begin{proof}

For any $f \in L_{0,1}^2(\pi)$, it holds that
\begin{align*}
&\int \{ f(y) - f(x) \}^2 \pi(dx) P_1(x,dy)
\\&=
\int \{ f(y) - f(x) \}^2 
\min\left\{ 1, \frac{\pi(y)q_1(y,x)}{\pi(x)q_1(x,y)}\right\}
\pi(x) q_1(x,y)dx\,dy
\\&=
\int \{ f(y) - f(x) \}^2 
\min\left\{ \pi(x) q_1(x,y), \pi(y)q_1(y,x)\right\}
dx\,dy \\
&\geq \gamma \int \{ f(y) - f(x) \}^2 
\min\left\{ \pi(x) q_2(x,y), \pi(y)q_2(y,x)\right\}
\\
&= \gamma \int \{ f(y) - f(x) \}^2 
\min\left\{ 1, \frac{\pi(y)q_2(y,x)}{\pi(x)q_2(x,y)}\right\}
\pi(x) q_2(x,y)dx\,dy
\\
&= \gamma \int \{ f(y) - f(x) \}^2\pi(dx) P_2(x,dy).
\end{align*}
The result follows from the Dirichlet forms characterization of spectral gaps in \eqref{eq:spectral_gap}.
\end{proof}

To find upper bounds we use the notion of \textit{conductance} for a Markov chain.  Define the conductance of a set $K \in \mathcal{B}$ with $0<\pi(K)\leq 1/2$ for a $\pi$-reversible Markov chain with transition kernel $P$ as
\[
\Phi(K):=\frac{\int_{K}\pi(dx)P(x,K^{c})}{\pi(K)},
\]
which is the conditional probability $\mathbb{P}(X^{(t+1)}\in K^{c}|X^{(t)}\in K)$
provided $X^{(t)}\sim\pi(\cdot)$. Recall the spectral gap bound for $P$
that for any such $K$
\begin{equation} \label{eq:conductance}
\text{Gap}(P)\leq 2\Phi(K).
\end{equation}
This can be seen directly by setting $g(x) = \pi(K^C)\mathbb{I}(x \in K)-\pi(K)\mathbb{I}(x \in K^c)$, letting $f(x) := g(x)/\int g(x)^2\pi(dx)$ and computing the Dirichlet form of $f$ using \eqref{eq:spectral_gap}.  Here $\mathbb{I}(\cdot)$ denotes the indicator function.
\comment{Intuitively, $\Phi(K)$ gives information about how able a Markov chain is to move between the disjoint regions $K$ and $K^c$ in a single time step, and so can be naturally related to first-order auto-correlations, and hence also spectral gaps.  Conductance arguments can also be used to establish lower bounds for spectral gaps, but these are more difficult to show in practice since one must find the infimum of $\Phi(K)$ over all appropriate $K \in \mathcal{B}$.}

\section{Change of variables and isomorphic Markov chains}\label{sec:equiv} 
In this section we provide two lemmas showing that bijective mappings do not change the spectral gaps of Markov chains, nor the Metropolis-Hastings dynamics. 
These lemmas will allow us to prove the results of Section \ref{sec:hetero} working with the equivalent formulation where the target is fixed and the proposal distribution is changing, rather than having a target that changes with $\lambda$. 
This will in turn allow us to exploit results such as Lemma \ref{lemma:lowerbound}, thus significantly simplifying the proofs.

We follow the terminology of Johnson and Geyer (2013), and introduce the notion of \emph{isomorphic} Markov chains.
Intuitively, two Markov chains $(X^{(t)})_{t\geq 1}$ and $(Y^{(t)})_{t\geq 1}$ are isomorphic if $(\phi(X^{(t)}))_{t\geq 1}$ is equal in distribution to $(Y^{(t)})_{t\geq 1}$ for some bijective map $\phi$.
More formally, let $(X^{(t)})_{t\geq 1}$ and $(Y^{(t)})_{t\geq 1}$ be Markov chains with transition kernels $P$ and $K$ and state spaces $(S,\mathcal{A})$ and $(T,\mathcal{B})$, respectively.
We say that $(X^{(t)})_{t\geq 1}$ and $(Y^{(t)})_{t\geq 1}$ are {isomorphic} if there exists a bijective 
 function $\phi$ from $S$ to $T$ such that 
\begin{align}\label{eq:isomorphic_chains}
P(x,A)=&K(\phi(x),\phi(A))
&x\in S,\,A\in\mathcal{A}.
\end{align}
Equation \eqref{eq:isomorphic_chains} means that $K(\phi(x),\cdot)$ is the push-forward of $P(x,\cdot)$ under $\phi$ for every $x\in\R^d$, which we write as $K=\phi\circ P$.
We will use $\circ$ to denote the push-forward operator for both probability distributions and transition kernels, so that $(\phi\circ\pi)(B)=\pi(\phi^{-1}(B))$ and $(\phi\circ P)(y,B)=P(\phi^{-1}(y),\phi^{-1}(B))$.

Isomorphic Markov chains share the same convergence behaviour and, in particular, they have the same $L^2$-spectral gap, as stated in the following lemma (see Lemma 1 of Papaspilioupulus et al.\ (2019) 
for a proof of analogous results).
\begin{lemma}\label{lemma:isomorphic_MC}
Isomorphic Markov chains have the same $L^2$-spectral gap.
\end{lemma}

In the following we will exploit the fact that the Metropolis-Hastings (MH) algorithm preserves isomorphisms of Markov chains under transformations of both the target and candidate distributions, as shown by the following lemma.
\begin{lemma}\label{lemma:MH_isomorphic}
Let $\phi: S\to T$ be a bijective function and $(X^{(t)})_{t\geq 1}$ and $(Y^{(t)})_{t\geq 1}$ be Metropolis-Hastings Markov chains defined on $(S,\mathcal{A})$ and $(T,\mathcal{B})$ with target distributions $\pi$ and $\phi\circ\pi$, respectively, and proposal kernels $Q$ and $\phi\circ Q$, respectively.
Then $(X^{(t)})_{t\geq 1}$ and $(Y^{(t)})_{t\geq 1}$ are isomorphic Markov chains.
\end{lemma}
\begin{proof}
Let $\mu^\phi(dy,dy'):= (\phi\circ \pi)(dy')(\phi\circ Q)(y',dy)$ and $\mu^\phi_T(dy,dy'):= \mu^\phi(dy',dy)$. Then using Proposition 1 of \cite{tierney1998note} there exists a set $R^\phi \in \mathcal{B} \times \mathcal{B}$ such that $\mu^\phi$ and $\mu^\phi_T$ are mutually absolutely continuous on $R^\phi$ and mutually singular on its complement.  The Radon-Nikodym derivative $d\mu^\phi/d\mu^\phi_T (y,y')^T$ is therefore finite and positive when restricted to $R^\phi$.
Let $r^\phi(y,y'):= d\mu^\phi/d\mu^\phi_T(y,y')$ if $(y,y') \in R^\phi$ and $r^\phi(y,y'):= 0$ otherwise. Then the Metropolis--Hastings acceptance probability for the chain $(Y_t)_{t\geq 1}$ can be written $\alpha^\phi(y,y'):=\min(1,r^\phi(y,y'))$.
Similarly, letting $\mu(dx,dx'):=\pi(dx')Q(x',dx)$ and $\mu_T(dx,dx'):=\mu(dx',dx)$ the acceptance probability for $(X_t)_{t\geq 1}$ can be written $\alpha(x,x'):=\min(1,r(x,x'))$ where $r(x,x'):= d\mu/d\mu_T(x,x')$ when $(x,x') \in R \in S \times S$ and 0 otherwise, with $R$ defined analogously to $R^\phi$ for the measures $\mu$ and $\mu_T$.

Note first that from the definitions of push-forward measure and transition kernel given above that $\mu^\phi(A,B) = \mu(\phi^{-1}(A),\phi^{-1}(B))$ and $\mu_T^\phi(A,B) = \mu_T(\phi^{-1}(A),\phi^{-1}(B))$  for any $(A,B) \in \mathcal{B} \times \mathcal{B}$.  From this it follows that $R \in \mathcal{A} \times \mathcal{A}$ is the pre-image under $\phi$ of $R^\phi \in \mathcal{B} \times \mathcal{B}$, and further that
$$
\alpha^\phi(y,y') = \min(1, r^\phi(y,y')) = \min(1, r(\phi^{-1}(y),\phi^{-1}(y')) = \alpha(\phi^{-1}(y),\phi^{-1}(y')).
$$
Denoting the transition kernels of $(X^{(t)})_{t\geq 1}$ and $(Y^{(t)})_{t\geq 1}$ as $P$ and $K$ respectively, it therefore holds that
\begin{align*}
K(y,B)
=&
\delta_{B}(y)\int_{T}(1-\alpha^{\phi}(y,y'))\phi\circ Q(y,dy')
+\int_{B}\alpha^{\phi}(y,y')\phi\circ Q(y,dy')
\\=&
\delta_{\phi^{-1}(B)}(\phi^{-1}(y))\int_{S}(1-\alpha(\phi^{-1}(y),x'))Q(\phi^{-1}(y),dx')
\\
&+\int_{\phi^{-1}(B)}\alpha^{\phi}(\phi^{-1}(y),x')Q(\phi^{-1}(y),dx')
\\=&
P(\phi^{-1}(y),\phi^{-1}(B))
\end{align*}
meaning that $(X^{(t)})_{t\geq 1}$ and $(Y^{(t)})_{t\geq 1}$ are isomorphic.
\end{proof}

\section{Proofs}
Throughout the proofs we often use $\|\cdot\|$ to denote the standard euclidean norm $\|\cdot\|_2$.

\subsection{Proofs for Section \ref{sec:hetero}}

\begin{proof}[Proof of Proposition \ref{prop:rwm_condition}]

We first establish that $\mu(\delta_\lambda) \geq \mu(\delta)$ whenever $\lambda \leq \lambda_0$ for some $\lambda_0>0$.  In cases (i) and (iii) $\mu(z)$ is monotonically decreasing in $\|z\|_2^2$. So when $\lambda < 1$ it holds that
\[
\|\delta_\lambda\|_2^2 = \sum_{i=1}^d (y_i - x_i)^2 + \lambda^2(y_1 - x_1)^2 \leq \|\delta\|_2^2,
\]
which proves the condition for $\lambda_0=1$. In case (ii) $\mu(z)$ is monotonically decreasing in $\|z\|_1$, so again $\|\delta_\lambda\|_1 = \lambda|y_1 - x_1| + \sum_{i=2}^d |y_i - x_i| \leq \|\delta\|$ when $\lambda<1$. 
The statement that $\sup_{z_1 \in \R}\mu_1(z_1)<\infty$ follows by noting that in all three cases the marginal $\mu_1$ is known in closed form and is, respectively, a Gaussian, Laplace and Student's t distribution, all of which have bounded density.
\end{proof}

\begin{proof}[Proof of Theorem \ref{thm:RW_lambda_to_0}]
Instead of studying directly $P_\lambda^R$, we will study the transition kernel $\tilde{P}^R_\lambda$ corresponding to a Metropolis-Hastings (MH) algorithm with proposal $\phi\circ Q^R$ and target $\phi\circ\pi^{(\lambda)}$, for some bijective $\phi:\R^d\to\R^d$.
By Lemma \ref{lemma:MH_isomorphic}, $\tilde{P}^R_\lambda$ and $P^R_\lambda$ induce isomorphic Markov chains and thus by Lemma \ref{lemma:isomorphic_MC} we have $\text{Gap}(\tilde{P}^R_\lambda)=\text{Gap}(P^R_\lambda)$. 
We consider $\phi$ given by $\phi(x_1,\dots,x_d)=(\lambda^{-1}x_1,x_2,\dots,x_d)$.
It follows that $\phi\circ\pi^{(\lambda)}=\pi$ and that $\tilde{Q}^R_{\lambda}=\phi\circ Q$ satisfies $\tilde{Q}^R_{\lambda}(x,dy)=\tilde{q}^R_\lambda(x,y)dy$ with
\begin{equation}\label{eq:qRh_def}
\tilde{q}^R_\lambda(x,y) := \frac{\lambda}{\sigma^d}\mu\left( \frac{\delta_\lambda}{\sigma} \right),
\end{equation}
and $\delta_\lambda$ defined as in equation \eqref{eqn:delta_h}.

First we show that for all $\lambda \leq \lambda_0$ and all $x,y \in \R^d$ it holds that $\tilde{q}^R_\lambda(x,y) \geq \lambda \tilde{q}^R_1(x,y)$, where $\lambda_0>0$ is the value defined in Condition \ref{cond:rwm}. 
From \eqref{eq:qRh_def}, we have
\begin{equation}\label{eq:RW_ratio_proposals}
\frac{\tilde{q}_\lambda^R(x,y)}{\tilde{q}^R_1(x,y)} = \lambda\frac{\mu( \delta_\lambda/\sigma )}{\mu(\delta/\sigma)}.
\end{equation}
Condition \ref{cond:rwm} guarantees $\mu( \delta_\lambda/\sigma )\geq\mu(\delta/\sigma)$ for all $\lambda \leq \lambda_0$, which together with \eqref{eq:RW_ratio_proposals} gives $\tilde{q}^R_\lambda(x,y) \geq \lambda \tilde{q}^R_1(x,y)$.  Combining the latter inequality with Lemma \ref{lemma:lowerbound} gives 
$$
\text{Gap}(\tilde{P}_\lambda^R)\geq \lambda \text{Gap}(\tilde{P}^R_1)=\Theta(\lambda)\qquad \hbox{as }\lambda\downarrow 0\,.
$$

To show that $\text{Gap}(\tilde{P}_\lambda^R)\leq \Theta(\lambda)$, take $X^{(t)} \sim \pi(\cdot)$ and $X^{(t+1)} | X^{(t)} \sim \tilde{P}_\lambda^R(X^{(t)},\cdot)$.
We consider the set $K:=\{y\in\mathbb{R}^{d}:|y_{1}|>k\}$, with $k$ chosen such that $0<\pi(K)<1/2$ (since $\pi(\cdot)$ is defined on a Polish space,
it is tight, meaning this is always possible).
Recall from \eqref{eq:conductance} that $\text{Gap}(\tilde{P}_\lambda^R)\leq 2\mathbb{P}(X^{(t+1)}\in K^c\,|\,X^{(t)}\in K)$.
We have
\begin{align*}
\mathbb{P}(X^{(t+1)}\in K^c\,\big|\,X^{(t)}\in K)
\leq\,&
\mathbb{P}(|X^{(t)}_1+ \sigma \lambda^{-1} \xi_1|\leq k\,\big|\,X^{(t)}\in K)
\\
=\,&
\mathbb{P}(-X^{(t)}_1-k \leq \sigma \lambda^{-1} \xi_1\leq-X^{(t)}_1+k\,\big|\,X^{(t)}\in K)
\\
\leq\,&
\sigma^{-1}\lambda 2k \sup_{z_1 \in \R} \mu_1(z_1)\,,
\end{align*}
where $\xi_1$ is the first component of $\xi\sim\mu$.
Conditon \ref{cond:rwm} implies $\sup_{z_1\in \R}\mu_1(z_1)<\infty$, giving that $\text{Gap}(\tilde{P}_\lambda^R)\leq\Theta(\lambda)$ for $\lambda\downarrow 0$, as desired.
\end{proof}

\begin{proof}[Proof of Theorem \ref{thm:MALA_product}]
This follows directly from the proof of Theorem \ref{thm:HMC} below, by noting that setting $L=1$ in Hamiltonian Monte Carlo gives the Langevin algorithm.
\end{proof}

\begin{proof}[Proof of Theorem \ref{thm:MALA_general}]
Similarly to the proof of Theorem \ref{thm:RW_lambda_to_0}, instead of studying directly $P_\lambda^M$, we will study the MH transition kernel $\tilde{P}^M_\lambda$ with proposal $\phi\circ Q^M_\lambda$ and target $\phi\circ\pi^{(\lambda)}$.
Lemma \ref{lemma:MH_isomorphic} and Lemma \ref{lemma:isomorphic_MC} imply $\text{Gap}(\tilde{P}^M_\lambda)=\text{Gap}(P^M_\lambda)$. 
We consider the same $\phi$ as in the proof of Theorem \ref{thm:RW_lambda_to_0}, which we write as $\phi(x)=\Sigma_\lambda^{1/2}x$ with
\[
 \Sigma_\lambda = 	
\begin{pmatrix} 
    		\lambda^{-2} & (0,\dots,0) \\
            (0,\dots,0)^T & \I_{d-1}
    		\end{pmatrix}.
\]
We have $\pi=\phi\circ\pi^{(\lambda)}$ as stated above.
Also, $\tilde{Q}_\lambda^M=\phi\circ Q_\lambda^M$ satisfies $\tilde{Q}_\lambda^M(y,\cdot)=N(y+\frac{\sigma^2}{2} \Sigma_\lambda 
\nabla\log\pi(y), \sigma^2 \Sigma_\lambda)$ where $ \Sigma_\lambda $ is as above.
This is a fairly standard calculation, analogous to the derivation of the preconditioned MALA algorithm with preconditioning matrix $ \Sigma_\lambda^{1/2}$, which
we report here for completeness.
By definition of $\phi$ and $Q^M_\lambda(x,\cdot)=N(x+\frac{\sigma^2}{2}\nabla \log\pi^{(\lambda)}(x) , \sigma^2\mathbb{I}_d)$, we have 
$\phi\circ Q^M_\lambda(x,\cdot)=N\left(\phi\big(x+\frac{\sigma^2}{2}\nabla \log\pi^{(\lambda)}(x)\big),\sigma^2 \Sigma_\lambda\right)$  for each $x\in\R^d$.
Also, since $\log\pi^{(\lambda)}(x)=\log\pi(\phi(x)) +const$ and $\phi(x)=\Sigma_\lambda^{1/2}x$,
 we have $\nabla \log\pi^{(\lambda)}(x)=\Sigma_\lambda^{1/2}\nabla\log\pi(\phi(x))$.
Therefore 
\begin{align*}
\phi\left(x+\frac{\sigma^2}{2}\nabla \log\pi^{(\lambda)}(x)\right)
=\Sigma_\lambda^{1/2}\left(x+\frac{\sigma^2}{2}\Sigma_\lambda^{1/2}\nabla \log\pi(\phi(x))\right)
=
\phi(x)+\frac{\sigma^2}{2}\Sigma_\lambda\nabla\log\pi(\phi(x))\,
\end{align*}
meaning that $\tilde{Q}^M_\lambda(\phi(x),\cdot)$ is the push-forward of $Q^M_\lambda(x,\cdot)$ under $\phi$ for every $x\in\R^d$, as desired.

We now prove $\text{Gap}(\tilde{P}^M_\lambda)\leq \Theta(e^{-\lambda^{-\alpha}})$ as $\lambda\downarrow 0$ for some $\alpha>0$.
We take $\sigma=1$ for simplicity of notation (or otherwise replace $\lambda$ by $\sigma^{-1} \lambda$) and we assume $\lambda<1 $ without loss of generality (we are studying a limit $\lambda\downarrow 0$). 
Let $(X^{(t)})_{t=1}^{\infty}$ be a Markov chain with transition kernel $\tilde{P}^M_\lambda$ started in stationarity.
We consider the sets $A_\lambda:=\{y\in\mathbb{R}^{d}:|y_{1}|\leq \lambda^{-1/(2\tilde{\gamma})}\}$, where $\tilde{\gamma}=\max\{1,\gamma\}$, 
and $K:=\{y\in\mathbb{R}^{d}:|y_{1}|>k\}$, with $k$ chosen such that $0<\pi(K)<1/2$.
Given
$$\epsilon\in \left(0,\liminf_{|x_1|\to\infty} \left(\inf_{(x_2,\dots,x_d)\in\R^{d-1}}\left|\frac{\partial \log \pi(x)}{\partial x_1}\right|\|x\|^\gamma \right)\right)\,,$$
Condition \ref{cond:gap_MALA}(i) implies that we can choose $k$ large enough such that 
$$
\left|\frac{\partial \log \pi(x)}{\partial x_1}\right|\|x\|^\gamma\geq \epsilon
\qquad\hbox{ for all }x\in K\,.
$$
We will now show $\Pr(X^{(t+1)}\in K^c\,|\,X^{(t)}\in K)\leq \Theta(e^{-\lambda^{-\alpha}})$ for some $\alpha>0$ as $\lambda\downarrow 0$.
Note that
\begin{multline*}
\pi(K)\mathbb{P}(X^{(t+1)}\in K^{c}|X^{(t)}\in K)=\mathbb{P}(X^{(t+1)}\in K^{c}|X^{(t)}\in K\cap A_\lambda)\mathbb{P}(X^{(t)}\in K\cap A_\lambda)\\
+\mathbb{P}(X^{(t+1)}\in K^{c}|X^{(t)}\in K\cap A_\lambda^{c})\mathbb{P}(X^{(t)}\in K\cap A_\lambda^{c}),
\end{multline*}
meaning that
\[
\pi(K)\mathbb{P}(X^{(t+1)}\in K^{c}|X^{(t)}\in K)\leq\mathbb{P}(X^{(t+1)}\in K^{c}|X^{(t)}\in K\cap A_\lambda)+\mathbb{P}(X^{(t)}\in K\cap A_\lambda^{c}).
\]
Condition \ref{cond:gap_MALA}(ii) implies $\mathbb{P}(X^{(t)}\in K\cap A_\lambda^{c})\leq\mathbb{P}(X^{(t)}\in A_\lambda^{c})\leq \Theta(e^{-\lambda^{-\beta/(2\tilde{\gamma})}})$.

Also, given $Y=(Y_1,\dots,Y_d)$ with 
$Y|X^{(t)}\sim \tilde{Q}^M_h(X^{(t)},\cdot)$, 
we have
\begin{align*}
\Pr(X^{(t+1)}\in K^c\,|\,X^{(t)}\in K\cap A_\lambda)
\leq&
\Pr(|Y_1|\leq k\,|\,X^{(t)}\in K\cap A_\lambda)\,.
\end{align*}
Denote $\frac{\partial}{\partial x_{1}}\log\pi(x)$ by $\partial_1(x)$ for brevity.
If $X^{(t)}\in K\cap A_\lambda$ we have $|X^{(t)}_1|\leq \lambda^{-1/(2\tilde{\gamma})}\leq \lambda^{-1/2}$ and $$|\partial_1(X^{(t)})|\geq\epsilon\|X^{(t)}\|^{-\gamma}\geq\epsilon \lambda^{\gamma/(2\tilde{\gamma})}\geq \epsilon \lambda^{1/2}\,,$$
which imply
\begin{align*}
|Y_1|
&=|X^{(t)}_1+ \lambda^{-2} \partial_1(X^{(t)})+ \lambda^{-1} \xi_1|
\\&\geq
 \lambda^{-2}|\partial_1(X^{(t)})|-\lambda^{-1}|\xi_1|-|X^{(t)}_1|\\
&\geq
 \lambda^{-3/2} \epsilon -\lambda^{-1}|\xi_1|-\lambda^{-1/2}\,, 
\end{align*}
where $\xi_1\sim N(0,1)$.
It follows that 
\begin{align*}
\Pr(|Y_1|\leq k\,|\,X^{(t)}\in K\cap A_\lambda)
\leq&
\Pr( \lambda^{-3/2} \epsilon -\lambda^{-1}|\xi_1|-\lambda^{-1/2}\leq k\,|\,X^{(t)}\in K\cap A_\lambda)\\
=&
\Pr(|\xi_1|\geq \epsilon  \lambda^{-1/2} -\lambda^{1/2}-k\lambda)\,.
\end{align*}
Since $\Pr(|\xi_1|\geq t)\leq\exp(-t^2/2)$ for every $t>0$ (which follows from standard bounds on Gaussian tails and $\xi_{1}\sim N(0,1)$) and $\epsilon \lambda^{-1/2} -\lambda^{1/2}-k\lambda \geq 2\lambda^{-1/3}$ eventually as $\lambda\downarrow 0$, it follows
\begin{align*}
\Pr(|Y_1|\leq k\,|\,X^{(t)}\in K\cap A_\lambda)
\leq&
 \;\Theta(e^{-\lambda^{-2/3}})
 \qquad \hbox{as }\lambda\downarrow 0\,.
\end{align*}
Combining the inequalities above and noting that $\pi(K)$ does not depend on $\lambda$, it follows that $\mathbb{P}(X^{(t+1)}\in K^{c}|X^{(t)}\in K)\leq 
\Theta(e^{-\lambda^{-\alpha}})$ for $\alpha=\min\{\beta/(2\tilde{\gamma}),2/3\}>0$
as $\lambda\downarrow 0$.
Finally, the conductance bound in \eqref{eq:conductance} imply
$\text{Gap}(\tilde{P}^M_\lambda)\leq \Theta(e^{-\lambda^{-\alpha}})$ as $\lambda\downarrow 0$.
\end{proof}

\begin{proof}[Proof of Theorem \ref{thm:HMC}]
Similarly to the case of the random walk and Langevin schemes, we will study the MH transition kernel $\tilde{P}^H_\lambda$ with proposal $\phi\circ Q^H_\lambda$ and target $\phi\circ\pi^{(\lambda)}$, and exploit the fact that $\text{Gap}(\tilde{P}^H_\lambda)=\text{Gap}(P^H_\lambda)$ by Lemma \ref{lemma:MH_isomorphic} and Lemma \ref{lemma:isomorphic_MC}.
Considering $\phi(x)=\Sigma_\lambda^{1/2}x$ as above we have $\pi=\phi\circ\pi^{(\lambda)}$ and $\tilde{Q}_\lambda^H=\phi\circ Q_\lambda^H$ evolving according to a preconditioned HMC algorithm as follows.
Writing the current point $x \in \R^d$ as $x(0)$, as in Section \ref{sub:HMC} of the the paper, the proposal $y := x(L)\sim \tilde{Q}_\lambda^H(x,\cdot)$ is obtained using the update
\begin{equation}
x(L) = x(0)+\sigma^{2}\left( \frac{L}{2}\Sigma_\lambda\nabla \log\pi(x(0))+\sum_{j=1}^{L-1} (L-j)\Sigma_\lambda\nabla\log\pi\left(x(j)\right)\right) +L\sigma\Sigma_\lambda^{1/2} \xi(0),\label{eq:transition}
\end{equation}
where each $x(j)$ is defined recursively in the same manner, and $\xi(0)\sim N(0,\I_d)$.  
It is easy to check that $\tilde{Q}_\lambda^H=\phi\circ Q_\lambda^H$ using the same calculations as in the proof of Theorem \ref{thm:MALA_general}.

We now prove $\text{Gap}(\tilde{P}^H_\lambda)\leq \Theta(e^{-\lambda^{-\alpha}})$ as $\lambda\downarrow 0$ for some $\alpha>0$.
To simplify the notation in the following we prove the equivalent statement that $\text{Gap}(\tilde{P}^H_{h^{-1}})\leq \Theta(e^{-h^{\alpha}})$ as $h\to\infty$.
Fix $\delta\in(0,(1-q)/2)$ with $q$ defined in Condition \ref{cond:regularity_conditions} and consider the sets $A_h := \{y\in\R^d:|y_1|<k+h\}$ and $K := \{y\in\R^d:|y_1|>k\}$.
Here $k$ is chosen large enough that Lemma \ref{lemma:polynomial_tail} below is satisfied and that $0<\pi(K)<1/2$, which can always be done thanks to the tightness and positiveness of $\pi$ (see above).
Lemma \ref{lemma:polynomial_tail} implies that if $X^{(t)}\in K\cap A_{h}$, $|\xi_{1}|\leq h^{1-\delta}$ and $h\geq h_0$, where $h_0=\lambda_0^{-1}$ with $\lambda_0$ defined as in Lemma \ref{lemma:polynomial_tail}, then $X^{(t+1)}\in K$.
We now upper bound the probability $\mathbb{P}(X^{(t+1)}\in K^{c}|X^{(t)}\in K)$.
First note that
\begin{align*}
\mathbb{P}(X^{(t+1)}\in K^{c},X^{(t)}\in K)=\mathbb{P}(X^{(t+1)}\in K^{c}|X^{(t)}\in K\cap A_{h})\mathbb{P}(X^{(t)}\in K\cap A_{h})\\
+\mathbb{P}(X^{(t+1)}\in K^{c}|X^{(t)}\in K\cap A_{h}^{c})\mathbb{P}(X^{(t)}\in K\cap A_{h}^{c}),
\end{align*}
which implies
\begin{align*}
\mathbb{P}(X^{(t+1)}\in K^{c},X^{(t)}\in K)\leq\, 
&\mathbb{P}(X^{(t+1)}\in K^{c}|X^{(t)}\in K\cap A_{h})
+\mathbb{P}(X^{(t)}\in K\cap A_{h}^{c}).
\end{align*}
Breaking out the first term on the right-hand side gives
\begin{align*}
&\mathbb{P}(X^{(t+1)}\in K^{c}|X^{(t)}\in K\cap A_{h})
\leq  \\
&\mathbb{P}(X^{(t+1)}\in K^{c}|X^{(t)}\in K\cap A_{h},|\xi_{1}|\leq h^{1-\delta}) \mathbb{P}(|\xi_{1}|\leq h^{1-\delta})
+
\mathbb{P}(|\xi_{1}|>h^{1-\delta}),
\end{align*}
which, using the result of Lemma \ref{lemma:polynomial_tail},
reduces to
\[
\mathbb{P}(X^{(t+1)}\in K^{c}|X^{(t)}\in K\cap A_{h})\leq\mathbb{P}(|\xi_{1}|>h^{1-\delta})\,.
\]
Hence we obtain the overall bound
\[
\mathbb{P}(X^{(t+1)}\in K^{c},X^{(t)}\in K)\leq\mathbb{P}(|\xi_{1}|>h^{1-\delta})+\mathbb{P}(X^{(t)}\in K\cap A_{h}^{c}).
\]
Using standard bounds on Gaussian tails and $\xi_{1}\sim N(0,1)$, we have $\mathbb{P}(|\xi_{1}|>h^{1-\delta})\leq\exp(-h^{2(1-\delta)}/2)$.
Also, from Lemma \ref{lemma:density_tails}, we have $\mathbb{P}(X^{(t)}\in K\cap A_{h}^{c})\leq \Theta(e^{-\gamma h^{1+q} - q\log(h)})$ as $h\to\infty$ for some $\gamma \in(0,\infty)$.
Hence, since $\delta < (1-q)/2$ and $\mathbb{P}(X^{(t)}\in K)=\pi(K)$ is constant with respect to $h$, we obtain
\[
\mathbb{P}(X^{(t+1)}\in K^{c}|X^{(t)}\in K)\leq \Theta \left( e^{-\gamma h^{1+q} - q\log(h)} \right)
\qquad
\hbox{ as }h\to\infty\,.
\]
Finally, the conductance bound in \eqref{eq:conductance} gives
\[
\text{Gap}(\tilde{P}^H_{h^{-1}}) \leq \Theta\left( e^{-\gamma h^{1+q} - q\log(h)} \right)
\qquad
\hbox{ as }h\to\infty\,.
\]
\end{proof}

\begin{proof}[Proof of Proposition \ref{prop:h_to_zero}]
For the Langevin case, standard results on the total variation distance between two Gaussian measures with differing means reveals that
\[
\|Q^M_\lambda(x,\cdot) - Q^R(x,\cdot)\|_{TV} = 1+ \frac{1}{\sqrt{2\pi}}\int_0^t e^{-u^2/2}du.
\]
where $t = \sigma |\nabla \log\pi(x/\lambda)|/(4\lambda)$. Because $\nabla\log\pi$ is bounded in a neighbourhood of zero, for large enough $\lambda$ we can write $t \leq C/\lambda$ for some $C<\infty$. Then note that as $\lambda \uparrow \infty$
\begin{align*}
\int_0^{C/\lambda} e^{-u^2/2}du \leq \frac{C}{\lambda}.
\end{align*}
For the Barker case, note that the total variation distance here can be written
\[
\frac{1}{2}\int \mu_\sigma(z)|2g(e^{\nabla\log\pi(x/\lambda)z/\lambda}) - 1|dz,
\]
where $g(t)=1/(1+t^{-1})$. Setting $u:= \nabla\log\pi(x)z$, a Taylor series expansion about $u=0$ of $2g(u)$ is
\[
2g(u) = 1 + \frac{u}{2} + g''(\xi)u^2,
\]
for some $\xi$ satisfying $|\xi| \leq |\nabla\log\pi(x)z|$, using the Lagrange form of the remainder.  Substituting this into the integral and simplifying gives
\begin{align*}
\frac{1}{2}\int \mu_\sigma(z) \left| \frac{u}{2} + g''(\xi)u^2 \right| dz
&\leq 
\frac{1}{4\lambda}|\nabla\log\pi(x/\lambda)| \int |z|\mu_\sigma(z)dz 
\\ 
&\quad + \frac{1}{2\lambda^2} \nabla\log\pi(x/\lambda)^2\int |g''(\xi)|z^2 \mu_\sigma(z)dz.
\end{align*}
For large enough $\lambda$ the boundedness assumption allows us to write $|\nabla\log\pi(x/\lambda)| \leq c$ for some $c < \infty$.  In addition note that $g''(\xi) = 2e^{-2\xi}/(1+e^{-\xi})^3 - e^{-\xi}/(1+e^{-\xi})^2$, and so $\sup_{\xi \in \R}|g''(\xi)| = (6\sqrt{3})^{-1}$. Substituting into the bound and evaluating the two integrals gives an upper bound to the total variation distance of
\[
\left(\frac{c}{4}\sqrt{\frac{2\sigma^2}{\pi}}\right)\lambda^{-1} + \left(\frac{c^2\sigma^2}{12\sqrt{3}}\right)\lambda^{-2},
\]
which is $\Theta(1/\lambda)$ as $\lambda \uparrow \infty$, as desired.

\end{proof}

\subsubsection{Lemmas used to prove Theorem \ref{thm:HMC}}

\begin{lemma}
\label{lemma:polynomial_tail}
Assume Condition \ref{cond:regularity_conditions} and let $\delta\in(0,1)$. 
For every $L\geq1$ there exist $\lambda_0>0$ and a large enough $k$
such that for every $\lambda\leq \lambda_{0},$ $|\xi_{1}|\leq \lambda^{-(1-\delta)}$ and
$|x_{1}(0)|\in[k,k+\lambda^{-1})$ it holds that $|x_{1}(L)|\geq k$, where $x(L)$ is defined in \eqref{eq:transition}.
\end{lemma}

\begin{proof}
Recall that, for each $i\geq 1$, $x_{1}(i)$ is implicitly a function of the starting location $x_1(0)$, the parameter $\lambda$ and the noise $\xi_1$.
For notational convenience, in the following we set $h=\lambda^{-1}$ and study the limit $h\uparrow\infty$.
In order to prove the thesis it is sufficient to show that for fixed, sufficiently large $k>0$ we have
\begin{equation}\label{eq:inf_diverging}
\inf_{\xi_1,x_1(0)}
|x_{1}(L)|
\to\infty \qquad \hbox{ as }h\uparrow\infty\,,    
\end{equation}
where $\xi_1$ and $x_1(0)$ in the infimum are restricted as in the lemma's statement, i.e.\
$\xi_1\in (-h^{1-\delta}, h^{1-\delta})$ and 
$x_1(0)\in (-(k+h),-k]\cup [k,k+h)$.
In order to prove \eqref{eq:inf_diverging} we will show that for all $i\geq 1$, as $h\uparrow\infty$ we have
\begin{align}
\Theta(h^{2\sum_{j=0}^{i-1}q^{j}})
& \leq
\inf|x_{1}(i)|
\leq
\sup|x_{1}(i)|
\leq
\Theta(h^{q^{i}+2\sum_{j=0}^{i-1}q^{j}}) 
,\label{eq:ineq_for_x}
\\
\Theta(h^{2\sum_{j=1}^{i}q^{j}})
&\leq
\inf|\partial_{1}(i)|
\leq
\sup|\partial_{1}(i)|
\leq
\Theta(h^{q^{i+1}+2\sum_{j=1}^{i}q^{j}})\,,\label{eq:ineq_for_grad}
\end{align}
where infima and suprema run over $\xi_1\in (-h^{1-\delta}, h^{1-\delta})$ and
$x_1(0)\in (-(k+h),-k]\cup [k,k+h)$ as in \eqref{eq:inf_diverging}, 
and $\partial_1(i)$ stands for $\partial\log\pi_1/\partial x_1(x_{1}(i))$.
Note that, for any $i\geq 1$, \eqref{eq:ineq_for_grad} is implied by \eqref{eq:ineq_for_x} thanks to
$\inf|x_{1}(1)|\to\infty$ as $h\to\infty$ and \eqref{eq:condition_1}.
Thus it suffices to prove that \eqref{eq:ineq_for_x} holds for all $i\geq 1$, which we will do by induction over $i$.

In the following, $k$ is chosen large enough that 
$c|x_1|^q\leq |\partial\log\pi_1/\partial x_1(x_1)|\leq C |x_1|^q$ for some $0<c\leq C<\infty$ and all $|x_1|>k$, which can be done by \eqref{eq:condition_1}.
Also, unless otherwise stated, we assume $\xi_1\in (-h^{1-\delta}, h^{1-\delta})$ and 
$x_1(0)\in (-(k+h),-k]\cup [k,k+h)$, and all infima and suprema are taken over those sets.

Considering $i=1$, we have $x_{1}(1)=x_{1}(0)+h\xi_{1}+(h^{2}/2)\partial_{1}(0)$, which implies
\[
\frac{h^{2}}{2}|\partial_{1}(0)|-|x_{1}(0)|-h|\xi_{1}|
\leq
|x_{1}(1)|
\leq
\frac{h^{2}}{2}|\partial_{1}(0)|+|x_{1}(0)|+h|\xi_{1}|\,.
\]
Then, since $|\xi_1|\in(0, h^{1-q})$, $|x_1(0)|\in[k, k+h)$ and $ck^q\leq c|x_1(0)|^q\leq |\partial_1(0)| \leq C|x_1(0)|^q\leq C(h+k)^q$,
we have
\[
\Theta(h^{2}) 
=
\frac{h^{2}}{2}ck^q-(k+h)-h^{2-\delta}
\leq
|x_{1}(1)|
\leq
\frac{h^{2}}{2}C(h+k)^q+(k+h)+h^{2-q}
=\Theta(h^{2+q}) 
\]
meaning that \eqref{eq:ineq_for_x} is satisfied for $i=1$.

We then show that if \eqref{eq:ineq_for_x} and \eqref{eq:ineq_for_grad} hold for $i=1,\dots,\ell-1$, where $\ell\geq 2$, then they also hold for $i=\ell$.
First note that when $\ell\geq 2$, \eqref{eq:transition} implies
\begin{equation}\label{eq:leap_1d}
x_{1}(\ell)=x_{1}(\ell-1)+h\xi_{1}+\frac{h^{2}}{2}\partial_{1}(0)+h^{2}\sum_{j=1}^{\ell-1}\partial_{1}(j)\,.
\end{equation}
From \eqref{eq:leap_1d} and $|\xi_1|\in(0, h^{1-q})$, we can deduce that
\begin{equation}\label{eq:double_ineq_leapfrog}
h^{2}|\partial_{1}(\ell-1)|-|x_{1}(\ell-1)|-h^{2}\sum_{j=0}^{\ell-2}|\partial_{1}(j)|
\leq
|x_{1}(\ell)|
\leq
|x_{1}(\ell-1)|+h^{2-q}+h^{2}\sum_{j=0}^{\ell-1}|\partial_{1}(j)|.    
\end{equation}
Combining the lower bound in \eqref{eq:double_ineq_leapfrog} with \eqref{eq:ineq_for_x} and \eqref{eq:ineq_for_grad} for $i=1,\dots,\ell-1$ we obtain
\begin{align*}
\inf|x_{1}(\ell)|&\geq
\inf h^{2}|\partial_{1}(\ell-1)|-
\sup
\left(|x_{1}(\ell-1)|+h^{2}\sum_{j=0}^{\ell-2}|\partial_{1}(j)|\right)
\\&\geq
\Theta(h^{2\sum_{j=0}^{\ell-1}q^{j}})
-\Theta(h^{q^{\ell-1}+2\sum_{j=0}^{\ell-2}q^{j}})
=\Theta(h^{2\sum_{j=0}^{\ell-1}q^{j}})\,,
\end{align*}
where the last equality follows from 
$q^{\ell-1}+2\sum_{j=0}^{\ell-2}q^{j}\leq  2\sum_{j=0}^{\ell-1}q^{j}$.
Thus the lower bound in \eqref{eq:ineq_for_x} holds also for $i=\ell$.
Similarly, combining the upper bound in \eqref{eq:double_ineq_leapfrog} with \eqref{eq:ineq_for_x} and \eqref{eq:ineq_for_grad} for $i=1,\dots,\ell-1$ we obtain
\begin{align*}
\sup|x_{1}(\ell)|
&\leq
\sup 
\left(|x_{1}(\ell-1)|+h^{2-q}+h^{2}\sum_{j=0}^{\ell-1}|\partial_{1}(j)|\right)
\\
&\leq
\Theta(h^{q^{\ell-1}+2\sum_{j=0}^{\ell-2}q^{j}}
+h^{2-q}+
h^{q^{i}+2\sum_{j=0}^{\ell-1}q^{j}})
=\Theta(h^{q^{i}+2\sum_{j=0}^{\ell-1}q^{j}})\,,
\end{align*}
where the last equality follows from 
$2-q\leq q^{\ell-1}+2\sum_{j=0}^{\ell-2}q^{j}\leq q^i+ 2\sum_{j=0}^{\ell-1}q^{j}$.
Thus the upper bound in \eqref{eq:ineq_for_x} holds also for $i=\ell$ and the proof is complete.
\end{proof}

\begin{lemma}
\label{lem:tail_bound}Condition \ref{cond:regularity_conditions}
(ii) implies that there exist $t$, $c$ and $C$ in $(0,\infty)$ such that
\begin{equation}
\pi_1(x_1) \leq Ce^{-c|x_1|^{1+q}}\,,\qquad \hbox{for all }|x_1| \geq t\,.
\end{equation}
\end{lemma}
\begin{proof}
Condition \ref{cond:regularity_conditions} implies that there exists $t,c \in (0,\infty)$ such that 
\[
\left| \frac{d}{dx_1}\log\pi_1(x_1) \right| \geq c|x_1|^{1+q}
\,,\qquad \hbox{for all }|x_1| \geq t\,.
\]
Since $\log\pi_{1}\in C_{1}(\R)$, the above implies that either 
\[
\frac{d}{dx_{1}}\log\pi_{1}(x_{1})> c x_1^{1+q} \quad\text{or}\quad\frac{d}{dx_{1}}\log\pi_{1}(x_{1})<- c x_1^{1+q}\,,
\]
holds for all $|x_1| \geq t$.
Since $\int \pi_1(x_1) dx_1 = 1$ the latter option must be true.
Computing the anti-derivative gives
\[
\log\pi_1(x_1) \leq -c x_1^{1+q} + \log C,
\]
for some constant $\log C$. An analogous argument can be used in the case $x_{1}\downarrow-\infty$,
and the two combined give the result.
\end{proof}

\begin{lemma}\label{lemma:density_tails}
If Condition \ref{cond:regularity_conditions}
holds and $X_1 \sim \pi_1(\cdot)$, then there exists $\gamma \in(0,\infty)$ such that
\[
\mathbb{P}(|X_1|>k+h) \leq \Theta \left( e^{-\gamma h^{1+q} -q\log(h)} \right)\qquad \hbox{ as }h\to\infty.
\]
\end{lemma}
\begin{proof}
Using Lemma \ref{lem:tail_bound}, provided $k+h>t$ we have 
\begin{align*}
\mathbb{P}(|X_1|>k+h)= & \int_{k+h}^{\infty}\pi_{1}(x_{1})dx_{1}+\int_{-\infty}^{-(k+h)}\pi_{1}(x_{1})dx_{1}\\
\leq & 2C\int_{k+h}^{\infty}e^{-cx_{1}^{1+q}}dx_{1} \\
 = &2C \frac{c^{-1/(q+1)}}{q+1}\Gamma \left( \frac{1}{1+q}, c(k+h)^{1+q} \right),
\end{align*}
where $\Gamma (a,b) := \int_b^\infty u^{a-1}e^{-u}du$ is the incomplete Gamma function. 
For the case $q>0$ the upper bound of \cite{gautschi1959some}, which is described on pages 771-772 of  \cite{alzer1997some}, states that for fixed $a\in (0,1)$ and $x>0$ we have
\[
\Gamma (a,x^{-a}) \leq e^{-x^{-a}} \frac{c_a}{a} \left( (x^{-a}+c_a^{-1})^a - x \right),
\]
where $c_a := \Gamma(1+a)^{1/(1-a)}$.  Setting $C_2 := 2C c^{-1/(q+1)}/(q+1)$, $a := 1/(1+q)$ and using this upper bound gives
\[
\mathbb{P}(|X_1|>k+h) \leq e^{-c(k+h)^{1+q} }  C_2 \frac{c_a}{a} \left[ \left(c(k+h)^{\frac{1}{a}}+c_a^{-1} \right)^{a} - c^a(k+h) \right].
\]
We use a Taylor series expansion of $f(c(k+h)^{1/a} + c_a^{-1})$ about $f(c(k+h)^{1/a})$, where $f(x) = x^a$. The terms each have a different power of $h$.  This gives
\[
(c(k+h)^{\frac{1}{a}}+c_a^{-1})^{a} = c^a(k+h) + c_a^{-1} a (c(k+h)^{\frac{1}{a}})^{a-1} + O(h^{(a-2)/a})
\]
Since $a = 1/(1+q) < 1$ then $(a - 1)/a = -q$ and $(a-2)/a = -(1+2q)$, and therefore
\[
(c(k+h)^{\frac{1}{a}}+c_a^{-1})^{a} - c^a(k+h) = \Theta((c(k+h)^{\frac{1}{a}})^{a-1}) = \Theta(h^{-q}).
\]
Combining with the above, we can write that for any fixed $k$ and fixed $q > 0$, there exists $ \gamma \in (0,\infty)$ such that as $h\uparrow \infty$ 
\[
\mathbb{P}(|X_1|>k+h) \leq \Theta \left( e^{-\gamma h^{1+q} -q\log(h)} \right).
\]
In the case $q=0$ the integral $\int_{k+h}^\infty e^{-cx_1}dx_1 = e^{-c(k+h)}/c$ and the result is immediate.
\end{proof}

\subsection{Proofs for Section \ref{sec:barker}}

\begin{proof}[Proof of Proposition
\ref{prop:barker_normalising}]
Setting $y-x = z$, then $t(z) = e^{z \nabla \log\pi(x)}$ and $1/t(z) = e^{-z \nabla\log\pi(x)} = t(-z)$, meaning 
\begin{align*}
Z(x) 
&= \int_{\R} \frac{t(z)}{1+t(z)} \mu_{\sigma}(z)dz 
\\
&= \int_0^\infty \left( \frac{t(z)}{1+t(z)} \mu_{\sigma}(z) + \frac{t(-z)}{1+t(-z)} \mu_{\sigma}(-z) \right) dz.
\end{align*}
Noting that $\mu_{\sigma}(z) = \mu_{\sigma}(-z)$ and $t(-z) = 1/t(z)$ then gives
\[
Z(x) = \int_0^\infty \mu_{\sigma}(z)dz = \frac{1}{2}
\]
which completes the proof.
\end{proof}

\begin{proof}[Proof of Proposition \ref{prop:barker_sample}]
Assume $y = x + b(x,z) \times z$ is generated using Algorithm \ref{alg:barker_1d}.  Then for any $A \in \mathcal{B}(\R)$
\begin{align*}
\mathbb{P}[y \in A] 
&= 
\mathbb{P} \left[ \{z \in A - x \} \cap \{b(x,z) = 1\} \right] + \mathbb{P} \left[ \{-z \in A - x \} \cap \{b(x,z) = -1\} \right].
\end{align*}
Note that the second term on the right-hand side can be re-written 
\[
\mathbb{P} \left[ \{z \in A - x \} \cap \{b(x,-z) = -1\}\right],
\]
owing to the symmetry of $\mu_{\sigma}$. Because of this, we can write
\begin{align*}
\mathbb{P}[y \in A] &= \int_{A-x} \frac{e^{z\nabla\log\pi(x)}}{1+ e^{z\nabla\log\pi(x)}}\mu_{\sigma}(z)dz + \int_{A-x} \frac{1}{1+ e^{-z\nabla\log\pi(x)}}\mu_{\sigma}(z)dz \\
&= 2 \int_{A-x} \frac{e^{z\nabla\log\pi(x)}}{1+ e^{z\nabla\log\pi(x)}}\mu_{\sigma}(z)dz \\
&= Q^B(x,A)\,
\end{align*}
which completes the proof.
\end{proof}

\begin{proof}[Proof of Proposition \ref{prop:barkermod}]

We establish a point-wise bound on the candidate transition densities of the two algorithms.  Combining this with Lemma \ref{lemma:lowerbound} gives an equivalent bound on the spectral gaps.  To reach this point-wise bound, first note that the candidate transition density associated with the Random Walk algorithm is $q^R(x,x+z) = \mu_\sigma(z)$ for any $x,z\in \mathbb{R}^d$. Now, for the modified Barker proposal, the candidate density can be written
\begin{align*}
\check{q}^B(x,x+z) 
&= 
\mu_\sigma(z)\check{p}(x,z) + \mu_\sigma(-z)(1-\check{p}(x,-z)) 
\\
&= 
\mu_\sigma(z) \left( \check{p}(x,z) - \check{p}(x,-z) + 1 \right) \\
&=
2 \check{p}(x,z)\mu_\sigma(z),
\end{align*}
where on the last line we have used that $\check{p}(x,-z) = 1-\check{p}(x,z)$. Noting that $\tilde{p}(x,z) \leq 1$ establishes that $q^R(x,x+z) \geq \check{q}^B(x,x+z)/2$ for any $x,z\in\R^d$, and upon combining this with Lemma \ref{lemma:lowerbound} the result follows.
\end{proof}

\subsection{Proofs for Section \ref{sec:barker_proofs}}
Interestingly, the proof of the lower bound of Theorem \ref{thm:Barker_lambda_to_0} is analogous to the one of Theorem \ref{thm:RW_lambda_to_0},
providing further insight into the similarity between the Barker scheme and random walk in terms of robustness to scales.
\begin{proof}[Proof of Theorem \ref{thm:Barker_lambda_to_0}]
As in the proof of Theorem \ref{thm:RW_lambda_to_0}, we write $Q_\lambda^B$ to denote the Barker candidate kernel targeting $\pi^{(\lambda)}$, and $\tilde{Q}_\lambda^B(x,dy):=\tilde{q}_\lambda^B(x,y)dy$ to denote the isomorphic kernel defined as $\tilde{Q}_\lambda^B=\phi\circ Q_\lambda^B$, where $\phi$ is the same function used in the proof of Theorem \ref{thm:RW_lambda_to_0}. 
Also, we denote by $P^B_\lambda$ and $\tilde{P}_\lambda^B$ the Metropolis-Hastings kernels with candidate kernels $Q^B_\lambda$ and $\tilde{Q}_\lambda^B$, respectively, and target distributions $\pi^{(\lambda)}$ and $\pi$, respectively.

From \eqref{eq:barker_prop} and \eqref{eq:barker_multi}
it follows that
\begin{equation}\label{eq:Barker_prop_het}
\tilde{q}_\lambda^B(x,y)=
2^d\;\frac{\lambda}{\sigma^d}\mu\left( \frac{\delta_\lambda}{\sigma} \right)
\prod_{i=1}^d \big(1+e^{-\partial_i\log\pi(x)(y_i-x_i)}\big)^{-1}\,.
\end{equation}
Here we are using $\mu$ to denote the $d$-dimensional distribution obtained by proposing each coordinate independently as in Section \ref{sec:Barker_multi_d}.
We therefore have
\begin{equation}\label{eq:barker_ratio_proposals}
\frac{\tilde{q}_\lambda^B(x,y)}{\tilde{q}_1^B(x,y)}=
\lambda\frac{\mu( \delta_\lambda/\sigma )}{\mu(\delta/\sigma)}\,,
\end{equation}
which holds after noting that $(1+e^{-\partial_i\log\pi(x)(y_i-x_i)})$ 
does not depend on $\lambda$, and hence cancels in the ratio.  Note that the expression above coincides with the expression for the random walk proposals in \eqref{eq:RW_ratio_proposals}.
Thus, arguing as in the proof of Theorem \ref{thm:RW_lambda_to_0}, we have that $\tilde{q}^B_\lambda(x,y) \geq \lambda \tilde{q}^B(x,y)$ for all $\lambda \leq \lambda_0$ and all $x,y \in \R^d$, where $\lambda_0 \leq 1$ is the value defined in Condition \ref{cond:rwm}.
Combining the latter inequality with Lemma \ref{lemma:lowerbound} and using the isomorphism property between $\tilde{P}_\lambda^B$ and $P^B_\lambda$ given in Lemmas \ref{lemma:isomorphic_MC} and \ref{lemma:MH_isomorphic}, we obtain
$$
\text{Gap}(P_\lambda^B) \geq \lambda \text{Gap}(P^B) = \Theta(\lambda) \qquad \hbox{as }\lambda\downarrow 0\,.
$$
To show that $\text{Gap}(P_\lambda^B) \leq \Theta(\lambda)$, note that 
$\tilde{q}^B_\lambda(x,y)\leq 2^{d}\tilde{q}^R_\lambda(x,y)$ for all $x,y\in\R^d$ by \eqref{eq:Barker_prop_het} and \eqref{eq:rwm}.
Thus, Lemma \ref{lemma:lowerbound} and Theorem \ref{thm:RW_lambda_to_0} give $\text{Gap}(P_\lambda^B)\leq 2^d \text{Gap}(P_\lambda^R)= \Theta(\lambda)$ as $\lambda\downarrow 0$.
\end{proof}

\subsubsection{Proof of Theorem \ref{thm:barker_ergodicity_multi_d}}
The following lemma, which is an extension of Theorem 4.1 of \citep{roberts1996exponential}, provides generic sufficient conditions for the geometric ergodicity of Metropolis--Hastings algorithms.

\begin{lemma}
\label{lemma:M_H_geom_erg_short}
Let $P$ be a $\phi$-irreducible and aperiodic 
Metropolis--Hastings kernel on $\R^d$ with proposal $Q$ such that compact sets are small under $P$.
If there exist a function $V:\R^d\to (0,\infty)$ such that
$\sup_{x\in\R^d} \dfrac{QV(x)}{V(x)}<\infty$
and
\begin{equation}\label{eq:condition_for_ergodicity}
\liminf_{\| x \|\rightarrow +\infty} \, \int_{R^d} q(x,y)\alpha(x,y)\, dy > \limsup_{\| x\|\to\infty} \dfrac{QV(x)}{V(x)}\,,
\end{equation}
then $P$ is $\pi$-a.e.\ geometrically ergodic.
\end{lemma}
\begin{proof}
We show that \eqref{eq:condition_for_ergodicity} implies the following Foster-Lyapunov drift conditions:
$$\sup_{x\in\R^d}\frac{PV(x)}{V(x)}<\infty
\hbox{ and } \limsup_{\|x\|\to\infty}\frac{PV(x)}{V(x)}<1\,,
$$
which imply $\pi$-a.e.\ geometric ergodicity (see e.g.\ Theorem 3.1 and Lemma 3.5 of \cite{jarner2000geometric}).
First note that
\begin{align*}
\dfrac{PV(x)}{V(x)} &= 
\int_{\R^d} \left(\frac{V(y)}{V(x)}\alpha(x,y)+1-\alpha(x,y)\right) q(x,y)dy
\\&\leq
\int_{\R^d} \frac{V(y)}{V(x)}q(x,y)dy+\int_{\R^d} \left(1-\alpha(x,y)\right) q(x,y)dy
\leq\frac{QV(x)}{V(x)}+1
\,,
\end{align*}
which implies $\sup_{x\in\R^d}\frac{PV(x)}{V(x)}\leq \sup_{x\in\R^d}\frac{QV(x)}{V(x)}+1< \infty$.
Also, the inequalities above imply
\begin{equation}\label{eq:PV}
\dfrac{PV(x)}{V(x)}\leq
1-\left(
\int_{\R^d}\alpha(x,y) q(x,y)dy-
\frac{QV(x)}{V(x)}
\right)
\,.
\end{equation}
From \eqref{eq:condition_for_ergodicity} we have 
\begin{align}
0&<
\liminf_{\| x \|\rightarrow +\infty} \, \int_{\R^d} q(x,y)\alpha(x,y)\, dy - \limsup_{\| x\|\to\infty} \dfrac{QV(x)}{V(x)}\nonumber\\
&\leq
\liminf_{\| x \|\rightarrow +\infty} \left( \int_{\R^d} q(x,y)\alpha(x,y)\, dy -\dfrac{QV(x)}{V(x)}\right)\,.\label{eq:great_than_0}
\end{align}
Combining \eqref{eq:PV} and \eqref{eq:great_than_0} we obtain $\limsup_{\|x\|\to\infty}\frac{PV(x)}{V(x)}<1$, as desired.
\end{proof}

We will show that the conditions of Lemma \ref{lemma:M_H_geom_erg_short} are satisfied when considering a Lyapunov function $V_s(x)=\exp(s\|x\|_\infty)$ based on the sup norm, $\|x\|_\infty=\sup_i|x_i|$.

In the following results we denote $\sup_{t>0}g(t)$ by $M$.
We denote the log-target and its derivatives as $U(x)=\log\pi(x)$ and $U_i(x)=\frac{\partial}{\partial x_i}U(x)$, respectively.
Condition \ref{assumption:spherical} implies that $\nabla U(x)=f'(\|x\|) \frac{x}{\|x\|}$ and $U_i(x)=f'(\|x\|) \frac{x_i}{\|x\|}$ for $\|x\|>R$.
Also, we denote the kernel $Q^{(g)}$ in \eqref{eq:balanced_multi} as $Q$ for brevity and its density function as
\begin{equation}
\label{eq:proposal_multi_d}
q(x,y)=
\prod_{i=1}^d 
\dfrac{g(e^{w_iU_i(x)}) \mu_\sigma(w)}{Z_i(x)}
=
\prod_{i=1}^d q_i(w_i;x)
\,,
\end{equation}
where $w_i=y_i-x_i$ and $q_i(w_i;x)=g(e^{w_iU_i(x)}) \mu_\sigma(w)/Z_i(x)$.

First, we provide some simple results on the behaviour of $g$, $Z_i$ and $q_i$ that will be useful afterwards.
\begin{lemma}\label{lemma:g_prop}
Let $g:(0,\infty)\to(0,\infty)$ be bounded, non-decreasing and such that $g(t)=tg(1/t)$ for all $t>0$.
Then  $g(t)\geq g(1)\min\{1,t\}$ and $\frac{g(1)}{2}\leq Z_i(x)\leq M$, where $M=\sup_{t>0}g(t)$.
\end{lemma}
\begin{proof}
If $t\geq 1$ then $g(t)\geq g(1)= g(1)\min\{1,t\}$ by the monotonicity of $g$. 
If $t< 1$ then $g(t)=tg(1/t)\geq t g(1)= g(1)\min\{1,t\}$ by $g(t)=tg(1/t)$ and the monotonicity of $g$.
From $Z_i(x)=\int_{\R} g(e^{w_iU_i(x)}) \mu_\sigma(w)dw$ and $g(t)\leq M$ it follows $Z_i(x)\leq M$.
If $U_i(x)\leq 0$, then $g(e^{w_iU_i(x)})\geq g(1)$ for all $w\leq 0$ and thus $Z_i(x)\geq \int_{-\infty}^0 g(1) \mu_\sigma(w)dw=\frac{g(1)}{2}$. The case $U_i(x)\geq 0$ is analogous.
\end{proof}

\begin{lemma}
\label{lemma:q_conv}
If g is bounded and non-decreasing, then $Z_i(x) \to \frac{M}{2}$ as $U_i(x)\rightarrow-\infty$ or $U_i(x)\rightarrow+\infty$ and for all $w_i\in\R$ it holds
\begin{align*}
q_i(w_i;x)&
\rightarrow 2\mu_{\sigma}(w_i)\,\mathbb{I}_{\left(-\infty,0\right]}(w_i)
\qquad\hbox{ as }U_i(x)\rightarrow-\infty\;\hbox{ and}\\
q_i(w_i;x)&
\rightarrow 2\mu_{\sigma}(w_i)\,\mathbb{I}_{\left[0,+\infty\right)}(w_i)
\qquad\hbox{ as }U_i(x)\rightarrow+\infty
\,.
\end{align*}
\end{lemma}
\begin{proof}
Consider the case $U_i(x)\rightarrow-\infty$.
From $g(t)=t \, g(1/t)\leq t M$ it follows $g(t)\to 0$ as $t\to 0$.
Also, from the boundedness and monotonicity of $g$ it holds $g(t)\to M$ as $t\to \infty$.
Therefore, for all $w_i\in\R$,
\begin{equation}
\label{eqn: g_convergence}
g(\exp(w_iU_i(x)))\rightarrow M\,\mathbb{I}_{\left(-\infty,0\right]}(w_i)
\qquad 
\hbox{as }U_i(x)\rightarrow-\infty\,.
\end{equation}
Thus, from the bounded convergence theorem 
$Z_i(x)\to \int_{-\infty}^0 M\; \mu_{\sigma}(w_i)dw_i = \frac{M}{2}$ as $U_i(x)\rightarrow-\infty$ and, consequently, $q_i(w_i;x)\rightarrow2\mu_{\sigma}(w_i)\,\mathbb{I}_{\left(-\infty,0\right]}(w_i)$ as $U_i(x)\rightarrow-\infty$.
The case $U_i(x)\rightarrow+\infty$ is analogous.
\end{proof}

We now provide two lemmas that will be used to prove the inequality in \eqref{eq:condition_for_ergodicity}.
\begin{lemma}\label{lemma:rhs_barker_d}
Suppose Condition \ref{assumption:spherical} holds.
Let $V_s(x)=\exp(s\|x\|_\infty)$ and $Q$ the kernel with density $q$ as in \eqref{eq:proposal_multi_d}. Then
$$
\inf_{s>0}\limsup_{\| x\|\to\infty} \dfrac{QV_s(x)}{V_s(x)}=0\,.
$$
\end{lemma}
\begin{proof}
Let $x\in\R^d$ and $Y\sim Q(x,\cdot)$.
Since $V_s(y)\leq \sum_{i=1}^d\exp(s |y_i|)$ we have
$$
\mathbb{E}\left[\frac{V_s(Y)}{V_s(x)}\right]
\leq
\sum_{i=1}^d
\mathbb{E}\left[\frac{e^{s|Y_i|}}{e^{s\|x\|_\infty}}\right]\,.
$$
We now bound $\mathbb{E}\left[e^{s(|Y_i|-\|x\|_\infty)}\right]$ differently depending on whether $|x_i|\leq \frac{1}{2}\|x\|_\infty$ or $\frac{1}{2}\|x\|_\infty<|x_i|\leq \|x\|_\infty$.

If $|x_i|\leq \frac{1}{2}\|x\|_\infty$ it follows from the triangle inequality that $|x_i+w|-\|x\|_\infty\leq |x_i|+|w|-\|x\|_\infty\leq |w| -\|x\|_\infty/2$ for any $w\in\R$.
Also, from \eqref{eq:proposal_multi_d} and Lemma \ref{lemma:g_prop} we have $q_i(w_i;x)\leq \frac{2M}{g(1)}\mu_\sigma (w_i)$.
It follows
\begin{equation*}
\mathbb{E}\left[e^{s(|Y_i|-\|x\|_\infty)}\right]
\mathbb{I}\left(|x_i|\leq \frac{\|x\|_\infty}{2}\right)\leq
\frac{2M}{g(1)} e^{-s\|x\|_\infty/2}\int_{\R} e^{s|w|}\mu_\sigma(w)dw\,,
\end{equation*}
and thus
\begin{equation}\label{eq:limsup_case1}
\limsup_{\|x\|\to\infty}\mathbb{E}\left[e^{s(|Y_i|-\|x\|_\infty)}\right]
\mathbb{I}\left(|x_i|\leq \frac{\|x\|_\infty}{2}\right)
=0\,.
\end{equation}

If $\frac{1}{2}\|x\|_\infty<|x_i|\leq \|x\|_\infty$ we have
\begin{multline*}
\mathbb{E}\left[e^{s(|Y_i|-\|x\|_\infty)}\right]\mathbb{I}\left(|x_i|> \frac{\|x\|_\infty}{2}\right)
\leq\\
\mathbb{I}\left(|x_i|> \frac{\|x\|_\infty}{2}\right)
\int_{\R}
e^{s(|x_i+w|-|x_i|)}q_i(w;x)
dw
\,.
\end{multline*}
If $\|x\|\to\infty$ and $|x_i|> \frac{\|x\|_\infty}{2}$ it follows $|x_i|\to\infty$.
Moreover, by Condition \ref{assumption:spherical} and $|x_i|> \frac{\|x\|_\infty}{2}$, we have $U_i(x)\leq \frac{f(\|x\|)}{2}\to-\infty$ as $x_i\to+\infty$ and $U_i(x)\geq - \frac{f(\|x\|)}{2}\to+\infty$ as $x_i\to-\infty$.
Therefore, by Lemma \ref{lemma:q_conv} 
\begin{align*}
\limsup_{\|x\|\to\infty}
\mathbb{I}\left(|x_i|> \frac{\|x\|_\infty}{2}\right)
\int_{\R}
e^{s(|x_i+w|-|x_i|)}q_i(w;x)
dw
\leq
2\int_{-\infty}^0
e^{sw}\mu_\sigma (w)dw
\,.
\end{align*}
Combining the last two displayed equations we get 
\begin{align}\label{eq:limsup_case2}
&\limsup_{\|x\|\to\infty}
\mathbb{E}\left[e^{s(|Y_i|-\|x\|_\infty)}\right]\mathbb{I}\left(|x_i|> \frac{\|x\|_\infty}{2}\right)
\leq 2\int_{-\infty}^0
e^{sw}\mu_\sigma (w)dw\,.
\end{align}
From \eqref{eq:limsup_case1}, \eqref{eq:limsup_case2} and basic properties of the $\limsup$ we get
\begin{align*}
&\limsup_{\|x\|\to\infty}
\mathbb{E}\left[e^{s(|Y_i|-\|x\|_\infty)}\right]
\leq 2\int_{-\infty}^0
e^{sw}\mu_\sigma (w)dw\,.
\end{align*}

Thus
$$
\limsup_{\|x\|\to\infty}
\mathbb{E}\left[\frac{V_s(Y)}{V_s(x)}\right]
\leq
d\left(
\int_{-\infty}^0 e^{sw}2\mu_\sigma(w)dw
\right)
$$
which goes to $0$ as $s\to\infty$.
\end{proof}

\begin{lemma}\label{lemma:lhs_barker_d}
Assume that $\inf_{w\in(-\delta,\delta)}\mu_\sigma(w)>0$ for some $\delta>0$.
Under Condition \ref{assumption:spherical} it holds
\begin{equation}\label{eq:non_zero_acceptance}
\liminf_{\|x\|\to\infty}\int_{\R^d}
q(x,y)\alpha(x,y)dy>0\,.
\end{equation}
\end{lemma}

\begin{proof}
Let $w=y-x$ and $\mu_\sigma(w)=\prod_{i=1}^d\mu_\sigma(w_i)$. Also, denote by
$\alpha(w;x)=\alpha(x,y)$ the MH acceptance rate when moving from $x$ to $y$.
We write $f(w;x)\gtrsim g(w;x)$ if the function $f(w;x)$ is greater or equal than $g(w;x)$ 
up to positive constants independent of $x$ and $w$.
From Lemma \ref{lemma:g_prop} we have $\frac{g(1)}{2}\leq Z_i(x)\leq M$ and thus
\begin{align*}
&q(w;x)\alpha(w;x)\\
&=
\frac{\mu_\sigma(w)}{\prod_{i=1}^dZ_i(x)}
\min\left\{
\prod_{i=1}^d g(e^{w_iU_i(x)})
,
e^{U(x+w)-U(x)}
\prod_{i=1}^d\frac{g(e^{-w_iU_i(x+w)})Z_i(x)}{Z_i(x+w)}
\right\}
\\&\gtrsim\;
\mu_\sigma(w)
\min\left\{
\prod_{i=1}^d g(e^{w_iU_i(x)})
,
e^{U(x+w)-U(x)}
\prod_{i=1}^d g(e^{-w_iU_i(x+w)})
\right\}
\,.
\end{align*}
Then, using $g(t)\geq g(1)\min\{1,t\}$ from Lemma \ref{lemma:g_prop} we obtain
\begin{align*}
&q(w;x)\alpha(w;x)\\
&\gtrsim\;
\mu_\sigma(w)
\min\left\{
\prod_{i=1}^d g(e^{w_iU_i(x)})
,
g(1)^d
e^{U(x+w)-U(x)+\sum_{i=1}^d\min\{-w_iU_i(x+w),0\}}
\right\}
\\
&\gtrsim\;
\mu_\sigma(w)
\min\left\{
\prod_{i=1}^d g(e^{w_iU_i(x)})
,
e^{U(x+w)-U(x)+\sum_{i=1}^d\min\{-w_iU_i(x+w),0\}}
\right\}
\,.
\end{align*}
Assume $\|x\|$ large and $w\in A(x)$, where $A(x)=\{w\in\R^d\,:\, \|x+w\|\leq \|x\|-\epsilon,\;\|w\|\leq 2\epsilon \hbox{ and }x_iw_i\leq 0 \hbox{ for all } i\}$ for some fixed $\epsilon>0$.
From $x_iw_i\leq 0$ it follows $w_iU_i(x)\geq 0$ and thus, from the monotonicity of $g$, $g(e^{w_iU_i(x)})\geq g(1)$. 
Combining the latter with the last displayed equation we have
\begin{align}\label{eq:gtrapprox}
q(w;x)\alpha(w;x)
\gtrsim&\;
\mu_\sigma(w)
\min\left\{
g(1)^d
,
e^{U(x+w)-U(x)+\sum_{i=1}^d\min\{-w_iU_i(x+w),0\}}
\right\}
\,.
\end{align}
We now lower bound $U(x+w)-U(x)+\sum_{i=1}^d\min\{-w_iU_i(x+w),0\}$.
For $\|x\|>R$, from Condition \ref{assumption:spherical}
\begin{align*}
&U(x+w)-U(x)+\sum_{i=1}^d\min\{-w_iU_i(x+w),0\}
\\&=
f(\|x+w\|)-f(\|x\|)+\frac{f'(\|x+w\|)}{\|x+w\|}\sum_{i=1}^d\min\{-w_i(x_i+w_i),0\}
\,.
\end{align*}
Using the non-increasingness of $f'$ and $w\in A(x)$ we have
$f(\|x+w\|)-f(\|x\|)
\geq
-f'(\|x+w\|)(\|x\|-\|x+w\|)
\geq
-f'(\|x+w\|)\epsilon$.
Thus 
\begin{align*}
&U(x+w)-U(x)+\sum_{i=1}^d\min\{-w_iU_i(x+w),0\}
\\&\geq
-f'(\|x+w\|)\left(\epsilon+\frac{\sum_{i=1}^d\min\{-w_i(x_i+w_i),0\}}{\|x+w\|}\right)\,.
\end{align*}
Since $w\in A(x)$ it follows $x_iw_i\leq 0$ and $\min\{-w_i(x_i+w_i),0\}\geq -w_i^2\geq -(2\epsilon)^2$. Thus
\begin{align*}
&
\inf_{w\in A(x)}\frac{\sum_{i=1}^d\min\{-w_i(x_i+w_i),0\}}{\|x+w\|}
\geq
-\frac{(2\epsilon)^2}{\|x\|-2\epsilon}\,,
\end{align*}
which goes to 0 as $\|x\|\to\infty$.
It follows that
\begin{align*}
&\liminf_{\|x\|\to\infty}\inf_{w\in A(x)}
U(x+w)-U(x)+\sum_{i=1}^d\min\{-w_iU_i(x+w),0\}
\geq\\
&\liminf_{\|x\|\to\infty} -f'(\|x\|-2\epsilon)\epsilon
=\infty\,.
\end{align*}
Combining the last displayed equation with \eqref{eq:gtrapprox} we have
\begin{align*}
\liminf_{\|x\|\to\infty}\inf_{w\in A(x)}q(w;x)\alpha(w;x)
\gtrsim&\;
\liminf_{\|x\|\to\infty}\inf_{w\in A(x)}
\mu_\sigma(w)
>0
\,,
\end{align*}
where the last inequality holds for sufficiently small $\epsilon$ because of the assumption $\inf_{w\in(-\delta,\delta)}\mu_\sigma(w)>0$.
Therefore
\begin{align*}
\liminf_{\|x\|\to\infty}\int_{\R^d}
q(x,y)\alpha(x,y)dy
&\geq
\liminf_{\|x\|\to\infty}\int_{A(x)}
q(w;x)\alpha(w;x)dw\\
&\gtrsim
\liminf_{\|x\|\to\infty}\int_{A(x)}1\,dw
\,.
\end{align*}
The proof is completed noting that 
$\liminf_{\|x\|\to\infty}\int_{A(x)}1\,dw>0$ by the construction of $A(x)$.
\end{proof}

\begin{proof}[Proof of Theorem \ref{thm:barker_ergodicity_multi_d}]
Lemmas \ref{lemma:rhs_barker_d} and \ref{lemma:lhs_barker_d} imply that there exist an $s>0$ such that $V_s$ satisfy \eqref{eq:condition_for_ergodicity}.
The thesis then follows from Lemma \ref{lemma:M_H_geom_erg_short}, 
noting that compact sets are small for $P$ (which can be deduced from the fact that $\inf\pi(x)>0$ on compact sets) and that that $\sup_{x} QV_s(x)/V_s(x)<\infty$ because
$$
\frac{QV_s(x)}{V_s(x)}
\leq
\sum_{i=1}^d\int_{\R}e^{s|w_i|}q_i(w_i;x)dw_i
\leq
2d\int_{\R}e^{s|w_i|}\mu_\sigma(w_i)dw_i<\infty
$$
where we used 
$e^{\|y\|_\infty-\|x\|_\infty}\leq e^{\|y-x\|_\infty}\leq\sum_i e^{|y_i-x_i|}$,
$q_i(w_i;x)\leq 2\mu_\sigma(w_i)$ and $\int_{\R}\exp(s|w|)\mu_\sigma(w)dw\leq 2\int_{\R}\exp(sw)\mu_\sigma(w)dw<\infty$ for every $s>0$.
\end{proof}

\subsection{Proof of Proposition \ref{prop:scaling_heuristic}}
Proposition \ref{prop:scaling_heuristic} follows directly from Lemmas \ref{lemma:ar_expansion} and \ref{lemma:Z_expansion} below.

\begin{lemma}\label{lemma:ar_expansion}
Under the assumptions of Proposition \ref{prop:scaling_heuristic} we have
\begin{align}\label{eq:ar_expansion}
\log\left(
\frac{f(x_i+\sigma u_i)}{f(x_i)}
\frac{
g\big( e^{-\phi'(x_i+\sigma u_i)\sigma u_i} \big)
}{
g\big( e^{\phi'(x_i)\sigma u_i} \big)
}
\right)
&=
\mathcal{O}\left(\sigma^3\right)
&\hbox{as }\sigma\to 0
\,,
\end{align}
for all $x_i,w_i\in\R$.
\end{lemma}
\begin{proof}
Define the function $b$ as $b(s)=\log(g(\exp(s)))$ for all $s\in\R$.
For any $x_i,u_i$ in $\R$, we have 
\begin{align}
&\log\left(
\frac{f(x_i+\sigma u_i)}{f(x_i)}
\frac{
g\big( e^{-\phi'(x_i+\sigma u_i)\sigma u_i} \big)
}{
g\big( e^{\phi'(x_i)\sigma u_i} \big)
}
\right)
\nonumber\\
&=
\phi(x_i+\sigma u_i)-\phi(x_i)+b(-\phi'(x_i+\sigma u_i)\sigma u_i)-b(\phi'(x_i)\sigma u_i)
\nonumber\\&=
c_1(x_i)\phi'(x_i)u_i\sigma +
c_2(x_i)\frac{u_i^2\sigma^2}{2}+
c_3(x_i)
\frac{u_i^3\sigma^3}{6}+
\mathcal{O}\left(\sigma^4\right)\quad\hbox{as }\sigma\to 0
\,,\label{eq:taylor_exp_2}
\end{align}
where $c_1(x_i)$ and $c_2(x_i)$ are the coefficients of the second order Taylor expansion about $\sigma= 0$, and are given by
$c_1(x_i)=\left(1-2b'(0)\right)\phi'(x_i)$ and $c_2(x_i)=\left(1-2b'(0)\right)\phi''(x_i)$.
To conclude, we now show that the assumptions on $g$ imply $b'(0)=1/2$ and $c_1(x_i)=c_2(x_i)=0$.
By definition of $b$ it holds that $b'(0)=g'(1)/g(1)$. 
From $g(t)=t\,g(1/t)$ it follows $g(1+\epsilon)=(1+\epsilon)\,g((1+\epsilon)^{-1})$ and thus $\frac{g(1+\epsilon)-g((1+\epsilon)^{-1})}{2\epsilon}=\frac{g((1+\epsilon)^{-1})}{2}\,$.
Taking the limit $\epsilon \downarrow 0$ and using $(1+\epsilon)^{-1}=1-\epsilon+\mathcal{O}(\epsilon^2)$ it follows that $g'(1)=\frac{g(1)}{2}$ and thus $b'(0)=1/2$ and $c_1(x_i)=c_2(x_i)=0$.
Combining the latter with \eqref{eq:taylor_exp_2} we obtain \eqref{eq:ar_expansion}.
\end{proof}
\begin{remark}\label{rmk:third}
For general $\phi$, $x_i$ and $u_i$, we have $\log(\alpha_i(x_i,x_i+\sigma u_i))=\Theta(\sigma^3)$ because the third coefficient in the Taylor expansion in \eqref{eq:taylor_exp_2}, which is given by
$$
c_3(x_i)=
6 b''(0) \phi'(x_i) \phi''(x_i)
-2 b'''(0)  \phi'(x_i)^3+
(1-3 b'(0))\phi'''(x_i)\,,
$$
is non-zero in general.
\end{remark}

\begin{lemma}\label{lemma:Z_expansion}
Under the assumptions of Proposition \ref{prop:scaling_heuristic} we have
\begin{align*}
\log\left(\frac{Z_i(x_i)}{Z_i(x_i+\sigma u_i)}\right)
&=
\mathcal{O}\left(\sigma^3\right)
&\hbox{ as }\sigma\to 0\,,
\end{align*}
for all $x_i,w_i\in\R$.
\end{lemma}
\begin{proof}
Without loss of generality, assume $g(1)=1$ throughout the proof.
First consider $\log\left(Z_i(x_i)\right)$, which can be written as
\begin{equation}\label{eq:change_of_var}
Z_i(x_i)
=
\int_\R
g\left(e^{\phi'(x_i)(y_i-x_i)}\right)\sigma^{-1}\mu\left(\frac{y_i-x_i}{\sigma}\right)dy_i
=
\int_\R
g\left(e^{\phi'(x_i)\sigma s}\right)\mu(s)ds\,.
\end{equation}
For every non-negative integer $j$, denote by $\mom{j}$ the $j$-th moment of the distribution $\mu(\cdot)$.
Note that, since $\mu$ is a symmetric pdf, $\mom{0}=1$, $\mom{j}=0$ if $j$ is odd and $\mom{j}>0$ if $j$ is even.
For $j\in\{1,2,3\}$, we have
\begin{multline}\label{eq:Z_deriv}
\frac{\partial^j}{\partial\sigma^j}Z_i(x_i)\at{\sigma=0}
=
\int_\R
\frac{\partial^j}{\partial\sigma^j}g\left(e^{\phi'(x_i)\sigma s}\right)\at{\sigma=0}\mu(s)ds
=\\
\int_\R
\frac{\partial^j}{\partial\sigma^j}g\left(e^{\phi'(x_i)\sigma}\right)\at{\sigma=0}s^j\mu(s)ds
=
\frac{\partial^j}{\partial\sigma^j}g\left(e^{\phi'(x_i)\sigma}\right)\at{\sigma=0}\,\mom{j}\,,
\end{multline}
where the exchange of integration and derivation is justified by the assumptions on $g$ and $\mu$.
Using the Taylor expansion of the function $\sigma\mapsto \log(h(\sigma))$ for general $h$ about $\sigma=0$, and the fact that $Z_i(x_i)\at{\sigma=0}=1$ and $\frac{\partial^j}{\partial\sigma^j}Z_i(x_i)\at{\sigma=0}=0$ if $j$ is odd, we have
\begin{align}
\log(Z_i(x_i))
&=
\mom{2}\frac{\partial^2}{\partial\sigma^2}g\left(e^{\phi'(x_i)\sigma}\right)\at{\sigma=0}\frac{\sigma^2}{2}+\mathcal{O}(\sigma^4)
\nonumber\\&=
\mom{2}(g'(1)+g''(1))\phi'(x_i)^2\frac{\sigma^2}{2}+\mathcal{O}(\sigma^4)
\quad \hbox{ as }\sigma\to 0\,. \label{eq:expansion_Z_i_sigma}
\end{align}

Set $y_i=x_i+\sigma u_i$, then from \eqref{eq:change_of_var}
and \eqref{eq:Z_deriv}
\begin{equation*}
\frac{\partial^j}{\partial\sigma^j}Z_i(x_i+\sigma u_i)\at{\sigma=0}
=
\int_\R
\frac{\partial^j}{\partial\sigma^j}
g\left(e^{\phi'(x_i+\sigma u_i)\sigma s}\right)\at{\sigma=0}\mu(s)ds
\end{equation*}
Reordering the Taylor expansion of $g\left(e^{\phi'(x_i+\sigma u_i)\sigma s}\right)$ about $\sigma=0$ as a polynomial of $s$ and keeping only even powers in $s$ we get
\begin{align*}
Z_i(x_i+\sigma u_i)
=
1+
\mom{2}(g'(1)+g''(1))\phi'(x_i)^2\frac{\sigma^2}{2}
+\mathcal{O}(\sigma^3)\,.
\end{align*}
Using the expansion of $\log(h(\sigma))$ for general $h$ about $\sigma=0$, and the fact that $Z_i(x_i+\sigma u_i)\at{\sigma=0}=1$ and $\frac{\partial}{\partial\sigma}Z_i(x_i+\sigma u_i)\at{\sigma=0}=0$, we have
\begin{align*}
\log(Z_i(x_i+\sigma u_i))
=
\mom{2}(g'(1)+g''(1))\phi'(x_i)^2\frac{\sigma^2}{2}
+\mathcal{O}(\sigma^3)\,.
\end{align*}
Combining the latter equation with \eqref{eq:expansion_Z_i_sigma} we have
\begin{equation*}
\log\left(\frac{Z_i(x_i)}{Z_i(x_i+\sigma u_i)}\right)
=
\log\left(Z_i(x_i)\right)-\log\left(Z_i(x_i+\sigma u_i)\right)
=
\mathcal{O}(\sigma^3)
\end{equation*}

\end{proof}
\begin{remark}
For the Barker proposal, the normalization term  $Z_i(x_i)$ is constant over $x_i$ and thus Lemma \ref{lemma:Z_expansion} is trivially satisfied.
\end{remark}

\section{Condition \ref{cond:gap_MALA} for the exponential family class}

\begin{proposition}
Condition \ref{cond:gap_MALA} holds in the case in which there are $\alpha,\beta>0$ such that
\[
\pi(x) \propto \exp\{  - \alpha \|x\|^\beta \}
\]
\end{proposition}

\begin{proof}
Condition (ii) is immediate.  For (i), first note that here
\begin{equation} \label{eqn:efex}
\left| \frac{\partial \log\pi(x)}{\partial x_1} \right|\|x\|^\gamma = -\alpha \beta x_1\|x\|^{\gamma + \beta - 2}.
\end{equation}
Note that $\|x\| = \surd (\sum_i x_i^2)$ is a monotonically increasing function in each $|x_i|$, so the infimum over $(x_2,...,x_d)$ of \eqref{eqn:efex} is realised at $x_2 = ... = x_d = 0$.  Choosing $\gamma = 2$ condition (i) is satisfied because
\[
\liminf_{|x_1| \to \infty} \alpha\beta|x_1|^{1+\beta}=\infty.
\]
\end{proof}

\section{First-order exact Metropolis-Hastings proposals}
Intuitively, we would like any method that uses gradient information to be exact at the first order.
In a Metropolis-Hastings context, this means a proposal that incorporates gradient information should be reversible with respect to measures that possess a log-linear density function, i.e.\ $\pi(x)=\exp(ax+b)$ for some $a,b\in\R$.
In such cases the gradient at any location encompasses full information and this would therefore seem to be a sensible minimal goal for well-designed gradient-based methods.
The Langevin and Hamiltonian schemes both satisfy this stipulation.
As the following proposition shows, for any instance of the class defined in \eqref{eq:gradient_based}, the condition $g(t)=tg(1/t)$ is both sufficient and necessary for the proposal distribution to satisfy such a requirement.

\begin{proposition}\label{prop:first_order_exact}
Let $\mu_\sigma$ be a symmetric probability density function on $\R$ and $\pi(x)=\exp(ax+b)$ for some $a,b\in\R$, with $a\neq 0$. Then a transition kernel of the form in \eqref{eq:gradient_based} is $\pi$-reversible if and only if $g(t)=tg(1/t)$ for every $t>0$.
\end{proposition}

\begin{proof}
Since $\nabla\log\pi(x)=a$ for every $x\in\R$, it follows that the normalizing constant, $Z$, of $q(x,y)$ in \eqref{eq:gradient_based} is independent of $x$.
First we show that $g(t)=tg(1/t)$ implies reversibility.
From the symmetry of $\mu_{\sigma}$ and $g(t)=tg(1/t)$ it follows 
\begin{align*}
\pi(x)&q(x,y)= 
\exp(ax+b)Z^{-1}g\left(\exp\left(a(y-x)\right)\right)\mu_{\sigma}(x,y)\\
&=
\exp(ax+b)Z^{-1}\exp\left(a(y-x)\right)g\left(\exp\left(-a(y-x)\right)\right)\mu_{\sigma}(y,x)\\
&=
\pi(y)q(y,x)\,,
\end{align*}
which implies that $q$ is $\pi$-reversible.
Conversely, if $q$ is $\pi$-reversible, then  
\begin{align*}
1=\frac{\pi(x)q(x,y)}{\pi(y)q(y,x)}&
=
\frac{\exp(a(x-y))g\left(1/\exp\left(a(x-y)\right)\right)}{g\left(\exp\left(a(x-y)\right)\right)}
=\frac{tg(1/t)}{g(t)}
\,,
\end{align*}
for $t=\exp(a(x-y))$. For $a\neq 0$, $\exp(a(x-y))$ takes any positive value as $x,y\in \R^d$ and thus we have $g(t)=tg(1/t)$ for every $t>0$.
\end{proof}

\begin{remark}
Note that $\pi(x)=\exp(ax+b)$ is an improper density function because $\int_{\R}\exp(ax+b)dx=\infty$ for any choice of $a$ and $b$.
This, however, does not pose any issue in defining $\pi$-reversible kernels as usual. \comment{: a transition kernel $P$ is $\pi$-reversible if $\pi(dx)P(x,dy)=\pi(dy)P(y,dx)$ as measures on $\R^d\times\R^d$.}
\end{remark}

\section{Locally balanced proposals and skew-symmetric distributions}
In this section we show that the only balancing function $g$ leading to a skew-symmetric distribution is $g(t) = t/(1+t)$.
Following \citep{azzalini2013skew}, a skew-symmetric distribution on $\mathbb{R}$ is distribution for which the probability density can be written
\[
f(z) = 2f_0(z)G(z),
\]
for any $z \in \mathbb{R}$, where $f_0(z)=f_0(-z)$, $G(z)\geq 0$ and
\begin{equation} \label{eq:skew-sym}
G(z) + G(-z) = 1.
\end{equation}
In the first-order locally-balanced framework, if the current point is $x$ then the proposal has density
\[
f_x(z) = Z(x)^{-1}\mu_\sigma(z)g(e^{\nabla\log\pi(x)z}),
\]
where, setting $t=e^{\nabla\log\pi(x)z}$ the balancing function $g$ satisfies
\begin{equation} \label{eq:balance_supp}
    g(t) = tg(1/t).
\end{equation}
Equating \eqref{eq:skew-sym} and \eqref{eq:balance_supp} gives $G(z) = g(e^{\nabla\log\pi(x)z})=g(t)$, implying that in this case
\[
G(-z) = g(1/t).
\]
Therefore, dividing \eqref{eq:skew-sym} by $G(1/z)$, using the above and combining with \eqref{eq:balance_supp} gives
\[
t + 1 = \frac{1}{g(1/t)},
\]
and combining with \eqref{eq:balance_supp} gives
\[
g(t) = \frac{t}{1+t}.
\]
as required.

\section{Pre-conditioning the Barker proposal}\label{sec:barker_prec}
The diagonal non-isotropic version of the Barker scheme (corresponding to using a diagonal preconditioning matrix) is a simple variation of Algorithm \ref{alg:barkerMH} from the paper and is described in Algorithm \ref{alg:barker_diag}. 
The acceptance probability related to Algorithm \ref{alg:barker_diag} is exactly the same $\alpha^B(x,y)$ defined in \eqref{eq:Bark_accept_diag}.
\begin{algorithm}
\caption{Diagonal Barker proposal on $\R^d$}
\label{alg:barker_diag}
\raggedright \textbf{Require:} current point $x \in \R^d$ and local scales $(\sigma_1,\dots,\sigma_d)\in(0,\infty)^d$

Independently for each $i\in\{1,\dots,d\}$ do:
\begin{enumerate}
\item Draw $z_i\sim \mu_{\sigma_i}$
\item Calculate $p_i(x,z_i) = 1/(1+e^{-z_i\partial_i\log\pi(x)})$ 
\item Set $b_i(x,z_i) = 1$ with probability $p_i(x,z_i)$, and $b_i(x,z_i) = -1$ otherwise
\item Set $y_i = x_i + b_i(x,z_i) \times z_i$
\end{enumerate}
\textbf{Output:} the resulting proposal $y$.
\end{algorithm}

The general pre-conditioned version of the Barker algorithm is obtained by defining an appropriate linear transformation to the target variables $x$ and then applying the standard Barker algorithm (Algorithm \ref{alg:barkerMH} from the paper) in the transformed space.
More precisely, given a target $\pi$ and a covariance matrix $\Sigma$ with Cholesky factor $C$, define the transformed variables $\tilde{x}=(C^T)^{-1}x$ with distribution $\tilde{\pi}(\tilde{x})\propto \pi(C^T \tilde{x})$ and log-gradient $\nabla\log\tilde{\pi}(\tilde{x})=\nabla\log\pi(C^T \tilde{x}) C^T$.
One then applies the standard (isotropic) Barker scheme described in Algorithm \ref{alg:barkerMH} to the pre-conditioned target $\tilde{\pi}$.
As typically done with pre-conditioned MALA, the resulting preconditioned Barker scheme can be implemented without explicitly defining the auxiliary variables $\tilde{x}$ and transformed target $\tilde{\pi}$, but rather keeping the original target $\pi$ and modifying the proposal distribution. 
The resulting pre-conditioned Barker proposal distribution and corresponding Metropolis-Hastings scheme are described in Algorithms \ref{alg:barker_prec} and \ref{alg:MH_barker_prec}, respectively.

\begin{algorithm}
\caption{Preconditioned Barker proposal on $\R^d$}
\label{alg:barker_prec}
\raggedright \textbf{Require:} current point $x \in \R^d$ and preconditioning matrix $C=chol(\Sigma)$.
\begin{enumerate}
\item Draw $z_i\sim \mu$ independently for each $i\in\{1,\dots,d\}$
\item Calculate $p_i(x,z) = 1/(1+e^{-z_ic_i(x)})$ where $c_i(x)=(\nabla\log\pi(x)\cdot C^T)_i$ 
\item For each $i$, set $\tilde{z}_i = z_i$ with probability $p_i(x,z)$, and $\tilde{z}_i = -z_i$ otherwise
\item Set $y = x + C^T \tilde{z}$ where $\tilde{z}=(\tilde{z}_1,\dots,\tilde{z}_d)$
\end{enumerate}
\textbf{Output:} the resulting proposal $y$.
\end{algorithm}

\begin{algorithm}
\caption{Metropolis--Hastings with preconditioned Barker proposal}
\label{alg:MH_barker_prec}
\raggedright \textbf{Require:} starting point for the chain $x^{(0)} \in \R^d$, and preconditioning matrix $C=chol(\Sigma)$.

Set $t=0$ and do the following:
\begin{enumerate}
\item Given $x^{(t)}=x$, draw $y$ using Algorithm \ref{alg:barker_prec} and compute
\[
\alpha^B(x,y) = \min\left( 1, \frac{\pi(y)}{\pi(x)}\times 
\prod_{i=1}^d \frac{1+e^{-z_ic_i(x)}}
{1+e^{z_i c_i(y)}} \right).
\]
where $z_i=((C^T)^{-1}(y-x))_i$ and $c_i(x)=(\nabla\log\pi(x)\cdot C^T)_i$
\item Set $x^{(t+1)} = y$ with probability $\alpha^B(x,y)$, and $x^{(t+1)} = x$ otherwise 
\item If $t+1<N$, set $t \leftarrow t+1$ and return to step 1, otherwise stop.
\end{enumerate}
\textbf{Output:} the Markov chain $\{x^{(0)},\dots,x^{(N)}\}$.
\end{algorithm}

\section{Additional simulation studies}
In this section we provide various additional details on the simulation studies presented in the paper.

\subsection{Additional example for Section \ref{sec:sim_tuning}}
In Figure \ref{fig:20d} we display a phenomenon analogous to Figure \ref{fig:1d} on a 20-dimensional example in which each component of $\pi(\cdot)$ is an independent $N(0,\eta_i^2)$ random variable, with $\eta_1=0.01$ and $\eta_i=1$ for $i=2,...,20$.  
Here the performance of MALA starts deteriorating drastically as soon as the step-size exceeds the scale of the first component as we would expect from the theory developed in Section \ref{sec:hetero}.
On the other hand both the random walk and Barker schemes can function adequately with larger than optimal step-sizes, and as a result achieve a much higher expected squared jump distance on all the other coordinates. 
\begin{figure}[h!]
\centering
\includegraphics[width=0.9\linewidth]{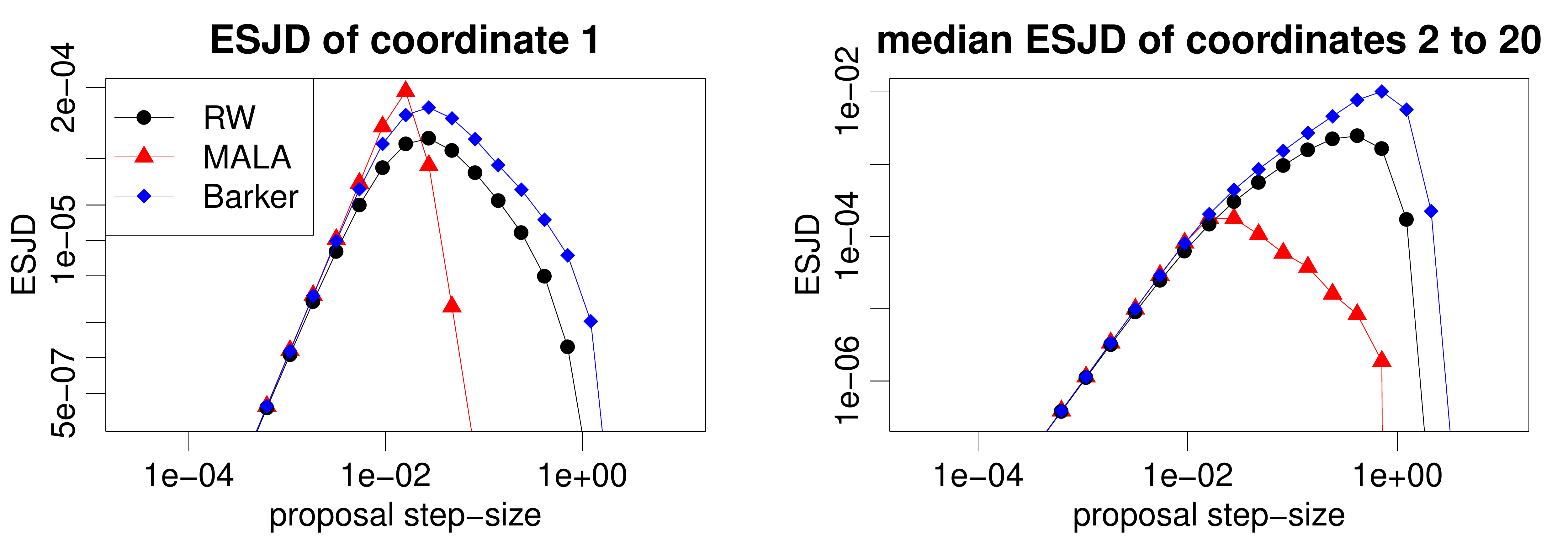}
\caption{Expected squared jump distance (ESJD) against proposal step-size for RW, MALA and Barker on a 20-dimensional target in which one component has a smaller scale than all others.}
\label{fig:20d}
\end{figure}

\subsection{Traceplots for Scenarios 2-4 from Section \ref{sec:diag_adapt}}\label{sec:traceplots}
Figure \ref{fig:k=1_learn06} displays the evolution of tuning parameters and MCMC trajectories when targeting the distribution described in Scenario 1 of that section.
\begin{figure}[h!]
\centering
\includegraphics[width=\linewidth]{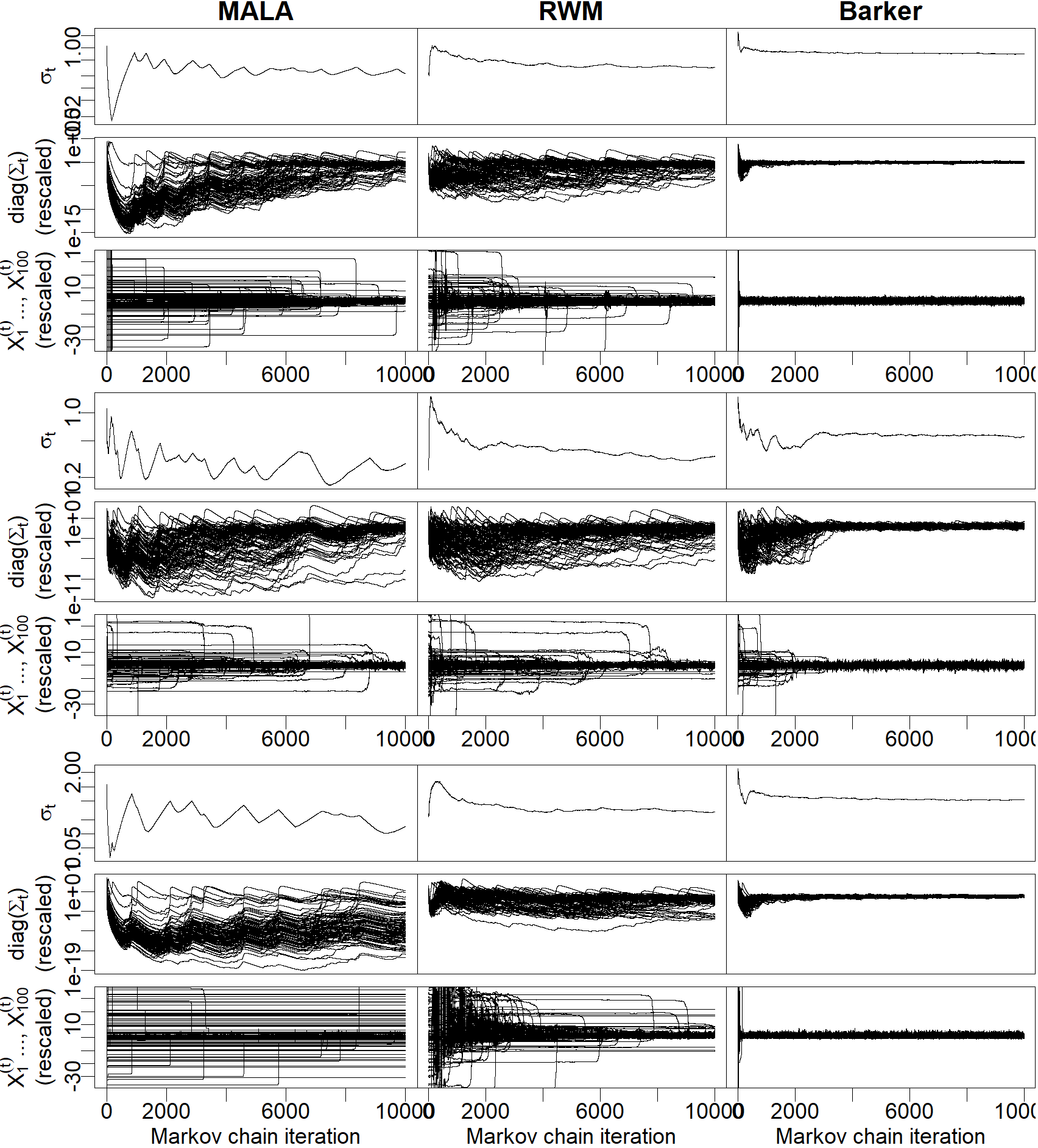}\caption{Same as Figure \ref{fig:k=1_learn06} for the target distributions of Scenario 2 (rows 1 to 3), Scenario 3 (rows 4 to 6) and Scenario 4 (rows 7 to 9).
For each scenario, the first row displays the traceplot of the global scale $\sigma_t$; the second row the ones of the \emph{normalized} local scales $\Sigma_{t,ii}/\Sigma_{ii}$ for $i=1,\dots,100$; and the third row the ones of the \emph{normalized} coordinates $X^{(t)}_i/\Sigma_{ii}^{1/2}$ for $i=1,\dots,100$.
See Section \ref{sec:diag_adapt} for more details.
}
\label{fig:scenarios_234}
\end{figure}
Here we provide analogous illustrations for Scenarios 2, 3 and 4.
Figure \ref{fig:scenarios_234} displays, for each scenario and each algorithm, the traceplot of the global scale $(\sigma_t)_{t\geq 1}$, the ones of the normalized local scales $(\Sigma_{t,ii}/\Sigma_{ii})_{t\geq 1}$ for $i\in\{1,\dots,100\}$ and the ones of the normalized Markov chains coordinates $(X^{(t)}_i/\Sigma_{ii}^{1/2})_{t\geq 1}$ for $i\in\{1,\dots,100\}$.
Here $\Sigma$ is the covariance of the target distribution and normalization is used to facilitate readability, so that all normalized local scales converge to $1$  as $t\to \infty$ and all normalized coordinates have a $N(0,1)$ limiting distribution as $t\to \infty$.
Overall, the traceplots for Scenarios 2-4 display a qualitatively similar behaviour to the ones of Scenario 1 in Figure \ref{fig:k=1_learn06}. 
See Section \ref{sec:diag_adapt} for more discussion.

\subsection{Comparison to truncated or tamed gradients}
Consider Metropolis--Hastings proposals of the form
$$
y  = x + \frac{\sigma^2}{2}G(x) + \sigma \xi\,,
$$
for some $\sigma>0$, $G:\R^d\to\R^d$ and $\xi \sim N(0,I)$.
Setting $G(x)=\nabla \log\pi(x)$ leads to the MALA proposal.
A common way to improve the stability of MALA in the literature is to truncate or tame the gradient $\nabla \log\pi(x)$. 
For example, in the truncated MALA algorithm (MALTA) \citep{atchade2006adaptive} we have
$$
G(x)  = \frac{\delta}{\max\left\{\delta,\|\nabla \log\pi(x)\|\right\}}\nabla \log\pi(x)\,,
$$
for some $\delta>0$,
while in the component-wise tamed MALA (MALTAc) \citep[eq.(4)]{brosse2018tamed} the function $G(x)=(G_1(x),\dots,G_d(x))$ is defined component-wise as
$$
G_i(x)  = \frac{\partial_i(x)}{1+\sigma^2 |\partial_i(x)|}\,.
$$
The above taming is defined in such a way that $|G_i(x)|$ converges to $\sigma^{-2}$ as $|\partial_i(x)|\to\infty$, meaning that in this case the upper bound for tamed gradients is authomatically chosen in a way that depends on $\sigma$.

These schemes are effective in achieving geometric ergodicity also for light tails \citep{atchade2006adaptive}. 
They are less effective, however, in terms of being robust to tuning and they are very sensitive to the choice of truncation parameter (respectively $\delta$ and $\sigma^{-2}$).
\begin{figure}[h!]
\centering
\includegraphics[width=\linewidth]{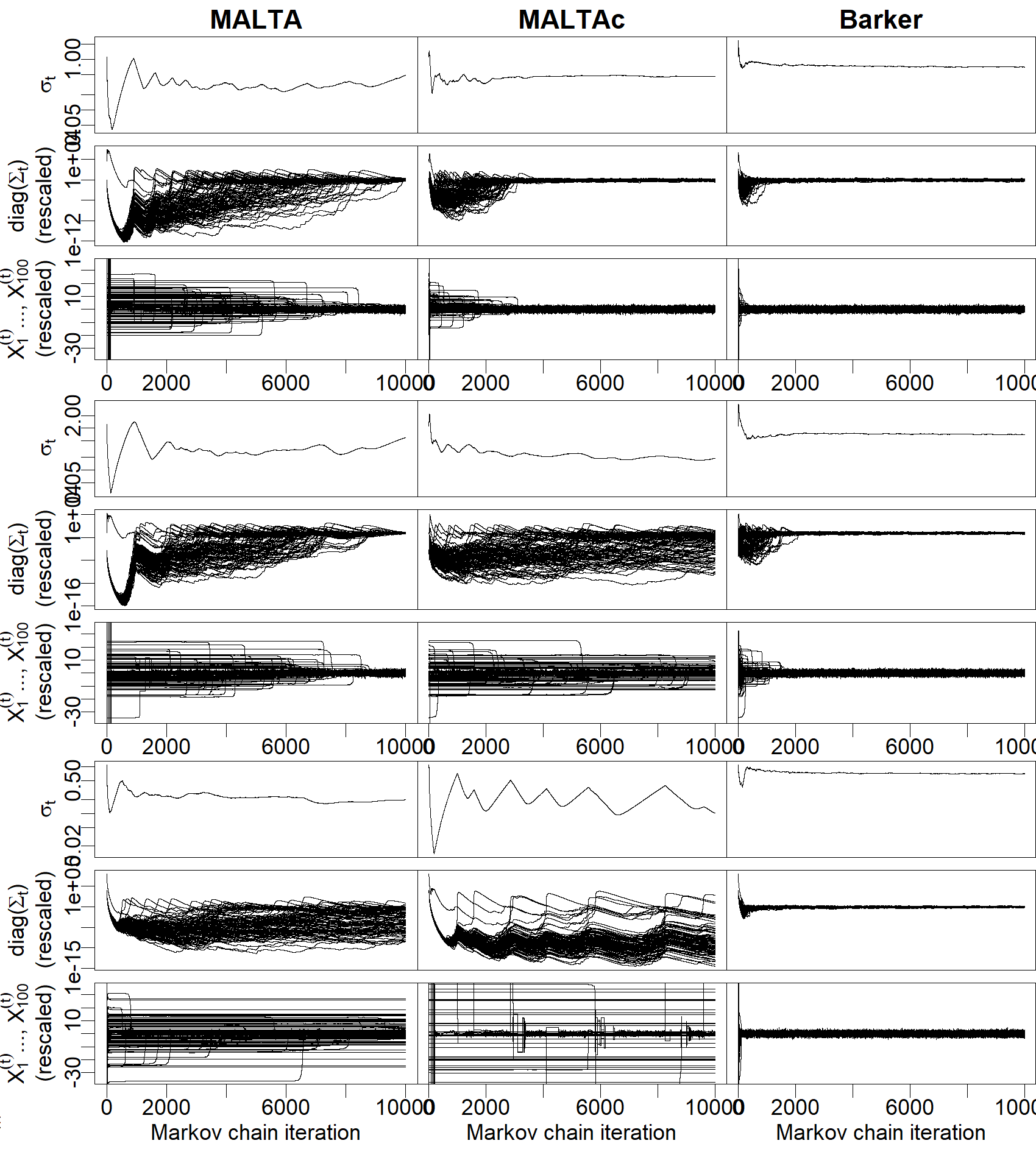}
\caption{Comparison of MALTA, MALTAc and Barker on target distributions with one small component.
Rows 1-3: same target considered in Figure \ref{fig:k=1_learn06} (100-dimensional Gaussian with first component standard deviation equal to $0.01$ and all others standard deviations equal to 1).
Rows 4-6 and rows 7-9: re-scaled versions where the scales of all coordinates are either multiplied or divided by 100.
See Figure \ref{fig:scenarios_234} for an explanation of each parameter plotted.
}
\label{fig:MALTA_k=1_scales}
\end{figure}
We illustrate this point in Figure \ref{fig:MALTA_k=1_scales}. 
There we compare MALTA, MALTAc and Barker on targets being 100-dimensional Gaussian distributions with one component much smaller than the others, analogously to the first scenario of Section \ref{sec:diag_adapt}. 
For MALTA we set $\delta=1000$, as is done for example, in \citet{atchade2006adaptive}. 
We also tried setting $\delta=100$ without observing major differences.
Rows 1-3 of Figure \ref{fig:MALTA_k=1_scales} consider exactly the same target of Figure \ref{fig:k=1_learn06}, which is a 100-dimensional Gaussian where the first component standard deviation is equal to $0.01$ and all others standard deviations are equal to 1.
In this case both MALTA and MALTAc improve over MALA, and in particular that MALTAc manages to converge to stationarity in around 4000 iterations, although this is still significantly slower than Barker (which only requires few hundred iterations).
We then consider modifying the target distribution by either multiplying the scales of all coordinates by 100, resulting in the first component standard deviation equal to $1$ and all others standard deviations equal to $100$, or dividing all scales by 100, resulting in the first component standard deviation equal to $10^{-4}$ and all others standard deviations equal to $10^{-2}$.
The results are reported in rows 4-6 and rows 7-9 of Figure \ref{fig:MALTA_k=1_scales}, respectively.
We observe a dramatic deterioration in the performances of both MALTA and MALTAc, while the performance of the Barker scheme are much less affected.
The underlying reason is that MALTA and MALTAc are highly sensitive to the choice of truncation parameter (respectively $\delta$ and $\sigma^{-2}$), which needs to be tuned appropriately depending on the scales of the target distribution.

These illustrative simulations suggest that ad-hoc strategies to improve the robustness of gradient-based MCMC, such as truncating or taming gradients, are intrinsically more fragile and sensitive to heterogeneity and scales compared to a more principled solution such as the Barker algorithm, in which robustness arises naturally from the proposal mechanism.
In addition to this, truncating and taming can be thought of as introducing a `bias' into the proposal mechanism, in the sense that the resulting proposal is no longer first-order exact.
Depending on how the truncation level $\delta$ is scaled with the dimensionality $d$, this can compromise the $d^{-1/3}$ scaling behaviour discussed in the paper.

\end{document}